\documentclass[11pt,pdfa]{article}
\usepackage[in]{fullpage}

\newif\ifdraft
\drafttrue

\usepackage[style=alphabetic,minalphanames=3,maxalphanames=4,maxnames=99,backref=true]{biblatex}

\usepackage[normalem]{ulem}
\usepackage{comment}
\usepackage[shortlabels]{enumitem}
\usepackage{bm}
\usepackage{amsfonts}
\usepackage{xspace}
\usepackage{amsmath}
\usepackage{amssymb}
\usepackage{amsthm}
\usepackage{xcolor}
\usepackage{graphicx}
\usepackage{qtree}
\usepackage{tree-dvips}
\usepackage{float}
\usepackage{hyperref}
\usepackage{breakurl}

\usepackage{braket}
\usepackage{mathtools}
\usepackage{mathrsfs}
\usepackage{tikz}
\usepackage{bm}
\usepackage[linesnumbered,ruled,vlined]{algorithm2e}
\usepackage{qcircuit}
\usepackage[nameinlink]{cleveref}

\newcommand{\subversion}[1]{}

  \hypersetup{colorlinks={true},linkcolor={blue},citecolor=magenta}

  \theoremstyle{plain}
  \newtheorem{theorem}{Theorem}[section]
  \newtheorem{lemma}[theorem]{Lemma}
  
  \newtheorem{definition}[theorem]{Definition}
  \newtheorem{remark}[theorem]{Remark}
  \newtheorem{claim}[theorem]{Claim}
 \newtheorem{construction}{Construction}
 \newtheorem*{theorem*}{Theorem}
 \newtheorem{conjecture}{Conjecture}

  \usepackage[style=alphabetic,minalphanames=3,maxalphanames=4,maxnames=99,backref=true]{biblatex}

\newcommand{\GenTrap}{\mathsf{GenTrap}}
\newcommand{\algo}{\mathcal}
\newcommand{\Zq}{\mathbb{Z}_q}
\newcommand{\mxA}{\mathbf{A}}
\newcommand{\bfA}{\mxA}

\newcommand{\td}{\mathsf{td}}

\newcommand{\bfs}{{\bf s}}
\newcommand{\bfe}{{\bf e}}

\newcommand{\bfu}{{\bf u}}

\renewcommand{\vec}[1]{\mathbf{#1}}  

\newcommand{\setup}{\mathsf{KeyGen}}

\newcommand{\Eval}{\mathsf{Eval}}

\newcommand{\SIS}{\mathsf{SIS}}
\newcommand{\ISIS}{\mathsf{ISIS}}
\newcommand{\LWE}{\mathsf{LWE}}
\newcommand{\enc}{\mathsf{Enc}}
\newcommand{\PKE}{\ensuremath{\mathsf{PKE}}\xspace}
\newcommand{\pqPRF}{\mathsf{pqPRF}}
\newcommand{\PRF}{\mathsf{PRF}}
\newcommand{\rPRF}{\mathsf{rPRF}}

\newcommand{\Revoke}{\mathsf{Revoke}}
\DeclareMathAlphabet\mathbfcal{OMS}{cmsy}{b}{n}
\newcommand{\FHE}{\ensuremath{\mathsf{FHE}}\xspace}
\newcommand{\proj}[1]{\ensuremath{|#1\rangle \langle #1|}}

\newcommand{\setS}{{\cal S}}

\newcommand{\dist}{{\sf Dist}} 
\newcommand{\lwe}{{\sf lwe}}
\newcommand{\unif}{{\sf unif}}
\newcommand{\simult}{{\sf simult}}

\newcommand{\lemkg}{{\cal K}{\cal G}}
\newcommand{\lemeval}{{\cal E}}
\newcommand{\zerof}{{\cal Z}}
\newcommand{\lemkey}{\kappa}

\newcommand{\bfS}{{\bf S}} 
\newcommand{\bfE}{{\bf E}}
\newcommand{\bindecomp}{{\sf bindecomp}}

\newcommand{\keygen}{\mathsf{KeyGen}}
\newcommand{\KeyGen}{\mathsf{KeyGen}}
\newcommand{\PPT}{\mathsf{PPT}}

\newcommand{\Invert}{\mathsf{Invert}}
\newcommand{\QPT}{\mathsf{QPT}}
\newcommand{\CPTP}{\mathsf{CPTP}}
\newcommand{\aux}{\mathsf{aux}}
\newcommand{\FT}{\mathsf{FT}}
\newcommand{\dec}{\mathsf{Dec}}
\newcommand{\pk}{\mathsf{PK}}
\newcommand{\Mod}[1]{\ (\mathrm{mod}\ #1)}
\newcommand{\sk}{\rho_{\mathsf{SK}}}
\newcommand{\msk}{\mathsf{MSK}}
\newcommand{\rand}{\raisebox{-1pt}{\ensuremath{\,\xleftarrow{\raisebox{-1pt}{$\scriptscriptstyle\$$}}\,}}}

\newcommand{\gen}{{\sf Gen}}

\newcommand{\prf}{{\sf PRF}}
\newcommand{\bfH}{{\bf H}}
\newcommand{\vectorize}{{\sf vec}}

\newcommand{\bfM}{{\sf M}}
\newcommand{\bfD}{{\bf D}}

\newcommand{\ct}{\mathsf{CT}}
\newcommand{\expt}{\mathsf{Expt}}

\newcommand{\delete}{\mathsf{Revoke}}
\newcommand{\revoke}{\mathsf{Revoke}}

\newcommand{\valid}{\mathsf{Valid}}
\newcommand{\invalid}{\mathsf{Invalid}}

\newcommand{\negl}{\mathsf{negl}}

\newcommand{\bfv}{\mathbf{v}}
\newcommand{\bfR}{\mathbf{R}}
\newcommand{\bfB}{\mathbf{B}}
\newcommand{\sampleleft}{\mathsf{SampleLeft}}
\newcommand{\sampleright}{\mathsf{SampleRight}}

\newcommand{\Z}{\mathbb{Z}}
\newcommand{\bfT}{\mathbf{T}}
\newcommand{\bfx}{\mathbf{x}}
\newcommand{\bft}{\mathbf{t}}

\newcommand{\distr}{\mathcal{D}}
\newcommand{\prob}{{\sf Pr}}
\newcommand{\evalct}{{\sf Eval}_{\sf ct}}
\newcommand{\evalpk}{{\sf Eval}_{\sf pk}}
\newcommand{\bfy}{\mathbf{y}}
\newcommand{\eval}{{\sf Eval}}
\newcommand{\ineffrevoke}{{\sf IneffRevoke}}

\ifdraft
\newcommand{\alex}[1]{{\color{blue} Alex: #1 }}
\newcommand{\pnote}[1]{{\color{blue} Prab: #1}}
\newcommand{\vnote}[1]{{\color{red} V: #1}}
\newcommand{\rc}[1]{{\color{red} #1}}
\else
\newcommand{\alex}[1]{}
\newcommand{\pnote}[1]{}
\newcommand{\vnote}[1]{}
\newcommand{\rc}[1]{{\color{red} #1}}
\fi 

\title{Revocable Cryptography from Learning with Errors \vspace{3.5mm}}

\author{Prabhanjan Ananth\footnote{\texttt{prabhanjan@cs.ucsb.edu}}\\UCSB \and 
Alexander Poremba\footnote{\texttt{aporemba@caltech.edu}}\\Caltech \and 
Vinod Vaikuntanathan\footnote{\texttt{vinodv@mit.edu}}\\MIT
}
\date{}

\newcommand{\poly}{\mathrm{poly}}

\newcommand{\hybrid}{\mathsf{H}}
\newcommand{\adversary}{\mathcal{A}}

\newcommand{\ketbra}[2]{\ket{#1}\!\bra{#2}}
\renewcommand{\cal}[1]{\mathcal{#1}}

\newcommand{\N}{\mathbb{N}}

\newcommand{\Tr}{\mathrm{Tr}}

\newcommand{\reg}[1]{\textsc{#1}}

\newcommand{\eps}{\varepsilon}

\newcommand{\Delete}{\mathsf{Delete}}
\newcommand{\sati}{{\sf ATI}}
\newcommand{\cP}{{\cal P}}
\newcommand{\Hs}{{\cal H}}

\newcommand{\bracketsC}[1]{\left\{ #1\right\}}

\newcommand{\Enc}{\mathsf{Enc}}
\newcommand{\Dec}{\mathsf{Dec}}

\DeclareMathOperator{\TD}{TD}

\newcommand{\bit}{{\{0, 1\}}}

\newcommand{\secparam}{\lambda}

\DeclareFontFamily{U}{skulls}{}
\DeclareFontShape{U}{skulls}{m}{n}{ <-> skull }{}

\begin{document}
\maketitle

\begin{abstract}
\noindent Quantum cryptography leverages many unique features of quantum information in order to construct cryptographic primitives that are oftentimes impossible classically. In this work, we build on the no-cloning principle of quantum mechanics and design cryptographic schemes with \emph{key-revocation capabilities}. We consider schemes where secret keys are represented as quantum states with the guarantee that, once the secret key is successfully revoked from a user, they no longer have the ability to perform the same functionality as before.

We define and construct several fundamental cryptographic primitives with {\em key-revocation capabilities}, namely pseudorandom functions, secret-key and public-key encryption, and even fully homomorphic encryption. Our constructions either assume the quantum subexponential hardness of the learning with errors problem or are based on new conjectures. Central to all our constructions is our approach for making the Dual-Regev encryption scheme (Gentry, Peikert and Vaikuntanathan, STOC 2008) revocable.
\end{abstract}

\newpage 
\tableofcontents
\newpage

\section{Introduction}

\noindent Quantum computing presents exciting new opportunities for cryptography, using remarkable properties of quantum information to construct cryptographic primitives that are unattainable classically. At the heart of quantum cryptography lies the \emph{no-cloning principle}~\cite{WZ82,Dieks82} of quantum information which stipulates that it is fundamentally impossible to copy an unknown quantum state.
Indeed, Wiesner~\cite{Wiesner83} in his seminal work from the 1970s used the no-cloning principle to construct a quantum money scheme, wherein quantum states are used to construct banknotes that can be verified to be authentic (using a secret key) but cannot be counterfeited. 
Ever since this watershed moment, and especially so in the recent years, a wide variety of primitives referred to as {\em unclonable} primitives have been studied and constructed in the context of encryption~\cite{Got02,BL20,BI20,GZ20}, digital signatures~\cite{LLLZ22} and pseudorandom functions~\cite{CLLZ21}.

\paragraph{Our Work: Revocable Cryptography.} 
Delegation and revocation of privilege are problems of great importance in cryptography. Indeed, the problem of revocation in the context of digital signatures and certificates in the classical world is a thorny problem~\cite{stubble,Riv98b}. In this work, we undertake a systematic study of {\em revocable (quantum) cryptography} which allows us to delegate and revoke privileges in the context of several fundamental cryptographic primitives. This continues a recent line of work in quantum cryptography dealing with revoking (or certifiably deleting) states such as quantum ciphertexts or simple quantum programs~\cite{Unruh2013,BI20,GZ20,ALP21,hiroka2021certified,Poremba22,BK22}.
In this framework, revocation is to be understood from the perspective of the recipient of the quantum state; namely, the recipient can certify the loss of certain privileges by producing a certificate (either classical or quantum) which can be verified by another party.

As a motivating example, consider the setting of an employee at a company who takes a vacation and wishes to authorize a colleague to perform certain tasks on her behalf, tasks that involve handling sensitive data. Since the sensitive data is (required to be) encrypted, the employee must necessarily share her decryption keys with her colleague. When she returns from vacation, she would like to have her decryption key back; naturally, one would like to ensure that her colleague should not be able to decrypt future ciphertexts (which are encrypted under the same public key) once the key is ``returned''. Evidently, if the decryption key is a classical object, this is impossible to achieve as classical keys can be copied at will.
     
In revocable (quantum) cryptography, we associate a cryptographic functionality, such as decryption using a secret key, with a quantum state in such a way that a user can compute this functionality if and only if they are in possession of the quantum state. We then design a revocation algorithm which enables the user to certifiably return the quantum state to the owner. Security requires that once the user returns the state (via our revocation algorithm), they should not have the ability to evaluate the functionality (e.g. decrypt ciphertexts) anymore. We refer to this new security notion as {\em revocation security}. 

Another, possibly non-obvious, application is to detecting malware attacks. Consider a malicious party who hacks into an electronic device and manages to steal a user's decryption keys. If cryptographic keys are  represented by classical bits, it is inherently challenging to detect such attacks that compromise user keys. For all we know, the intruder could have stolen the user's decryption keys without leaving a trace. Indeed, a few years ago, decryption keys which were used to protect cell-phone communications~\cite{SimHeist} were successfully stolen by spies without being detected.\footnote{The attack would indeed have gone undetected but for the Snowden revelations.} With revocable cryptography, a malicious user successfully stealing a user key would invariably revoke the decryption capability from the user. This latter event can be detected. 

     
\paragraph{Our Results in a Nutshell.} We construct revocable cryptographic objects under standard cryptographic assumptions. Our first main result constructs a key-revocable public-key encryption scheme, and our second main result constructs a key-revocable pseudorandom function. We obtain several corollaries and extensions, including key-revocable secret-key encryption and key-revocable fully homomorphic encryption. In all these primitives, secret keys are represented as quantum states that retain the functionality of the original secret keys. We design revocation procedures and guarantee that once a user successfully passes the procedure, they cannot compute the functionality any more.  

All our constructions are secure under the quantum subexponential hardness of learning with errors~\cite{Regev05}---provided that revocation succeeds with high probability. At the heart of all of our contributions lies our result which shows that the Dual-Regev public-key encryption scheme of \cite{cryptoeprint:2007:432} satisfies revocation security.

\paragraph{Related Notions.} 
There are several recent notions in quantum cryptography that are related to revocability. Of particular relevance is the stronger notion of copy-protection introduced by Aaronson~\cite{Aar09}. Breaking the revocable security of a task gives the adversary a way to make two copies of a (possibly different) state both of which are capable of computing the same functionality. Thus, copy-protection is a stronger notion. However, the only known constructions of copy-protection schemes~\cite{CLLZ21,LLLZ22} rely on the heavy hammer of {\em post-quantum secure} indistinguishability obfuscation for which there are no known constructions based on well-studied assumptions. Our constructions, in contrast, rely on the post-quantum hardness of the standard learning with errors problem. A different related notion is the significantly weaker definition of secure software leasing~\cite{ALP21} which guarantees that once the quantum state computing a functionality is returned, the {\em honest evaluation algorithm} cannot compute the original functionality. 
Yet another orthogonal notion is that of certifiably deleting {\em ciphertexts}, originating from the works of Unruh~\cite{Unruh2013} and Broadbent and Islam~\cite{BI20}.
In contrast, our goal is to delegate and revoke {\em cryptographic capabilities} enabled by private keys. 
For detailed comparisons, we refer the reader to Section~\ref{sec:related}.

\subsection{Our Contributions in More Detail}
We present our results in more detail below. First, we introduce the notion of key-revocable public-key encryption. Our main result is that the Dual-Regev public-key encryption scheme~\cite{cryptoeprint:2007:432} satisfies revocation security. After that, we study revocation security in the context of fully homomorphic encryption and pseudorandom functions. 

\paragraph{Key-Revocable Public-Key Encryption.} We consider public-key encryption schemes where the decryption key, modeled as
a quantum state, can be delegated to a third party and can later be revoked~\cite{GZ20}. The syntax of a key-revocable public-key scheme (\Cref{def:key-revocable-PKE}) is as follows:

\begin{itemize}
    \item $\keygen(1^\lambda)$: this is a setup procedure which outputs a public key $\pk$, a master secret key $\msk$ and a decryption key $\sk$. While the master secret key is typically a classical string, the decryption key is modeled as a quantum state. 
    (The use cases of $\msk$ and $\sk$ are different, as will be evident below.)
    \item $\enc(\pk,x)$: this is the regular classical encryption algorithm which outputs a ciphertext $\ct$. 
    
    \item $\dec(\sk,\ct)$: this is a quantum algorithm which takes as input the quantum decryption key $\sk$ and a classical ciphertext, and produces a plaintext.
    
    \item $\delete(\pk,\msk,\sigma)$: this is the revocation procedure that outputs $\valid$ or $\invalid$. If $\sigma$ equals the decryption key $\sk$, then $\delete$ outputs $\valid$ with high probability.
\end{itemize}

\noindent After the decryption key is returned,
we require that the sender loses its ability to decrypt ciphertexts. This is formalized as follows (see \Cref{def:krpke:security}): conditioned on revocation being successful, no adversary can distinguish whether it is given an encryption of a message versus the uniform distribution over the ciphertext space with advantage better than $\negl(\lambda)$. Moreover, we require that revocation succeeds with a probability negligibly close to 1 (more on this later).

\noindent We prove the following in \Cref{thm:security-Dual-Regev}.

\begin{theorem*}[Informal]
\label{thm:intro}
Assuming that the $\LWE$ and $\SIS$ problems with subexponential modulus are hard against quantum adversaries running in subexponential time (see \Cref{sec:lattices}), there exists a key-revocable public-key encryption scheme. 
\end{theorem*}

\noindent Due to the quantum reduction from $\SIS$ to $\LWE$~\cite{cryptoeprint:2009/285}, the two assumptions are, in some sense, equivalent. Therefore, we can in principle rely on the subexponential hardness of $\LWE$ alone.
\par Our work improves upon prior works, which either use post-quantum secure indistinguishability obfuscation~\cite{GZ20,CLLZ21} or consider the weaker private-key setting~\cite{KN22}. 

\paragraph{Key-Revocable Public-Key Encryption with Classical Revocation.}

Our previous notion of key-revocable encryption schemes models the key as 
a quantum state which can later be ``returned.'' One may therefore reasonably ask whether it is possible to achieve a classical notion of revocation, similar to the idea of deletion certificates in the context of quantum encryption~\cite{BI20,Hiroka_2021,BK22,Poremba22,hiroka2021certified}. Rather then return a quantum decryption key, the recipient would simply apply an appropriate measurement (as specified by a procedure $\Delete$), and output a classical certificate which can be verified using $\revoke$. We also consider key-revocable public-key encryption scheme with classical revocation (see \Cref{def:classical-key-revocable-PKE}) which has the same syntax as a regular key-revocable public-key scheme, except that revocation process occurs via the following two procedures:
\begin{itemize}
    \item $\Delete(\sk)$: this takes as input a quantum decryption key $\sk$, and produces a classical revocation certificate $\pi$.
    \item $\delete(\pk,\msk,\pi)$: this takes as input the (classical) master secret key $\msk$ and a (classical) certificate $\pi$, and outputs $\valid$ or $\invalid$. If $\pi$ is the output of $\Delete(\sk)$, then $\delete$ outputs $\valid$ with high probability.
\end{itemize}
Our notion of security (see \Cref{def:krpke:security-classical}) is essentially the same as for a regular key-revocable public-key encryption scheme where revocation is quantum. We prove the following in \Cref{thm:dual-regev-classical}.

\begin{theorem*}[Informal]
\label{thm:intro}
Assuming that the $\LWE$ and $\SIS$ problems with subexponential modulus are hard against quantum adversaries running in subexponential time (see \Cref{sec:lattices}), there exists a key-revocable public-key encryption scheme with classical revocation. 
\end{theorem*}

The assumptions required for the theorem are essentially the same as for our previous Dual-Regev scheme.

\paragraph{Key-Revocable Fully Homomorphic Encryption.}

 We go beyond the traditional public-key setting and design the first {\em fully homomorphic encryption} (\textsf{FHE}) scheme~\cite{Gentry09,BV14} with key-revocation capabilities. Our construction is based on a variant of the (leveled) \textsf{FHE} scheme of Gentry, Sahai and Waters~\cite{GSW2013}, which we extend to a key-revocable encryption scheme using Gaussian superpositions.
 The syntax of a key-revocable \textsf{FHE} scheme is the same as in the key-revocable public-key setting from before (\Cref{def:key-revocable-PKE}), except for the additional algorithm \textsf{Eval} which is the same as in a regular \textsf{FHE} scheme.
 We prove the following in \Cref{thm:security-Dual-GSW}. 

\begin{theorem*}[Informal]
\label{thm:intro}
Assuming that the $\LWE$ and $\SIS$ problems with subexponential modulus are hard against quantum adversaries running in subexponential time (see \Cref{sec:lattices}), there exists a key-revocable (leveled) fully homomorphic encryption scheme. 
\end{theorem*}

We prove the theorem by invoking the security of our key-revocable Dual-Regev public-key encryption scheme in \Cref{sec:dual-regev}. Similar to the case of revocable $\PKE$ with classical revocation, we can also consider revocable $\FHE$ with classical revocation, which can be achieved based on the same assumptions as above. 

\paragraph{(Key-)Revocable Pseudorandom Functions.} We consider other cryptographic primitives with key-revocation capabilities that go beyond decryption functionalities; specifically, we introduce the notion of
\emph{key-revocable} pseudorandom functions (\textsf{PRFs}) with the following syntax:
\begin{itemize}
    \item $\mathsf{Gen}(1^\lambda)$: outputs a \textsf{PRF} key $k$, a quantum key $\rho_k$ and a master secret key $\msk$. 
    \item $\prf(k;x)$: on key $k$ and input $x$, output a value $y$. This is a deterministic algorithm. 
    \item $\eval(\rho_k,x)$: on input a state $\rho_k$ and an input $x$, output a value $y$. 
    \item $\revoke(\msk,\sigma)$: on input a verification key $\msk$ and state $\sigma$, outputs $\valid$ or $\invalid$. 
\end{itemize}
After the quantum key $\rho_k$ is successfully returned,
we require that the sender loses its ability to evaluate the $\prf$. This is formalized as follows (see \Cref{def:revocable-PRF}): no efficient adversary can simultaneously pass the revocation phase and succeed in distinguishing the output of a pseudorandom function on a random challenge input $x^*$ versus uniform with advantage better than $\negl(\lambda)$. In fact, we consider a more general definition where the adversary receives many challenge inputs instead of just one. 

We give the first construction of key-revocable pseudorandom functions (\textsf{PRFs}) from standard assumptions.
Previous schemes implicit in~\cite{CLLZ21} either require indistinguishability obfuscation, or considered weaker notions of revocable \textsf{PRFs} in the form of \emph{secure software leasing}~\cite{ALP21,Kitagawa_2021}, which merely prevents the possiblity of \emph{honestly} evaluating the \textsf{PRF} once the key is revoked.
\par Since in the context of pseudorandom functions, it is clear what is being revoked, we instead simply call the notion revocable pseudorandom functions.

\begin{theorem*}[Informal]
\label{thm:intro:prfs}
Assuming that the $\LWE$ and $\SIS$ problems with subexponential modulus are hard against quantum adversaries running in subexponential time (see \Cref{sec:lattices}), there exist key-revocable pseudorandom functions. 
\end{theorem*}

\noindent Revocable pseudorandom functions immediately give us key-revocable (many-time secure) secret-key encryption schemes. We also revocable $\PRF$s with classical revocation can be achieved based on the same assumptions as above. 

\newcommand{\drconj}{dualRegevExtraction}
\newcommand{\sdre}{$\sf{SDRE}$}
\newcommand{\qsdre}{${\sf qSDRE}$}
\newcommand{\csdre}{${\sf cSDRE}$}
\paragraph{Inverse-polynomial revocation based on conjectures.} In all the results above, we assume that the probability of revocation is negligibly close to 1. Even in this restrictive setting, our proofs turn out to be highly non-trivial and require careful use of a diverse set of techniques! Moreover, to date, no constructions of key-revocable $\PRF$s or $\FHE$ were known based on assumptions weaker than post-quantum iO. 
\par A natural question to explore is whether we can achieve the following stronger security notion of revocable public-key encryption: if the adversary successfully revokes the decryption key with inverse-polynomial probability, then semantic security of revocable $\PKE$ still holds. If we achieved this stronger notion of revocable $\PKE$, then we would also achieve the corresponding stronger notions of revocable $\PRF$s and $\FHE$ based on the same computational assumptions. 
\par We show how to achieve these stronger results based on a plausible conjecture (see \Cref{conjecture-I}) which we can phrase in the language of \emph{asymmetric} cloning games~\cite{AKL23}. 
Informally, our conjecture states the following: suppose that our Dual-Regev $\PKE$ scheme is not key-revocable (according to the above stronger definition). Then, there exists a triplet of algorithms $(\algo A,\algo B,\algo C)$, where
\begin{itemize}
    \item $\adversary$ is a QPT algorithm which receives as input $(\bfA,\bfy,\ket{\psi_{\bfy}})$, where $\bfA \xleftarrow{\$} \Zq^{n \times m}$ and $\ket{\psi_{\bfy}}$ is a Gaussian superposition of all the short vectors mapping $\bfA$ to $\bfy$. It produces a bipartite state on two registers $\textsc{B}$ and $\textsc{C}$. 

    \item $\algo B$ is a fixed and inefficient revocation algorithm which receives as input $(\vec A,\vec y,\reg B)$ and applies the (inefficient) projective measurement 
   $\big\{\proj{\psi_{\vec y}}, I - \proj{\psi_{\vec y}}\big\}$ to register $\textsc{B}$. 
    \item $\algo C$ is a QPT algorithm which on input $(\vec A,\vec y,\reg C)$ distinguishes Dual-Regev ciphertexts (with respect to $\vec A$ and $\vec y$) from random\footnote{Strictly speaking, the challenge input here consists of a vector which resembles a Dual-Regev ciphertext but which implicitly also depends on a fixed Gaussian vector $\vec x_0$ such that $\vec A \vec x_0 = \vec y \Mod{q}$. More on this in \Cref{sec:overview}.} with inverse-polynomial advantage---conditioned on the event that the revocation algorithm $\algo B$ succeeds on register $\reg B$.
\end{itemize}
The difficulty in proving the conjecture lies in the fact that one needs to invoke the $\LWE$ assumption with respect to $\algo C$ who holds $\reg{C}$, while at the same time guaranteeing that an inefficient projective measurement succeeds on a separate register $\reg B$.
We leave proving (or refuting) the above conjecture to future works. 
\par We also consider another variant of the above conjecture (see \Cref{conjecture-II}) in the context of classical revocation. Here, $\adversary$ is instead given $\ket{\psi_{\bfy}^{\bfv}}$, which is essentially $\ket{\psi_{\bfy}}$, except that a phase of $\omega_q^{\langle \bfx, \lfloor \frac{q}{\nu} \rfloor \cdot \bfv \rangle}$ is planted for the term $\ket{\bfx}$ in the superposition. Moreover, we replace $\algo B$ with a fixed classical revocation algorithm which can detect $\vec v$ within ciphertexts $\ct \approx \bfs^\intercal \bfA + \lfloor \frac{q}{\nu} \rfloor \cdot \bfv^\intercal $, for some $\bfs \in \Zq^{n}$, given the register $\textsc{B}$ and an appropriate \emph{trapdoor} for $\vec A$.

\paragraph{Discussion: Unclonable Cryptography from \textsf{LWE}.}\ Over the years, the existence of many fundamental cryptographic primitives such as pseudorandom functions~\cite{BPR12}, fully homomorphic encryption~\cite{BV11}, attribute-based encryption~\cite{cryptoeprint:2014/356} and succinct argument systems~\cite{CJJ22} have been based on the existence of learning with errors. In fact, as far as we know, there are only a few foundational primitives remaining (indistinguishability obfuscation is one such example) whose existence is not (yet) known to be based on learning with errors. 
\par The situation is quite different in the world of unclonable cryptography. Most of the prominent results have information-theoretic guarantees but restricted functionalities~\cite{BI20,BL20} or are based on the existence of post-quantum indistinguishability obfuscation~\cite{10.1007/s00145-020-09372-x,CLLZ21}. While there are works~\cite{KNY21} that do propose lattice-based constructions of unclonable primitives, there are still many primitives, such as quantum money and quantum copy-protection, whose feasibility we would like to establish based on the post-quantum hardness of the learning with errors (or other such relatively well-studied assumptions). We hope that our work presents a new toolkit towards building more unclonable primitives from $\LWE$. 

\paragraph{Independent and Concurrent Work.} Independently and concurrently\footnote{Both of our works were posted online around the same time.}, Agrawal et al.~\cite{AKNYY23} explored the notion of public-key encryption with secure leasing which is related to key-revocable public-key encryption. We highlight the differences below:
\begin{itemize}
    \item {\em Advanced notions}: We obtain key-revocable
\emph{fully homomorphic encryption} and key-revocable \emph{pseudorandom functions} which are unique to our work. They explore other notions of advanced encryption with secure leasing including attribute-based encryption and functional encryption, which are not explored in our work.  
    \item {\em Security definition}: We consider the stronger notion of classical revocation\footnote{The stronger notion was updated in our paper subsequent to posting of our and their work.} whereas they consider quantum revocation. 
    \item {\em Public-key encryption}: They achieve a generic construction based on any post-quantum secure public-key encryption\footnote{Their construction achieves the stronger definition where the revocation only needs to succeed with inverse polynomial probability.} whereas our notion is based on the post-quantum hardness of the learning with errors problem, or on newly introduced conjectures (when revocation does not succeed with high probability). Their construction of revocable public-key encryption involves many complex abstractions whereas our construction is based on the versatile Dual-Regev public-key encryption scheme.

\end{itemize}

\subsection{Overview}\label{sec:overview}

We now give a technical overview of our constructions and their high level proof ideas. 
We begin with the key-revocable public-key encryption construction. A natural idea would be to start with Regev's public-key encryption scheme~\cite{Regev05} and to then upgrade the construction in order to make it revocable. However, natural attempts to associate an unclonable quantum state with the decryption key fail and thus, we instead consider the Dual-Regev public-key encryption scheme and make it key-revocable. We describe the scheme below.

\paragraph{Key-Revocable Dual-Regev Public-Key Encryption.} 

Our first construction is based on the \emph{Dual-Regev} public-key encryption scheme~\cite{cryptoeprint:2007:432} and makes use of Gaussian superpositions which serve as a quantum decryption key.
We give an overview of \Cref{cons:dual-regev} below. 
\begin{itemize}
    \item $\keygen(1^n)$: sample a matrix $\vec A \in \Zq^{n \times m}$ along with a \emph{short trapdoor basis} $\td_{\vec A}$. 
    To generate the decryption key, we employ the following procedure\footnote{In \Cref{sec:GenGauss}, this is formalized as the procedure ${\sf GenGauss}$ (see Algorithm \ref{alg:GenGauss}).}: Using the matrix $\vec A$ as input, first create a Gaussian superposition of short vectors in $\Z^m \cap (-\frac{q}{2},\frac{q}{2}]^m$, denoted by\footnote{Note that the state is not normalized for convenience.} 
    $$ \ket{\psi} = \sum_{\vec x \in \Z_q^m} \rho_\sigma(\vec x) \ket{\vec x} \otimes \ket{\vec A \cdot \vec x \Mod{q}}
    $$
    where $\rho_\sigma(\vec x)=\exp(-\pi \|\vec x\|^2/\sigma^2)$ is the Gaussian measure, for some $\sigma > 0$. Next, measure the second register which partially collapses the superposition and results in the \emph{coset state} $$ \ket{\psi_{\vec y}} = \sum_{\substack{\vec x \in \Z_q^m:\\ \vec A \vec x= \vec y \Mod{q}}} \rho_\sigma(\vec x) \ket{\vec x}
    $$
    for some outcome $\vec y \in \Z_q^n$.
    Finally, we let $\ket{\psi_{\bf y}}$ be the decryption key $\sk$, $({\vec A},{\vec y})$ be the public key $\pk$, and we let the trapdoor $\td_{\vec A}$ serve as the master secret key $\msk$.

    \item $\enc(\pk,\mu)$ is identical to Dual-Regev encryption. To encrypt a bit $\mu \in \bit$, sample a random string ${\vec s} \xleftarrow{\$} \Zq^{n}$ together with discrete Gaussian errors ${\vec e} \in \Z^{m}$ and $e' \in \Z$, and output the (classical) ciphertext $\ct$ given by
    $$
    \ct = (\vec s^\intercal \vec A + \vec e^\intercal, \vec s^\intercal \vec y + e' + \mu \cdot \lfloor \frac{q}{2} \rfloor) \,\,\, \in \Z_q^m \times \Z_q.
    $$
    \item $\dec(\sk,\ct)$: to decrypt a ciphertext $\ct$ using the decryption key $\sk=\ket{\psi_{\vec y}}$, first apply the unitary $U: \ket{\vec x}\ket{0} \rightarrow \ket{\vec x}\ket{\ct \cdot (-\vec x,1)^\intercal}$ on input $\ket{\psi_{\vec y}}\ket{0}$, and then measure the second register in the computational basis. Because $\ket{\psi_{\vec y}}$ is a superposition of short vectors ${\vec x}$ subject to ${\vec A} \cdot {\vec x} = {\vec y} \Mod{q}$, we obtain an approximation of $\mu \cdot \lfloor \frac{q}{2} \rfloor$ from which we can recover $\mu$.\footnote{For approriate choices of parameters, decryption via rounding succeeds at outputting $\mu$ with overwhelming probability and hence we can invoke the ``almost as good as new'' lemma~\cite{aaronson2016complexity} to recover the original state $\ket{\psi_{\vec y}}$.}  
    \item $\revoke(\pk,\msk,\rho)$: to verify the returned state $\rho$ given as input the public key $(\vec A,\vec y)$ and master secret key $\td_{{\vec A}}$, apply the projective measurement $\{\proj{\psi_{\vec y}}, I - \proj{\psi_{\vec y}}\}$ onto $\rho$. Output $\valid$, if the measurement succeeds, and output $\invalid$, otherwise. 
\end{itemize}

\paragraph{Implementing revocation, efficiently.}  Note that performing a projective measurement onto a fixed Gaussian state $\ket{\psi_{\vec y}}$ is, in general, computationally infeasible. In fact, if it were to be possible to efficiently perform this projection using $(\bfA,\bfy)$ alone, then one could easily use such a procedure to solve the short integer solution ($\SIS$) problem. Fortunately, we additionally have the trapdoor for $\bfA$ at our disposal in order to perform such a projection. 
\par One of our contributions is to design a \emph{quantum discrete Gaussian sampler for $q$-ary lattices}\footnote{In \Cref{sec:QSampGauss}, this is formalized as the procedure {\sf QSampGauss} (see Algorithm \ref{alg:SampGauss}).} which, given as input $({\vec A},\vec y,\td_{{\vec A}},\sigma)$, implements a unitary that efficiently prepares the Gaussian superposition $\ket{\psi_{\vec y}}$ from scratch with access to the trapdoor $\td_{{\vec A}}$. At a high level, our Gaussian sampler can be alternately thought of as an explicit quantum reduction from the \emph{inhomogenous} $\SIS$ problem~\cite{DBLP:conf/stoc/Ajtai96} to the search variant of the $\LWE$ problem (see \Cref{sec:QSampGauss}).

\paragraph{Insight: Reduction to SIS.} Our goal is to use the state returned by the adversary and to leverage the indistinguishability guarantee in order to break some computational problem. It should seem suspicious whether such a reduction is even possible: after all the adversary is returning the state we gave them! {\em How could this possibly help?} Our main insight lies in the following observation: while the adversary does eventually return the state we give them, the only way it can later succeed in breaking the semantic security of dual Regev PKE is if it retains useful information about the state. If we could somehow extract this information from the adversary, then using the extracted information alongside the returned state, we could hope to break some computational assumption. For instance, suppose we can extract a short vector $\bfx$ such that $\bfA \cdot \bfx = \bfy \Mod{q}$. By measuring the state returned by the adversary, we could then hope to get a second short vector $\bfx'$ such that $\bfA \cdot \bfx' = \bfy \Mod{q}$, and from this, we can recover a short solution in the kernel of $\bfA \in \Z_q^{n \times m}$. 
\par Even if, for a moment, we disregard the issue of being able to extract $\bfx$ from the adversary, there are still some important missing steps in the above proof template: 
\begin{itemize}
    \item Firstly, measuring the returned state should give a vector different from $\bfx$ with non-negligible probability. In order to prove this, we need to argue that the squared ampltidue of every term is bounded away from 1. We prove this statement (\Cref{lem:max:ampl:bound}) holds as long as $\bfA$ is full rank. 
    \item Secondly, the reduction to $\SIS$ would only get as input $\bfA$ and not a trapdoor for $\bfA$. This means that it will no longer be possible for the reduction to actually check whether the state returned by the adversary is valid. We observe that, instead of first verifying whether the returned state is valid and then measuring in the computational basis, we can in fact skip verification and immediately go ahead and measure the state in the computational basis; this is implicit in the analysis in the proof of~\Cref{clm:distinct:preimageext}. 
    \item Finally, the adversary could have entangled the returned state with its residual state in such a way that measuring the returned state always yields the same vector $\bfx$ as the one extracted from the adversary. In the same analysis in the proof of~\Cref{clm:distinct:preimageext}, we prove that, even if the adversary entangles its state with the returned state, with non-negligible probability we get two distinct short vectors mapping $\bfA$ to $\bfy$. 
\end{itemize}
\par All that is left is to argue that it is possible to extract $\bfx$ from the adversary while simultaneously verifying whether the returned state is correct or not. To show that we can indeed extract another short pre-image from the adversary's quantum side information, we make use of what we call a \emph{simultaneous search-to-decision reduction with quantum auxiliary input} for the Dual-Regev scheme.

\paragraph{Main contribution: Simultaneous search-to-decision reduction with quantum advice.}
 Informally, our theorem says the following:
any successful Dual-Regev distinguisher with access to quantum side information \textsc{Aux} (which depends on the decryption key)
can be converted into a successful extractor that finds a key on input \textsc{Aux} -- even conditioned on $\revoke$ 
succeeding on a seperate register $\reg R$. We now present some intuition behind our proof.

Suppose there exists a successful Dual-Regev distinguisher $\algo D$ (as part of the adversary $\algo A$) that, given quantum auxiliary information \textsc{Aux}, can distinguish between $(\vec s^\intercal \vec A + \vec e^\intercal, \vec s^\intercal \vec y + e')$ and uniform $(\vec u,r) \in \Z_q^m \times \Z_q$ with advantage $\epsilon$.\\

\noindent {\em Ignoring register $\reg R$}: For now, let us ignore the fact that $\revoke$ is simultaneously applied on system $R$. Inspired by techniques from the \emph{leakage resilience} literature~\cite{Dodis10}, we now make the following observation. Letting $\vec y =\vec A \cdot \vec x_0 \Mod{q}$, for some Gaussian vector $\vec x_0$ with distribution proportional to $\rho_\sigma(\vec x_0)$, the former sample can be written as $(\vec s^\intercal \vec A + \vec e^\intercal, (\vec s^\intercal \vec A  + \vec e^\intercal) \cdot \vec x_0 + e')$. Here, we assume a \emph{noise flooding} regime in which the noise magnitude of $e'$ is significantly larger than that of $\vec e^\intercal \cdot \vec x_0$. Because the distributions are statistically close, the distinguisher $\algo D$ must succeed at distinguishing the sample from uniform with probability negligibly close to $\epsilon$. Finally, we invoke the $\LWE$ assumption and claim that the same distinguishing advantage persists, even if we replace $(\vec s^\intercal \vec A + \vec e^\intercal)$ with a random string $\vec u \in \Z_q^m$. Here, we rely on the fact that the underlying $\LWE$ sample is, in some sense, independent of the auxiliary input \textsc{Aux} handed to the distinguisher $\algo D$. To show that this is the case, we need to argue that the reduction can generate the appropriate inputs to $\algo D$ on input $\vec A$; in particular it should be able to generate the auxiliary input \textsc{Aux} (which depends on a state $\ket{\psi_{\vec y}}$), while simultaneously producing a Gaussian vector $\vec x_0$ such that $\vec A \cdot \vec x_0 = \vec y \Mod{q}$. Note that this seems to violate the $\SIS$ assumption, since the ability to produce both a superposition $\ket{\psi_{\vec y}}$ of pre-images and a single pre-image $\vec x_0$ would allow one to obtain a collision for $\vec y$.\\

\noindent {\em Invoking Gaussian-collapsing}: To overcome this issue, we ask the reduction to generate the quantum auxiliary input in a different way; rather than computing \textsc{Aux} as a function of $\ket{\psi_{\vec y}}$, we compute it as a function of $\ket{{\vec x_0}}$, where ${\vec x_0}$ results from \emph{collapsing} the state $\ket{\psi_{\vec y}}$ via a measurement in the computational basis. By invoking the \emph{Gaussian collapsing property}~\cite{Poremba22}, we can show that the auxiliary information computed using $\ket{\psi_{\vec y}}$ is computationally indistinguishable from the auxiliary information computed using $\ket{{\vec x_0}}$. Once we invoke the collapsed version of $\ket{\psi_{\vec y}}$, we can carry out the reduction and conclude that $\algo D$ can distinguish between the samples $(\vec u,\vec u^\intercal \vec x_0)$ and $(\vec u,r)$, where $\vec u$ and $r$ are random and $\vec x_0$ is Gaussian, with advantage negligibly close to $\epsilon$.\footnote{Technically, $\algo D$ can distinguish between $(\vec u,\vec u^\intercal \vec x_0 + e')$ and $(\vec u,r)$ for a Gaussian error $e'$. However, by defining a distinguisher $\tilde{\algo D}$ that first shifts $\vec u$ by a Gaussian vector $e'$ and then runs $\algo D$, we obtain the desired distinguisher.}
Notice that $\algo D$ now resembles a so-called \emph{Goldreich-Levin} distinguisher~\cite{10.1145/73007.73010}.\\

\noindent {\em Reduction to Goldreich-Levin}: Assuming the existence of a quantum Goldreich-Levin theorem for the field $\Zq$, one could then convert $\algo D$ into an extractor that extracts ${\vec x_0}$ with high probability. Prior to our work, a quantum Goldreich-Levin theorem was only known for $\Z_2$~\cite{AdcockCleve2002,CLLZ21}. In particular, it is unclear how to extend prior work towards higher order fields $\Z_q$ because the interference pattern in the analysis of the quantum extractor does not seem to generalize beyond the case when $q=2$.
Fortunately, we can rely on the
\emph{classical} Goldreich-Levin theorem for finite fields due to Dodis et al.~\cite{Dodis10}, as well as recent work by Bitansky, Brakerski and Kalai.~\cite{BBK22} which shows that a large class of classical reductions can be generically converted into a quantum reductions. This allows us to obtain a quantum Goldreich-Levin theorem for large fields, which we prove in \Cref{sec:QGL}. Specifically, we can show that a distinguisher $\algo D$ that, given auxiliary input \textsc{Aux}, can distinguish between $(\vec u,\vec u^\intercal \vec x_0)$ and $(\vec u,r)$ with advantage $\eps$ can be converted into a quantum extractor that can extract ${\vec x_0}$ given \textsc{Aux} in time $\poly(1/\eps,q)$ with probability $\poly(\eps,1/q)$.\\

\noindent \noindent {\em Incorporating the revoked register $\reg R$}: To complete the security proof behind our key-revocable Dual-Regev scheme, we need to show something \emph{stronger}; namely, we need to argue that the Goldreich-Levin extractor succeeds on input \textsc{Aux} -- even conditioned on the fact that $\revoke$ outputs $\valid$ when applied on a separate register $\reg R$ (which may be entangled with \textsc{Aux}). We can consider two cases based on the security definition. 
\begin{itemize}
    \item Revocation succeeds with probability negligibly close to 1: in this case, applying the revocation or not does not make a difference since the state before applying revocation is negligibly close (in trace distance) to the state after applying revocation. Thus, the analysis is essentially the same as the setting where we ignore the register $\reg R$. 
    \item Revocation is only required to succeed with probability $1/\poly(\lambda)$: in this case, we do not know how to formally prove that the extractor and $\revoke$ simultaneously succeed with probability $1/\poly(\lambda)$. Thus, we state this as a conjecture  (see \Cref{thm:search-to-decision}) and leave the investigation of this conjecture to future works.  
\end{itemize}

\subsection{Key-Revocable Dual-Regev Encryption with Classical Revocation}

Recall that our key-revocable Dual-Regev public-key encryption scheme requires that a quantum state is \emph{returned} as part of revocation. In \Cref{cons:dual-regev-classical-revoc}, we give a  Dual-Regev encryption scheme with \emph{classical key revocation}.
The idea behind our scheme is the following: We use the same Gaussian decryption key from before, except that we also plant an appropriate phase into the state
$$
\ket{\psi_{\vec y}^{\vec v}} = \vec Z_q^{\lfloor \frac{q}{\nu} \rfloor \cdot \vec v} \ket{\psi_{\vec y}}= \sum_{\substack{\vec x \in \Z_q^m\\ \vec A \vec x = \vec y}}\rho_{\sigma}(\vec x)\, \omega_q^{\langle \vec x,\lfloor \frac{q}{\nu} \rfloor \cdot \vec v \rangle} 
 \, \ket{\vec x},
$$
where $\vec Z_q$ is the generalized $q$-ary Pauli-Z operator, $\vec v \rand \bit^m$ is a random string and $\nu$ is a sufficiently large integer with $\nu \ll q$ (to be determined later). The reason the complex phase is useful is the following: To revoke, we simply ask the recipient of the state to apply the Fourier transform and to measure in the computational basis. This results in a shifted $\LWE$ sample 
$$
\vec w = \hat{\vec s}^\intercal \vec A + \lfloor \frac{q}{\nu} \rfloor \cdot \vec v^\intercal  + \hat{\vec e}^\intercal \in \Z_q^m,
$$
where $\hat{\vec s}$ is random and $\hat{\vec e}$ is Gaussian. Note that $\vec w$ can easily be decrypted with access to a short trapdoor basis for $\vec A \in \Z_q^{n \times m}$. In particular, we will accept the \emph{classical} revocation certificate $\vec w$ if and only if it yields the string $\vec v$ once we decrypt it. This way there is no need to return a quantum state as part of revocation, and we can essentially “force” the adversary to forget the decryption key. Here, we crucially rely on the fact that Dual-Regev decryption keys are encoded in the \emph{computational basis}, whereas the revocation certificate can only be obtained via a measurement in the incompatible \emph{Fourier basis}. Intuitively, it seems impossible for any computationally bounded adversary to recover both pieces of information \emph{at the same time}.

To prove security, we follow a similar proof as in our previous construction from \Cref{sec:dual-regev}. 
First, we observe that any successful distinguisher that can distinguish Dual-Regev from uniform with quantum auxiliary input \textsc{Aux} (once revocation has taken place) can be converted into an extractor that obtains a valid pre-image of $\vec y$. Here, we use our Goldreich-Levin search-to-decision reduction with quantum auxiliary input. However, because the information returned to the challenger is not in fact a superposition of pre-images anymore, we cannot hope to break $\SIS$ as we did before.
Instead, we make the following observation: if the adversary simultaneously succeeds at passing revocation as well as distinguishing Dual-Regev ciphertexts from uniform, this means the adversary
\begin{itemize}
    \item knows an $\LWE$ encryption $\vec w \in \Z_q^m$ that yields $\vec v \in \bit^m$ when decrypted using a short trapdoor basis for $\vec A \in \Z_q^{n \times m}$, and

    \item can extract a short pre-image of $\vec y$ from its auxiliary input Aux.
\end{itemize}
In other words, the adversary simultaneously knows information about the computational basis as well as the Fourier domain of the Gaussian superposition state. We show that this violates a generic collapsing-type property of the Ajtai hash function recently proven by Bartusek, Khurana and Poremba~\cite[Theorem 5.5]{bartusek2023publiclyverifiable}, which relies on the hardness of subexponential $\LWE$ and $\SIS$.

\subsection{Applications}
\noindent We leverage our result of key-revocable Dual-Regev encryption to get key-revocable fully homomorphic encryption and revocable pseudorandom functions. While our constructions can easily
be adapted to enable \emph{classical revocation} (via our Dual-Regev scheme with classical key revocation in \Cref{cons:dual-regev-classical-revoc}), we choose to focus on the quantum revocation setting for simplicity.

\paragraph{Key-Revocable Dual-Regev Fully Homomorphic Encryption.}

Our first application of our key-revocable public-key encryption concerns fully homomorphic encryption schemes. We extend our key-revocable Dual-Regev scheme towards a (leveled) $\FHE$ scheme in \Cref{cons:DualGSW} by using the \textsf{DualGSW} variant of the $\FHE$ scheme by Gentry, Sahai and Waters~\cite{GSW2013,mahadev2018classical}.

To encrypt a bit $\mu \in \bit$ with respect to the public-key $(\vec A,\vec y)$, sample a matrix $\vec S \rand \Z_q^{n \times N}$ together with a Gaussian error matrix $\vec E \in \Z^{m\times N}$ and row vector $\vec e \in \Z^{N}$, and
output the ciphertext 
$$
\ct= \left[\substack{
\vec A^\intercal \vec S + \vec E\vspace{1mm}\\
\hline\vspace{1mm}\\
\vec y^\intercal  \vec S + \vec e}
 \right] 
+ \mu \cdot \vec G \Mod{q} \,\in\, \Z_q^{(m+1)\times N}.
$$
Here, $\vec G$ is the \emph{gadget matrix} which converts a binary vector in into its field representation over $\Z_q$. As before, the decryption key consists of a Gaussian superposition $\ket{\psi_{\vec y}}$ of pre-images of $\vec y$.

Note that the \textsf{DualGSW} ciphertext can be thought of as a column-wise concatenation of $N$-many independent Dual-Regev ciphertexts.
In \Cref{thm:security-Dual-GSW}, we prove the security of our construction by invoking the security of our key-revocable Dual-Regev scheme.

\newcommand{\bfSc}{{\bf S}}
\paragraph{Revocable Pseudorandom Functions.} Our next focus is on leveraging the techniques behind key-revocable public-key encryption to obtain revocable pseudorandom functions. Recall that the revocation security of pseudorandom functions stipulates the following: any efficient adversary (after successfully revoking the state that enables it to evaluate pseudorandom functions) cannot distinguish whether it receives pseudorandom outputs on many challenge inputs versus strings picked uniformly at random with advantage better than $\negl(\secparam)$. An astute reader might notice that revocation security does not even imply the traditional pseudorandomness guarantee! Hence, we need to additionally impose the requirement that a revocable pseudorandom function should also satisfy the traditional pseudorandomness guarantee.
\par Towards realizing a construction satisfying our definitions, we consider the following template: 
\begin{enumerate}
    \item First show that there exists a $\mu$-revocable pseudorandom function for $\mu=1$. Here, $\mu$-revocation security means the adversary receives $\mu$-many random inputs after revocation. 
    \item Next, we show that any 1-revocable pseudorandom function also satisfies the stronger notion of revocation security where there is no a priori bound on the number of challenge inputs received by the adversary.  
    \item Finally, we show that we can generically upgrade any revocable $\PRF$ in such a way that it also satisfies the traditional pseudorandomness property. 
\end{enumerate}
The second bullet is proven using a hybrid argument. The third bullet is realized by combining a revocable $\PRF$ with a post-quantum secure $\PRF$ (not necessarily satisfying revocation security). 
\par Hence, we focus the rest of our attention on proving the first bullet. \\

\noindent {\em 1-revocation security.} We start with the following warmup construction. The secret key $k$ comprises of matrices $\bfA,\{\bfSc_{i,0}, \bfSc_{i,1}\}_{i \in [\ell],b \in \{0,1\}}$, where $\bfA \xleftarrow{\$} \Zq^{n \times m}$, $\bfSc_{i,b} \in \Zq^{n \times n}$ such that all $\bfSc_{i,b}$ are sampled from some error distribution and the output of the pseudorandom function on $x$ is denoted to be $\lfloor \sum_{i \in [\ell]} \bfSc_{i,x_i} \bfA \rceil_p$, where $q \gg p$ and $\lfloor \cdot \rceil_p$ refers to a particular rounding operation modulo $p$.
\par In addition to handing out a regular $\PRF$ key $k$, we also need to generate a quantum key $\rho_k$ such that, given $\rho_k$ and any input $x$, we can efficiently compute $\prf(k,x)$. Moreover, $\rho_k$ can be revoked such that any efficient adversary after revocation loses the ability to evaluate the pseudorandom function. To enable the generation of $\rho_k$, we first modify the above construction. We generate $\bfy \in \Zq^n$ and include this as part of the key. The modified pseudorandom function, on input $x$, outputs $\lfloor \sum_{i \in [\ell]} \bfSc_{i,x_i} \bfy \rceil_p$. We denote $\sum_{i \in [\ell]} \bfSc_{i,x_i}$ by $\bfSc_{x}$ and, with this new notation, the output of the pseudorandom function can be written as $\lfloor \vec S_{x} \bfy \rceil_p$. 
\par With this modified construction, we now describe the elements as part of the quantum key $\rho_k$: 
\begin{itemize}
    \item For every $i \in [\ell]$, include $\bfS_{i,b} \bfA + \bfE_{i,b}$ in $\rho_k$, where $i \in [\ell]$ and $b \in \{0,1\}$. We sample $\bfS_{i,b}$ and $\bfE_{i,b}$ from a discrete Gaussian distribution with appropriate standard deviation $\sigma >0$.  
    \item Include $\ket{\psi_{\bfy}}$ which, as defined in the key-revocable Dual-Regev construction, is a Gaussian superposition of short solutions mapping $\bfA$ to $\bfy$. 
\end{itemize} 
To evaluate on an input $x$ using $\rho_k$, compute $\sum_{i} \bfS_{i,x_i} \bfA + \bfE_{i,x_i}$ and then using the state $\ket{\psi_{\bfy}}$, map this to $\sum_{i} \bfS_{i,x_i} \bfy + \bfE_{i,x_i}$. Finally, perform the rounding operation to get the desired result. 
\par Our goal is to show that after the adversary revokes $\ket{\psi_{\bfy}}$, on input a challenge input $x^*$ picked uniformly at random, it cannot predict whether it has received $\lfloor \sum_{i \in [N]} \bfSc_{i,x_i^*} \bfy \rceil_p$ or a uniformly random vector in $\Z_p^{n}$. \\

\noindent {\em Challenges in proving security}: We would like to argue that when the state $\ket{\psi_{\bfy}}$ is revoked, the adversary loses its ability to evaluate the pseudorandom function. Unfortunately, this is not completely true. For all we know, the adversary could have computed the pseudorandom function on many inputs of its choice before the revocation phase and it could leverage this to break the security after revocation. For instance, suppose say the input is of length $O(\log \secparam)$ then in this case, the adversary could evaluate the pseudorandom function on all possible inputs before revocation. After revocation, on any challenge input $x^*$, the adversary can then successfully predict whether it receives a pseudorandom output or a uniformly chosen random output. Indeed, a pseudorandom function with $O(\log \secparam)$-length input is learnable and hence, there should be no hope of proving it to be key-revocable. This suggests that, at the very least, we need to explicitly incorporate the fact that $x^*$ is of length $\omega(\log \secparam)$, and more importantly, should have enough entropy, in order to prove security. \\

\noindent {\em Our insight}: Our insight is to reduce the security of revocable pseudorandom function to the security of key-revocable Dual-Regev public-key encryption. At a high level, our goal is to set up the parameters in such a way that the following holds:
\begin{itemize}
    \item $(\bfA,\bfy)$, defined above, is set to be the public key corresponding to the Dual-Regev public-key encryption scheme, 
    \item $\ket{\psi_{\bfy}}$, which is part of the pseudorandom function key, is set to be the decryption key of the Dual Regev scheme,
    \item Suppose that the challenge ciphertext, denoted by $\ct^*$, comprises of two parts: $\ct^*_1 \in \Zq^{n \times m}$ and $\ct^*_2 \in \Zq^{n}$. Note that if $\ct_1^* \approx \bfs^\intercal \bfA$ and $\ct_2^* \approx \bfs^\intercal \bfy$, for some $\LWE$ secret vector $\bfs$, then $\ct_1^*$ can be mapped onto $\ct^*_2$ using the state $\ket{\psi_{\bfy}}$. We use $\ct^*_1$ to set the challenge input $x^*$ in such a way that $\ct^*_2$ is the output of the pseudorandom function on $x^*$. This implicitly resolves the entropy issue we discussed above; by the semantic security of Dual-Regev, there should be enough entropy in $\ct_1^*$ which translates to the entropy of $x^*$. 
\end{itemize}
\noindent It turns out that our goal is quite ambitious: in particular, it is unclear how to set up the parameters such that the output of the pseudorandom function on $x$ is exactly $\ct^*_2$. Fortunately, this will not be a deterrant, we can set up the parameters such that the output is $\approx \ct_2^* + \vec u$, where $\vec u$ is a vector that is known to reduction. 
\par Once we set up the parameters, we can then reduce the security of revocable pseudorandom functions to revocable Dual Regev. \\

\noindent {\em Implementation details}: So far we established the proof template should work but the implementation details of the proof need to be fleshed out. Firstly, we set up the parameters in such a way that $\ell = nm\lceil \log q \rceil$. This means that there is a bijective function mapping $[n] \times [m] \times [\lceil \log q\rceil]$ to $[\ell]$. As a result, the quantum key $\rho_k$ can be alternately viewed as follows: 
\begin{itemize}
    \item For every $i \in [n],j \in [m],\tau \in [\lceil \log q \rceil], b \in \{0,1\}$, include $\bfS_{b}^{(i,j,\tau)} \bfA + \bfE_{b}^{(i,j,\tau)}$ in $\rho_k$. We sample $\bfS_{b}^{(i,j,\tau)}$ and $\bfE_{b}^{(i,j,\tau)}$ from a discrete Gaussian with appropriate standard deviation $\sigma >0$.
\end{itemize}
The output of the pseudorandom function on input $x$ can now be interpreted as 
$$
\PRF(k,x) = \left\lceil \sum_{\substack{i \in [n],j \in [m]\\ \tau \in [\lceil \log q \rceil ]}} \bfS^{(i,j,\tau)}_{x_i} \bfy \right\rceil_p
$$
\par Next, we modify $\rho_k$ as follows: instead of generating, $\bfS_{b}^{(i,j,\tau)} \bfA + \bfE_{b}^{(i,j,\tau)}$, we instead generate $\bfS_{b}^{(i,j,\tau)} \bfA + \bfE_{b}^{(i,j,\tau)} + \bfM^{(i,j,k)}_b$, for any set of matrices $\{\bfM^{(i,j,\tau)}_b \}$. The change should be undetectable to a computationally bounded adversary, thanks to the quantum hardness of learning with errors. In the security proof, we set up the challenge input $x^*$ in such a way that summing up the matrices $\bfM_{x^*_i}^{(i,j,\tau)}$ corresponds to $\ct_1^*$, where $\ct_1^*$ is part of the key-revocable Dual-Regev challenge ciphertext as discussed above. With this modification, when $\rho_k$ is evaluated on $x^*$, we get an output that is close to $\ct_2^* + \bfu$, where $\bfu \approx \sum_{i \in [n],j \in [m], \tau \in [\lceil \log(q) \rceil]} \bfS^{(i,j,\tau)}_{x_i} \bfy$ is known to the reduction (as discussed above) -- thereby violating the security of key-revocable Dual-Regev scheme.


\subsection{Related Work}
\label{sec:related}

\paragraph{Copy-Protection.} Of particular relevance to our work is the notion of copy-protection introduced by Aaronson~\cite{Aar09}. Informally speaking, a copy-protection scheme is a
compiler that transforms programs into quantum states in such a way that using the resulting states, one can run the original program. Yet, the security guarantee stipulates that any adversary given one copy of the state cannot produce a bipartite state wherein both parts compute the original program. 
\par While copy-protection is known to be impossible for arbitrary unlearnable functions~\cite{ALP21,AK22}, identifying interesting functionalities for which copy-protection is feasible has been an active research direction~\cite{coladangelo2020quantum,cryptoeprint:2022/884,AKL23}. Of particular significance is the problem of copy-protecting cryptographic functionalities, such as decryption and signing functionalities. Coladangelo et al.~\cite{CLLZ21} took the first step in this direction and showed that it is feasible to copy-protect decryption functionalities and pseudorandom functions assuming the existence of post-quantum indistinguishability obfuscation. While a very significant first step, the assumption of post-quantum iO is unsatisfactory: there have been numerous post-quantum iO candidate proposals~\cite{BMSZ16,CVW18,BDGM20,DQVW21,GP21,WW21}, but not one of them have been based on well-studied assumptions\footnote{We remark that, there do exist {post-quantum-insecure} iO schemes based on well-founded assumptions~\cite{JLS21}.}. 
\par Our work can be viewed as copy-protecting cryptographic functionalities based on the post-quantum hardness of the learning with errors problem under a weaker yet meaningful security guarantee.

\paragraph{Secure Software Leasing.} Another primitive relevant to revocable cryptography is secure software leasing~\cite{ALP21}. The notion of secure software leasing states that any program can be compiled into a functionally equivalent program, represented as a quantum state, in such a way that once the compiled program is returned\footnote{Acording to the terminology of~\cite{ALP21}, this refers to finite term secure software leasing.}, the honest evaluation algorithm on the residual state cannot compute the original functionality. Key-revocable encryption can be viewed as secure software leasing for decryption algorithms. However, unlike secure software leasing, key-revocable encryption satisfies a much stronger security guarantee, where there is no restriction on the adversary to run honestly after returning back the software. Secure leasing for different functionalities, namely, point functions~\cite{coladangelo2020quantum,BJLPS21}, evasive functions~\cite{ALP21,KNY21} and pseudorandom functions~\cite{ALLZZ21} have been studied by recent works.

\paragraph{Encryption Schemes with Revocable Ciphertexts.} 
Unruh \cite{Unruh2013} proposed a (private-key) quantum timed-release
encryption scheme that is \emph{revocable}, i.e. it allows a user to \emph{return} the ciphertext of a quantum timed-release encryption scheme, thereby losing all access to the data. Unruh's scheme uses conjugate coding~\cite{Wiesner83,BB84} and relies on the \emph{monogamy of entanglement} in order to guarantee that revocation necessarily erases information about the plaintext. 
Broadbent and Islam~\cite{BI20} introduced the notion of \emph{certified deletion}\footnote{This notion is incomparable with another related notion called unclonable encryption~\cite{BL20,AK21,cryptoeprint:2022/884}, which informally guarantees that it should be infeasible to clone quantum ciphertexts without losing information about the encrypted message.} and constructed a private-key quantum encryption scheme with the aforementioned feature which is inspired by the quantum key distribution protocol~\cite{BB84,Tomamichel2017largelyself}. 
In contrast with Unruh's~\cite{Unruh2013} notion of revocable quantum ciphertexts which are eventually returned and verified, Broadbent and Islam~\cite{BI20} consider certificates which are entirely classical. Moreover, the security definition requires that, once the certificate is successfully verified, the plaintext remains hidden even if the secret key is later revealed.

Using a hybrid encryption scheme, Hiroka, Morimae, Nishimaki and Yamakawa~\cite{hiroka2021quantum} extended the scheme in~\cite{Broadbent_2020} to both public-key and attribute-based encryption with certified deletion via \emph{receiver non-committing} encryption~\cite{10.5555/1756169.1756191,10.1145/237814.238015}. As a complementary result, the authors also gave a public-key encryption scheme with certified deletion which is \emph{publicly verifiable} assuming the
existence of one-shot signatures and extractable witness encryption. Using 
\emph{Gaussian superpositions},
Poremba~\cite{Poremba22} proposed \emph{Dual-Regev}-based public-key and fully homomorphic encryption schemes with certified deletion which are publicly verifiable and proven secure assuming a \emph{strong Gaussian-collapsing conjecture} --- a strengthening of the collapsing property of the Ajtai hash. Bartusek and Khurana~\cite{BK22} revisited the notion of certified deletion and presented a unified approach for how to generically convert any public-key, attribute-based, fully-homomorphic, timed-release or witness encryption scheme into an equivalent quantum encryption scheme with certified deletion. In particular, they considered a stronger notion called \emph{certified everlasting security} which allows the adversary to be computationally unbounded once a valid deletion certificate is submitted. 

\section*{Acknowledgements}
We thank Fatih Kaleoglu and Ryo Nishimaki for several insightful discussions.
\par This work was done (in part) while the authors were visiting the Simons Institute for the Theory of Computing. P.A. is supported by a research gift from Cisco. A.P. is partially supported by AFOSR YIP (award number FA9550-16-1-0495), the Institute for Quantum Information and Matter (an NSF Physics Frontiers Center; NSF Grant PHY-1733907) and by a grant from the Simons 
Foundation (828076, TV).
V.V. is supported by DARPA under Agreement No. HR00112020023, NSF CNS-2154149 and a Thornton Family Faculty Research
Innovation Fellowship.

\section{Preliminaries}
\noindent Let $\secparam \in \N$ denote the security parameter throughout this work. We assume that the reader is familiar with the fundamental cryptographic concepts. 

\subsection{Quantum Computing} For a comprehensive background on quantum computation, we refer to \cite{NielsenChuang11,Wilde13}. We denote a finite-dimensional complex Hilbert space by $\mathcal{H}$, and we use subscripts to distinguish between different systems (or registers). For example, we let $\mathcal{H}_{A}$ be the Hilbert space corresponding to a system $A$. 
The tensor product of two Hilbert spaces $\algo H_A$ and $\algo H_B$ is another Hilbert space denoted by $\algo H_{AB} = \algo H_A \otimes \algo H_B$.
Let $\algo L(\algo H)$
denote the set of linear operators over $\algo H$. A quantum system over the $2$-dimensional Hilbert space $\mathcal{H} = \mathbb{C}^2$ is called a \emph{qubit}. For $n \in \mathbb{N}$, we refer to quantum registers over the Hilbert space $\mathcal{H} = \big(\mathbb{C}^2\big)^{\otimes n}$ as $n$-qubit states. More generally, we associate \emph{qudits} of dimension $d \geq 2$ with a $d$-dimensional Hilbert space $\mathcal{H} = \mathbb{C}^d$. 
For brevity, we write $\algo H_d^n = \algo H_d^{\otimes n}$, where $\algo H_d$ is $d$-dimensional. We use the word \emph{quantum state} to refer to both pure states (unit vectors $\ket{\psi} \in \mathcal{H}$) and density matrices $\rho \in \mathcal{D}(\mathcal{H)}$, where we use the notation $\mathcal{D}(\mathcal{H)}$ to refer to the space of positive semidefinite matrices of unit trace acting on $\algo H$. 
Occasionally, we consider \emph{subnormalized states}, i.e. states in the space of positive semidefinite operators over $\algo H$ with trace norm not exceeding $1$.

The \emph{trace distance} of two density matrices $\rho,\sigma \in \mathcal{D}(\mathcal{H)}$ is given by
$$
\TD(\rho,\sigma) = \frac{1}{2} \Tr\left[ \sqrt{ (\rho - \sigma)^\dag (\rho - \sigma)}\right].
$$

Let $q \geq 2$ be a modulus and $n \in \N$ and let $\omega_q = e^{ \frac{2 \pi i}{q}} \in \mathbb{C}$ denote the primitive $q$-th root of unity.
The $n$-qudit \emph{$q$-ary quantum Fourier transform} over the ring $\Z_q^n$ is defined by the operation,
$$
\FT_q : \quad \ket{\vec x} \quad \mapsto \quad \sqrt{q^{-n}} \displaystyle\sum_{\vec y \in \Z_q^n} \omega_q^{\langle \vec x,\vec y\rangle} \ket{\vec y}, \quad\quad \forall \vec x \in \Z_q^n.
$$
The $q$-ary quantum Fourier transform is \emph{unitary} and can be efficiently performed on a quantum computer for any modulus $q \geq 2$~\cite{892139}.

A quantum channel $\Phi: \algo L(\algo H_A) \rightarrow \algo L(\algo H_B)$ is a linear map between linear operators over the Hilbert spaces $\algo H_A$ and $\algo H_B$. Oftentimes, we use the compact notation $\Phi_{A \rightarrow B}$ to denote a quantum channel between $\algo L(\algo H_A)$ and $\algo L(\algo H_B)$. We say that a channel $\Phi$ is \emph{completely positive} if, for a reference system $R$ of arbitrary size, the induced map $I_R \otimes \Phi$ is positive, and we call it \emph{trace-preserving} if $\Tr[\Phi(X)] = \Tr[X]$, for all $X \in \algo L(\algo H)$. A quantum channel that is both completely positive and trace-preserving is called a quantum $\CPTP$ channel. 

A polynomial-time \emph{uniform} quantum algorithm (or $\QPT$ algorithm) is a polynomial-time family of quantum circuits given by $\algo C = \{C_\lambda\}_{\lambda \in \N}$, where each circuit $C \in \algo C$ is described by a sequence of unitary gates and measurements; moreover, for each $\lambda \in \N$, there exists a deterministic polynomial-time Turing machine that, on input $1^\lambda$, outputs a circuit description of $C_\lambda$. Similarly, we also define (classical) probabilistic polynomial-time $(\PPT)$ algorithms. A quantum algorithm may, in general, receive (mixed) quantum states as inputs and produce (mixed) quantum states as outputs. Occasionally, we restrict $\QPT$ algorithms implicitly; for example, if we write $\Pr[\mathcal{A}(1^{\lambda}) = 1]$ for a $\QPT$ algorithm $\mathcal{A}$, it is implicit that $\mathcal{A}$ is a $\QPT$ algorithm that outputs a single classical bit.

A polynomial-time \emph{non-uniform} quantum algorithm is a family $\{(C_\lambda,\nu_\lambda)\}_{\lambda \in \N}$, where $\{C_\lambda \}_{\lambda \in \N}$ is a polynomial-size (not necessarily uniformly generated) family of circuits where, for each $\lambda \in \N$, a subset of input qubits to $C_\lambda$ consists of a polynomial-size auxiliary density matrix $\nu_\lambda$.
We use the following notion of indistinguishability of quantum states
in the presence of auxiliary inputs.

\begin{definition}[Indistinguishability of ensembles of quantum states, \cite{https://doi.org/10.48550/arxiv.quant-ph/0511020}]
\label{def: indistinguishability}
Let $p: \mathbb{N} \rightarrow \mathbb{N}$ be a polynomially bounded function,
and let $\rho_\lambda$ and $\sigma_\lambda$
be $p(\lambda)$-qubit quantum states. We say that $\{\rho_{\lambda}\}_{\lambda \in \mathbb{N}}$ and $\{\sigma_\lambda\}_{\lambda \in \mathbb{N}}$ are quantum computationally indistinguishable ensembles of quantum states, denoted by $\rho_{\lambda} \approx_c \sigma_{\lambda}\,,$
if, for any $\QPT$ distinguisher $\mathcal{D}$ with single-bit output, any polynomially bounded $q: \mathbb{N} \rightarrow \mathbb{N}$, any family of $q(\lambda)$-qubit auxiliary states $\{\nu_{\lambda}\}_{\lambda \in \mathbb{N}}$, and every $\lambda \in \mathbb{N}$,
$$ \big| \Pr[\mathcal{D}(\rho_{\lambda} \otimes \nu_{\lambda})=1] - \Pr[\mathcal{D}(\sigma_{\lambda} \otimes \nu_{\lambda})=1] \big| \leq \negl(\lambda) \,.$$
\end{definition}

\begin{lemma}["Almost As Good As New" Lemma, \cite{aaronson2016complexity}]\label{lem:almost} Let $\rho \in \algo D(\algo H)$ be a density matrix over a Hilbert space $\algo H$. Let $U$ be an arbitrary unitary and let $(\bm{\Pi}_0,\bm{\Pi}_1 = \vec I - \bm{\Pi}_0)$ be projectors acting on $\algo H \otimes \algo H_\aux$. We interpret $(U,\bm{\Pi}_0,\bm{\Pi}_1)$ as a measurement performed by appending an ancillary system in the state $\ketbra{0}{0}_\aux$, applying the unitary $U$ and subsequently performing the two-outcome measurement $\{\bm{\Pi}_0,\bm{\Pi}_1\}$ on the larger system. Suppose that the outcome corresponding to $\bm{\Pi}_0$ occurs with probability $1-\eps$, for some $\eps \in [0,1]$. In other words, it holds that $\Tr[\boldsymbol{\Pi}_0(U \rho \otimes \ketbra{0}{0}_\aux U^\dag)] = 1- \eps$. Then,
$$
\TD(\rho,\widetilde{\rho}) \leq \sqrt{\eps},
$$
where $\widetilde{\rho}$ is the state after performing the measurement and applying $U^\dag$, and after tracing out $\algo H_\aux$:
$$
\widetilde{\rho} = \Tr_\aux\left[U^\dag \left( 
\bm{\Pi}_0 U (\rho \otimes \ketbra{0}{0}_\aux)U^\dag \bm{\Pi}_0 + \bm{\Pi}_1 U (\rho \otimes  \ketbra{0}{0}_\aux)U^\dag \bm{\Pi}_1 
\right)U \right].
$$
\end{lemma}

\begin{lemma}[Quantum Union Bound, \cite{Gao15}]\label{lem:union} Let $\rho \in \algo D(\algo H)$ be a state and let $\bm \Pi_1,\dots,\bm \Pi_n \geq 0$ be sequence of (orthogonal) projections acting on $\algo H$. Suppose that, for every $i \in [n]$, it holds that $\Tr[\bm \Pi_i \rho] = 1 - \eps_i$, for $\eps_i \in [0,1]$. Then, if we sequentially measure $\rho$ with projective measurements $\{\bm \Pi_1, \vec I - \bm \Pi_1\}, \dots, \{\bm \Pi_n, \vec I - \bm \Pi_n\}$, the probability that all measurements succeed is at least
$$
\Tr[\bm \Pi_n \cdots \bm \Pi_1 \rho \bm \Pi_1 \cdots \bm \Pi_n ] \geq 1- 4 \sum_{i =1}^n \eps_i. 
$$

\end{lemma}

\subsection{Lattices and Cryptography}\label{sec:lattices} Let $n,m,p,q \in \N$ be positive integers. The rounding operation for $q \geq p \geq 2$ is the function
$$
\lfloor \cdot \rfloor_p \,\, : \,\, \Z_q \rightarrow \Z_p \,\, : \,\, x \mapsto \lfloor (p/q) \cdot x \rfloor \Mod{p}.
$$
A $\emph{lattice}$ $\Lambda \subset \mathbb{R}^m$ is a discrete subgroup of $\mathbb{R}^m$. 
Given a lattice $\Lambda \subset \mathbb{R}^m$ and a vector $\vec t \in \mathbb{R}^m$, we define the coset with respect to vector $\vec t $ as the lattice shift $\Lambda - \vec t = \{\vec x \in \mathbb{R}^m :\, \vec x + \vec t \in \Lambda\}$. Note that many different shifts $\vec t$ can define the same coset. The \emph{dual} of a lattice $\Lambda \subset \mathbb{R}^m$, denoted by $\Lambda^*$, is the lattice of all $y \in \mathbb{R}^m$ that satisfy $\langle \vec y,\vec x \rangle \in \Z$, for every $\vec x \in \Lambda$. In other words, we let
$$
\Lambda^* = \left\{ \vec y \in \mathbb{R}^m \, : \, \langle \vec y,\vec x \rangle \in \Z, \text{ for all }  \vec x \in \Lambda\right\}.
$$

In this work, we mainly consider \emph{$q$-ary lattices} $\Lambda$ that that satisfy $q\Z^m \subseteq \Lambda \subseteq \Z^m$, for some integer modulus $q \geq 2$. Specifically, we consider the lattice generated by a matrix $\vec A \in \Z_q^{n \times m}$ for some $n,m \in \N$ that consists of all vectors which are perpendicular to the rows of $\vec A$, namely
$$
\Lambda_q^\bot(\vec A) = \{\vec x \in \Z^m: \, \vec A \cdot \vec x = \vec 0 \Mod{q}\}.
$$
For any \emph{syndrome} $\vec y \in \Z_q^n$ in the column span of $\vec A$, we also consider the coset $\Lambda_q^{\vec y}(\vec A)$ given by
$$
\Lambda_q^{\vec y}(\vec A) = \{\vec x \in \Z^m: \, \vec A \cdot \vec x = \vec y \Mod{q}\} = \Lambda_q^\bot(\vec A) + \vec c,
$$
where $\vec c\in \Z^m$ is an arbitrary integer solution to the equation $\vec A \vec c = \vec y \Mod{q}$.

We use the following facts due to Gentry, Peikert and Vaikuntanathan~\cite{cryptoeprint:2007:432}.

\begin{lemma}[\cite{cryptoeprint:2007:432}, Lemma 5.1]\label{lem:full-rank}
Let $n\in \N$ and let $q \geq 2$ be a prime modulus with $m \geq 2 n \log q$. Then, for all but a $q^{-n}$ fraction of $\vec A \in \Z_q^{n \times m}$, the subset-sums of the columns of $\vec A$ generate $\Z_q^n$. In other words, a uniformly random matrix $\vec A \rand \Z_q^{n \times m}$ is full-rank with overwhelming probability.
\end{lemma}

\paragraph{Gaussian Distribution.} 

The \emph{Gaussian measure} $\rho_\sigma$ with parameter $\sigma > 0$ is defined as
\begin{align*}
\rho_\sigma(\vec x) = \exp(-\pi \|\vec x \|^2/ \sigma^2), \quad \,\, \forall \vec x \in \mathbb{R}^m.    
\end{align*}
Let $\Lambda \subset \mathbb{R}^m$ be a lattice and let $\vec t \in \mathbb{R}^m$. We define the \emph{Gaussian mass} of $\Lambda - \vec t$ as the quantity
\begin{align*}
\rho_\sigma(\Lambda - \vec t) = \sum_{\vec y \in \Lambda}\rho_\sigma(\vec y- \vec t).
\end{align*}

The \emph{discrete Gaussian distribution} $D_{\Lambda - \vec t,\sigma}$ assigns probability proportional to $e^{-\pi \|\vec x \|^2/ \sigma^2}$ to every vector $\vec x \in \Lambda - \vec t$. In other words, we have
$$
D_{\Lambda - \vec t,\sigma} (\vec x)= \frac{\rho_\sigma(\vec x)}{\rho_\sigma(\Lambda - \vec t)}, \quad \,\, \forall \vec x \in \Lambda - \vec t.
$$
In particular, for any coset $\Lambda_q^{\vec y}(\vec A)$ with $\vec y \in \Z_q^n$, the discrete Gaussian  $D_{\Lambda_q^{\vec y}(\vec A), \sigma}$ (centered around the origin) assigns probability proportional to $e^{-\pi \|\vec x \|^2/ \sigma^2}$ to every vector $\vec x \in \Lambda_q^{\vec y}(\vec A)$, and $0$ otherwise.

The following lemma follows from \cite[Lemma 2.11]{10.1007/11681878_8} and \cite[Lemma 5.3]{cryptoeprint:2007:432}.

\begin{lemma}\label{lem:Gaussian1}
Let $n \in \N$ and let $q$ be a prime with $m \geq 2 n\log q$. Let $\vec A \in \Z_q^{n \times m}$ be a matrix whose columns generate $\Z_q^n$. Let $\vec y \in \Z_q^n$ be arbitrary. Then, for any $\sigma \geq \omega(\sqrt{\log m})$, there exists a negligible function $\eps(m)$ such that
$$
D_{\Lambda_q^{\vec y}(\vec A),\sigma}(\vec x) \, \leq \, 2^{-m} \cdot \frac{1+\eps}{1-\eps}, \quad\quad \forall \, \vec x \in \Lambda_q^{\vec y}(\vec A).
$$
\end{lemma}

Let $\algo B^m(\vec 0, r) = \{ \vec x \in \mathbb{R}^m \, : \, \|\vec x\| \leq r\}$ denote the $m$-dimensional ball of radius $r > 0$.
We use of the following tail bound for the Gaussian mass of a lattice \cite[Lemma 1.5 (ii)]{Banaszczyk1993}.

\begin{lemma}\label{lem:gaussian-tails}
For any $m$-dimensional lattice $\Lambda$, shift $\vec t \in \mathbb{R}^m$, $\sigma>0$ and $c \geq (2 \pi)^{-\frac{1}{2}}$ it holds that
$$
\rho_\sigma \left( (\Lambda - \vec t)\setminus \algo B^m(\vec 0, c \sqrt{m} \sigma) \right) \leq (2 \pi e c^2)^{\frac{m}{2}} e^{- \pi c^2m} \rho_\sigma(\Lambda).
$$
\end{lemma}

In addition, we also make use of the following tail bound for the discrete Gaussian which follows from \cite[Lemma 4.4]{1366257} and \cite[Lemma 5.3]{cryptoeprint:2007:432}.

\begin{lemma}\label{lem:Gaussian2}
Let $n \in \N$ and let $q$ be a prime with $m \geq 2 n\log q$. Let $\vec A \in \Z_q^{n \times m}$ be a matrix whose columns generate $\Z_q^n$. Let $\vec y \in \Z_q^n$ be arbitrary. Then, for any $\sigma \geq \omega(\sqrt{\log m})$, there exists a negligible function $\eps(m)$ such that
$$
\Pr_{\vec x \sim D_{\Lambda_q^{\vec y}(\vec A),\sigma}}\Big[ \|\vec x\| > \sigma \sqrt{m} \Big] \leq 2^{-m} \cdot \frac{1+\eps}{1-\eps}.
$$
\end{lemma}
A consequence of \Cref{lem:gaussian-tails} is that the Gaussian distribution $D_{\Z^m,\sigma}$ is essentially only supported on the finite set $\{\vec x \in \Z^m : \|\vec x\| \leq \sigma \sqrt{m}\}$, which suggests the use of \emph{truncation}.
\begin{definition}[Truncated discrete Gaussian distribution]
Let $m \in \N$, $q \geq 2$ be an integer modulus and let $\sigma >0$ be a parameter. Then, the \emph{truncated} discrete Gaussian distribution $D_{\Z_q^m,\sigma}$ with finite support $\{\vec x \in \Z^m \cap (-\frac{q}{2},\frac{q}{2}]^m : \|\vec x\| \leq \sigma \sqrt{m}\}$ is defined as the density
$$
D_{\Z_q^m,\sigma}(\vec x) = \frac{\rho_\sigma(\vec x)}{\displaystyle\sum_{\vec y \in \Z_q^m,\|\vec y\| \leq \sigma\sqrt{m} } \rho_\sigma(\vec y)}.
$$    
\end{definition}
Finally, we recall the following \emph{noise smudging} property. 

\begin{lemma}[Noise smudging, \cite{Dodis10}]\label{lem:shifted-gaussian}
Let $y,\sigma > 0$. Then, the statistical distance between the distribution $D_{\Z,\sigma}$ and $D_{\Z,\sigma+y}$ is at most $y/\sigma$.
\end{lemma}

We use the following technical lemma on the min-entropy of the truncated discrete Gaussian distribution, which we prove below.

\begin{lemma}
\label{lem:max:ampl:bound}
Let $n \in \N$ and let $q$ be a prime with $m \geq 2 n\log q$. Let $\vec A \in \Z_q^{n \times m}$ be a matrix whose columns generate $\Z_q^n$. Then, for any $\sigma \geq \omega(\sqrt{\log m})$, there exists a negligible $\eps(m)$ such that
$$
\max_{\vec y \in \Z_q^n} \,\,\max_{\substack{\vec x \in \Z_q^m, \, \|\vec x\| \leq \sigma \sqrt{m}\\
\vec A \vec x = \vec y \Mod{q}}}
\left\{
\frac{\rho_\sigma(\vec x)}{\displaystyle\sum_{\substack{\vec z \in \Z_q^m,\|\vec z\| \leq \sigma\sqrt{m}\\
\vec A \vec z = \vec y \Mod{q}}} \rho_\sigma(\vec z)} \right\} \,\,\leq \,\,  2^{-m+1} \cdot \frac{1+\eps}{1-\eps}.
$$
\end{lemma}
\begin{proof}
Suppose that $\vec A \in \Z_q^{n \times m}$ is a matrix whose columns generate $\Z_q^n$, i.e., $\vec A$ is full-rank. Then,
\begin{align*}
&\max_{\vec y \in \Z_q^n} \,\max_{\substack{\vec x \in \Z_q^m, \, \|\vec x\| \leq \sigma \sqrt{m}\\
\vec A \vec x = \vec y \Mod{q}}}
\left\{
\frac{\rho_\sigma(\vec x)}{\displaystyle\sum_{\substack{\vec z \in \Z_q^m,\|\vec z\| \leq \sigma\sqrt{m}\\
\vec A \vec z = \vec y \Mod{q}}} \rho_\sigma(\vec z)} \right\}\\
&\leq \quad 
\max_{\vec y \in \Z_q^n}\sup_{\vec x \in \Lambda_q^{\vec y}(\vec A)}D_{\Lambda_q^{\vec y}(\vec A),\sigma}(\vec x) \\
&\quad+ \max_{\vec y \in \Z_q^n}\max_{\substack{\vec x \in \Z_q^m, \, \|\vec x\| \leq \sigma \sqrt{m}\\
\vec A \vec x = \vec y \Mod{q}}} \left| 
\frac{\rho_\sigma(\vec x)}{\displaystyle\sum_{\substack{\vec z \in \Z_q^m,\|\vec z\| \leq \sigma\sqrt{m}\\
\vec A \vec z = \vec y \Mod{q}}} \rho_\sigma(\vec z)}-
\frac{\rho_\sigma(\vec x)}{\displaystyle\sum_{\substack{\vec z \in \Z^m\\
\vec A \vec z = \vec y \Mod{q}}} \rho_\sigma(\vec z)}    \right|\\
&\leq \quad  \max_{\vec y \in \Z_q^n}\sup_{\vec x \in \Lambda_q^{\vec y}(\vec A)}D_{\Lambda_q^{\vec y}(\vec A),\sigma}(\vec x)\\
&\quad+ \max_{\vec y \in \Z_q^n}\max_{\substack{\vec x \in \Z_q^m, \, \|\vec x\| \leq \sigma \sqrt{m}\\
\vec A \vec x = \vec y \Mod{q}}}
\frac{\rho_\sigma(\vec x)}{\displaystyle\sum_{\substack{\vec z \in \Z_q^m,\|\vec z\| \leq \sigma\sqrt{m}\\
\vec A \vec z = \vec y \Mod{q}}} \rho_\sigma(\vec z)}
\cdot
\frac{\rho_\sigma(\Lambda_q^{\vec y}(\vec A)\setminus \algo B^m(\vec 0, \sigma \sqrt{m}))}{\rho_\sigma(\Lambda_q^{\vec y}(\vec A))}
\end{align*}
where $B^m(\vec 0, r) = \{ \vec x \in \mathbb{R}^m \, : \, \|\vec x\| \leq r\}$.
Using the fact that
$$
\frac{\rho_\sigma(\vec x)}{\displaystyle\sum_{\substack{\vec z \in \Z_q^m,\|\vec z\| \leq \sigma\sqrt{m}\\
\vec A \vec z = \vec y \Mod{q}}} \rho_\sigma(\vec z)} \,\,\leq \,\, 1,
$$
for $\vec x \in \Z_q^m$ with $\vec A\vec x = \vec y \Mod{q}$, and the fact that
$$
\Pr_{\vec v \sim D_{\Lambda_q^{\vec y}(\vec A),\sigma}}\Big[ \|\vec v\| > \sigma \sqrt{m} \Big] = 
\frac{\rho_\sigma(\Lambda_q^{\vec y}(\vec A)\setminus \algo B^m(\vec 0, \sigma \sqrt{m}))}{\rho_\sigma(\Lambda_q^{\vec y}(\vec A))}
$$
we get that
\begin{align*}
&\max_{\vec y \in \Z_q^n} \,\max_{\substack{\vec x \in \Z_q^m, \, \|\vec x\| \leq \sigma \sqrt{m}\\
\vec A \vec x = \vec y \Mod{q}}}
\left\{
\frac{\rho_\sigma(\vec x)}{\displaystyle\sum_{\substack{\vec z \in \Z_q^m,\|\vec z\| \leq \sigma\sqrt{m}\\
\vec A \vec z = \vec y \Mod{q}}} \rho_\sigma(\vec z)} \right\}\\
&\leq  \max_{\vec y \in \Z_q^n}\left\{  \sup_{\vec x \in \Lambda_q^{\vec y}(\vec A)} D_{\Lambda_q^{\vec y}(\vec A),\sigma}(\vec x) + \Pr_{\vec v \sim D_{\Lambda_q^{\vec y}(\vec A),\sigma}}\Big[ \|\vec v\| > \sigma \sqrt{m} \Big]
\right\}.
\end{align*}
Because $\sigma \geq \omega(\sqrt{\log m})$, the claim then follows from \Cref{lem:Gaussian1} and \Cref{lem:Gaussian2}.

\end{proof}

\paragraph{The Short Integer Solution problem.}

The \emph{Short Integer Solution} ($\SIS$) problem was introduced by Ajtai~\cite{DBLP:conf/stoc/Ajtai96} in his seminal work on average-case lattice problems. 

\begin{definition}[Short Integer Solution problem, \cite{DBLP:conf/stoc/Ajtai96}]\label{def:ISIS} Let $n,m \in \N$, $q\geq 2$ be a modulus and let $\beta >0$ be a parameter. The Short Integer Solution problem $(\SIS_{n,q,\beta}^m)$ problem is to find a short solution $\vec x \in \Z^m$ with $\|\vec x\| \leq \beta$ such that $\vec A \cdot \vec x = \vec 0 \Mod{q}$ given as input a matrix $\vec A \rand \Z_q^{n \times m}$.
\end{definition}

Micciancio and Regev~\cite{DBLP:journals/siamcomp/MicciancioR07} showed that the $\SIS$ problem is, on the average, as hard as approximating worst-case lattice problems to within small factors. Subsequently, Gentry, Peikert and Vaikuntanathan~\cite{cryptoeprint:2007:432} gave an improved reduction showing that, for parameters $m=\poly(n)$, $\beta=\poly(n)$ and prime $q \geq \beta \cdot \omega(\sqrt{n \log q})$, the average-case $\SIS_{n,q,\beta}^m$ problem is as hard as approximating the shortest independent vector problem $(\mathsf{SIVP})$ problem in the
worst case to within a factor $\gamma = \beta \cdot \tilde{O}(\sqrt{n})$.
We assume that $\SIS_{n,q,\beta}^m$, for $m=\Omega(n \log q)$, $\beta = 2^{o(n)}$ and $q=2^{o(n)}$, is hard against quantum adversaries running in time $\poly(q)$ with success probability $\poly(1/q)$.

\paragraph{The Learning with Errors problem.}

The \emph{Learning with Errors} problem was introduced by Regev~\cite{Regev05} and serves as the primary basis of hardness of post-quantum cryptosystems. The problem is defined as follows.

\begin{definition}[Learning with Errors problem, \cite{Regev05}]\label{def:decisional-lwe} Let $n,m \in \N$ be integers, let $q\geq 2$ be a modulus and let $\alpha \in (0,1)$ be a noise ratio parameter. The (decisional) Learning with Errors $(\LWE_{n,q,\alpha q}^m)$ problem is to distinguish between the following samples
$$
(\vec A \rand \Z_q^{n \times m},\vec s^\intercal \vec A+ \vec e^\intercal \Mod{q}) \quad \text{ and } \quad (\vec A \rand \Z_q^{n \times m},\vec u \rand \Z_q^m),\,\,
$$
where $\vec s \rand  \Z_q^n$ is a uniformly random vector and where $\vec e \sim D_{\Z^m,\alpha q}$ is a discrete Gaussian error vector. We rely on the quantum $\LWE_{n,q,\alpha q}^m$ assumption which states that the samples above are computationally indistinguishable for any $\QPT$ algorithm.

\end{definition}

As shown in \cite{Regev05}, the $\LWE_{n,q,\alpha q}^m$ problem with parameter $\alpha q \geq 2 \sqrt{n}$ is at least as hard as approximating the shortest independent vector problem $(\mathsf{SIVP})$ to within a factor of $\gamma = \widetilde{O}(n / \alpha)$ in worst case lattices of dimension $n$. In this work we assume the subexponential hardness of $\LWE_{n,q,\alpha q}^m$ which relies on the worst case hardness of approximating short vector problems in
lattices to within a subexponential factor. 
We assume that the $\LWE_{n,q,\alpha q}^m$ problem, for $m=\Omega(n \log q)$, $q=2^{o(n)}$, $\alpha=1/2^{o(n)}$,  is hard against quantum adversaries running in time $\poly(q)$. We note that this parameter regime implies $\SIS_{n,q,\beta}^m$~\cite{cryptoeprint:2009/285}.

\paragraph{Trapdoors for lattices.} We use the following \emph{trapdoor} property for the $\LWE$ problem.

\begin{theorem}[\cite{cryptoeprint:2011:501}, Theorem 5.1]\label{thm:gen-trap} Let  $n,m \in \N$ and $q \in \N$ be a prime with $m = \Omega(n \log q)$. There exists a randomized algorithms with the following properties:
\begin{itemize}
    \item $\GenTrap(1^n,1^m,q)$: on input $1^n, 1^m$ and $q$, returns a matrix $\vec A \in \Z_q^{n \times m}$ and a trapdoor $\mathsf{td}_{\vec A}$ such that the distribution of $\vec A$ is negligibly (in the parameter $n$) close to uniform.
    
    \item $\Invert(\vec A,\mathsf{td}_{\vec A},\vec b)$: on input $\vec A$, $\mathsf{td}_{\vec A}$ and $\vec b=\vec s^\intercal \cdot \vec A + \vec e^\intercal \Mod{q}$, where $\| \vec e \| \leq q/(C_T \sqrt{n \log q})$ and $C_T >0$ is a universal constant, returns $\vec s$ and $\vec e$ with overwhelming probability over $(\vec A,\mathsf{td}_{\vec A}) \leftarrow \GenTrap(1^n,1^m,q)$.
    \end{itemize}
\end{theorem}

\section{Quantum Discrete Gaussian Sampling for $q$-ary Lattices}\label{sec:qdgs}

In this section, we review some basic facts about Gaussian superpositions and present our \emph{quantum discrete Gaussian sampler} which is used to revoke the decryption keys for our schemes.

\subsection{Gaussian 
Superpositions}

In this section, we review some basic facts about \emph{Gaussian superpositions}.
Given $q \in \N$, $m \in \N$ and $\sqrt{8m}< \sigma <q/\sqrt{8m}$, we consider Gaussian superpositions over $\Z^m \cap (-\frac{q}{2},\frac{q}{2}]^m$ of the form
    $$
 \ket{\psi} =    \sum_{\vec x \in \Z_q^m} \rho_\sigma(\vec x) \ket{\vec x}.
    $$
Note that the state $\ket{\psi}$ is not normalized for convenience and ease of notation. The tail bound in \Cref{lem:gaussian-tails} implies that (the normalized variant of) $\ket{\psi}$ is within negligible trace distance of a \emph{truncated} discrete Gaussian superposition $ \ket{\tilde{\psi}}$
with support $\{\vec x \in \Z_q^m : \|\vec x\| \leq \sigma \sqrt{\frac{m}{2}}\}$, where
$$
\ket{\tilde{\psi}} = \sum_{\vec x \in \Z_q^m} \sqrt{D_{\Z_q^m,\frac{\sigma}{\sqrt{2}}}(\vec x)} \, \ket{\vec x}
= 
\left(\sum_{\vec z \in \Z_q^m,\|\vec z\| \leq \sigma \sqrt{\frac{m}{2}} } \rho_{\frac{\sigma}{\sqrt{2}}}(\vec z) \right)^{-\frac{1}{2}}\sum_{\vec x \in \Z_q^m : \|\vec x\| \leq \sigma \sqrt{\frac{m}{2}}}
\rho_\sigma(\vec x) \ket{\vec x}.
$$
In this work, we consider Gaussian superpositions with parameter $\sigma = \Omega(\sqrt{m})$ which can be efficiently implemented using standard quantum state preparation techniques; for example using \emph{quantum rejection sampling} and the \emph{Grover-Rudolph algorithm}~\cite{Grover2002CreatingST,Regev05,Brakerski18,brakerski2021cryptographic}.

\paragraph{Gaussian coset states.}

Our key-revocable encryption schemes in \Cref{sec:dual-regev} and \Cref{sec:dual-GSW} rely on Gaussian superpositions over $\vec x \in \Z_q^m$ subject to a constraint of the form $\vec A \cdot \vec x = \vec y \Mod{q}$, for some matrix $\vec A \in \Z_q^{n \times m}$ and image $\vec y \in \Z_q^n$. In Algorithm \ref{alg:GenGauss}, we give a procedure called $\mathsf{GenGauss}$ that, on input $\vec A$ and $\sigma >0$, generates a  Gaussian superposition state of the form
   $$
    \ket{\psi_{\vec y}} = \sum_{\substack{\vec x \in \Z_q^m:\\ \vec A \vec x= \vec y}} \rho_\sigma(\vec x) \ket{\vec x},
    $$
for some $\vec y \in \Z_q^n$ which is statistically close to uniform whenever $m \geq 2n\log q$ and $\sigma \geq \omega(\sqrt{\log m})$. Because $\ket{\psi_{\vec y}}$ corresponds to a (truncated) Gaussian superposition over a particular lattice coset,
$$\Lambda_q^{\vec y}(\vec A) = \{ \vec x \in \Z^m  :  \vec A \cdot \vec x = \vec y \Mod{q}\},$$
of the $q$-ary lattice $\Lambda_q^\bot(\vec A) = \{\vec x \in \Z^m: \, \vec A \cdot\vec x = \vec 0 \Mod{q}\}$, we refer to it as a \emph{Gaussian coset state}.

Finally, we recall an important property of Gaussian coset states.

\paragraph{Gaussian-collapsing hash functions.}
Unruh~\cite{cryptoeprint:2015/361} introduced the notion of collapsing hash functions in his seminal work on computationally binding quantum commitments. Informally, a hash function is called \emph{collapsing} if it is computationally difficult to distinguish between a superposition of pre-images and a single (measured) pre-image.

In recent work, Poremba~\cite{Poremba22} proposed a special variant of the collapsing property with respect to \emph{Gaussian superpositions}. Previously, Liu and Zhandry~\cite{cryptoeprint:2019/262} implicitly showed that the \emph{Ajtai} hash function $h_{\vec A}(\vec x) = \vec A \vec x \Mod{q}$ is collapsing -- and thus \emph{Gaussian-collapsing} -- via the notion of \emph{lossy functions} and by assuming the superpolynomial hardness of (decisional) $\LWE$.

We use the following result on the Gaussian-collapsing property of the Ajtai hash function.

\begin{theorem}[Gaussian-collapsing property, \cite{Poremba22}, Theorem 4]\label{thm:Gauss-collapsing}
Let $n\in \N$ and $q$ be a prime with $m \geq 2n \log q$, each parameterized by $\lambda \in \N$. Let $\sqrt{8m}< \sigma <q/\sqrt{8m}$.
Then,
the following samples are computationally indistinguishable assuming the quantum hardness of decisional $\mathsf{LWE}_{n,q,\alpha q}^m$, for any noise ratio $\alpha \in (0,1)$ with relative noise magnitude $1/\alpha= \sigma \cdot 2^{o(n)}:$
$$
\Bigg(\vec  A \rand \Z_q^{n \times m},\,\, \ket{\psi_{\vec y}}=\sum_{\substack{\vec x \in \Z_q^m\\ \vec A \vec x = \vec y}}\rho_{\sigma}(\vec x) \,\ket{\vec x}, \,\,\vec y\in \Z_q^n \Bigg)\,\, \approx_c \,\,\,\, \Bigg(\vec  A \rand \Z_q^{n \times m}, \,\,\ket{\vec x_0},\,\, \vec A \cdot \vec x_0 \,\in \Z_q^n\Bigg)
$$
where $(\ket{\psi_{\vec y}},\vec y) \leftarrow \mathsf{GenGauss}(\vec A,\sigma)$ and where $\vec x_0 \sim D_{\Z_q^m,\frac{\sigma}{\sqrt{2}}}$ is a discrete Gaussian error.
\end{theorem}

\subsection{Algorithm: \textsf{GenGauss}}\label{sec:GenGauss}

The state preparation procedure $\mathsf{GenGauss}(\vec A,\sigma)$ is defined as follows.
\ \\

    \begin{algorithm}[H]\label{alg:GenGauss}
\DontPrintSemicolon
\SetAlgoLined
\KwIn{Matrix $\vec A \in \Z_q^{n \times m}$ and parameter $\sigma = \Omega(\sqrt{m})$.}
\KwOut{Gaussian state $\ket{\psi_{\vec y}}$ and $\vec y \in \Z_q^n$.}

Prepare a Gaussian superposition in system $X$ with parameter $\sigma > 0$:
    $$
 \ket{\psi}_{XY} =    \sum_{\vec x \in \Z_q^m} \rho_\sigma(\vec x) \ket{\vec x}_X \otimes \ket{\vec 0}_Y.
    $$\;
Apply the unitary $U_{\vec A}: \ket{\vec x}\ket{\vec 0} \rightarrow \ket{\vec x} \ket{\vec A \cdot \vec x \Mod{q}}$ on systems $X$ and $Y$:
$$
 \ket{ \psi}_{XY} =   \sum_{\vec x \in \Z_q^m} \rho_\sigma(\vec x) \ket{\vec x}_X \otimes \ket{\vec A \cdot \vec x \Mod{q}}_Y.
  $$\;    
Measure system $Y$ in the computational basis, resulting in the state
    $$
    \ket{\psi_{\vec y}}_{XY} = \sum_{\substack{\vec x \in \Z_q^m:\\ \vec A \vec x= \vec y}} \rho_\sigma(\vec x) \ket{\vec x}_X \otimes \ket{\vec y}_Y.
    $$\;
Output the state $\ket{\psi_{\vec y}}$ in system $X$ and the outcome $\vec y \in \Z_q^n$ in system $Y$.
 \caption{$\mathsf{GenGauss}(\vec A,\sigma)$}
\end{algorithm}
\ \\

\subsection{Algorithm: \textsf{QSampGauss}}\label{sec:QSampGauss}

Recall that, in Algorithm \ref{alg:GenGauss}, we gave a procedure called $\mathsf{GenGauss}(\vec A,\sigma)$ that prepares a Gaussian coset state $\ket{\psi_{\vec y}}$,
for a randomly generated $\vec y \in \Z_q^n$. In general, however, generating a specific Gaussian coset state on input $(\vec A,\vec y)$ requires a \emph{short trapdoor basis} $\mathsf{td}_{\vec A}$ for the matrix $\vec A$. This task can be thought of as a quantum analogue of the \emph{discrete Gaussian sampling problem}~\cite{cryptoeprint:2007:432}, where the goal is to output a sample $\vec x \sim D_{Z^m,\sigma}$ such that $\vec A \cdot \vec x = \vec y \Mod{q}$ on input $(\vec A,\vec y)$ and $\sigma >0$.

In Algorithm \ref{alg:SampGauss}, we give a procedure called $\mathsf{QSampGauss}$ which, on input $(\vec A,\mathsf{td}_{\vec A},\vec y,\sigma)$ generates a specific Gaussian coset state $\ket{\psi_{\vec y}}$ of the form
$$
\ket{\psi_{\vec y}}=\sum_{\substack{\vec x \in \Z_q^m\\ \vec A \vec x = \vec y}}\rho_{\sigma}(\vec x) \,\ket{\vec x}.
$$
Our procedure $\mathsf{QSampGauss}$ in Algorithm \ref{alg:SampGauss} can be thought of as an explicit quantum reduction from $\ISIS_{n,q,\sigma\sqrt{m/2}}^m$ to $\LWE_{n,q,q/\sqrt{2}\sigma}^m$ which is inspired by the quantum reduction of Stehlé et al.~\cite{cryptoeprint:2009/285} which reduces $\SIS$ to $\LWE$. To obtain the aforementioned reduction, one simply needs to replace the procedure
$\mathsf{Invert}(\vec A,{\mathsf{td}_{\vec A}},\cdot)$ in Step $4$ in Algorithm \ref{alg:SampGauss} with a solver for the $\LWE_{n,q,q/\sqrt{2}\sigma}^m$ problem.

In \Cref{lem:qdgs}, we prove the correctness of Algorithm \ref{alg:SampGauss}. As a technical ingredient, we rely on a \emph{duality lemma}~\cite{Poremba22} that characterizes the Fourier transform of a Gaussian coset state
in terms of its dual state. Note that $\ket{\psi_{\vec y}}$ corresponds to a Gaussian superposition over a lattice coset,
$$\Lambda_q^{\vec y}(\vec A) = \{ \vec x \in \Z^m  :  \vec A \cdot \vec x = \vec y \Mod{q}\},$$
of the $q$-ary lattice $\Lambda_q^\bot(\vec A) = \{\vec x \in \Z^m: \, \vec A \cdot\vec x = \vec 0 \Mod{q}\}$. Here, the \emph{dual} of $\Lambda_q^\bot(\vec A)$ satisfies $q \cdot \Lambda_q^\bot(\vec A)^* =\Lambda_q(\vec A)$, where $\Lambda_q(\vec A)$ corresponds to the
lattice generated by $\vec A^\intercal$, i.e.
$$
\Lambda_q(\vec A) = \{\vec z \in \Z^m: \, \vec z = \vec A^\intercal \cdot \vec s \Mod{q}, \text{ for some } \vec s \in \Z^n\}.
$$
The following lemma relates the Fourier transform of $\ket{\psi_{\vec y}}$ with a superposition of
$\LWE$ samples with respect to a matrix $\vec A \in \Z_q^{n \times m}$ and a phase which depends on $\vec y$. In other words, the resulting state can be thought of as a superposition of Gaussian balls around random lattice vectors in $\Lambda_q(\vec A)$.

\begin{lemma}[\cite{Poremba22}, Lemma 16]\label{lem:duality}
Let $m \in \N$, $q \geq 2$ be a prime and $\sqrt{8m}< \sigma <q/\sqrt{8m}$.
Let $\vec A \in \Z_q^{n \times m}$ be a matrix whose columns generate $\Z_q^n$ and let $\vec y \in \Z_q^n$ be arbitrary. Then, the $q$-ary quantum Fourier transform of the (normalized variant of the) Gaussian coset state
$$
 \ket{\psi_{\vec y}} = \sum_{\substack{\vec x \in \Z_q^m\\ \vec A \vec x = \vec y \Mod{q}}}\rho_{\sigma}(\vec x) \ket{\vec x}
$$
is within negligible (in $m \in \N$) trace distance of the (normalized variant of the) Gaussian state
$$
 \ket{\hat\psi_{\vec y}} = \sum_{\vec s \in \Z_q^n} \sum_{\vec e \in \Z_q^m} \rho_{q/\sigma}(\vec e) \, \omega_q^{-\langle\vec s,\vec y \rangle} \ket{\vec s^\intercal \vec A + \vec e^\intercal \Mod{q}}.
$$
\end{lemma}

The procedure $\mathsf{QSampGauss}(\vec A,\mathsf{td}_{\vec A},\vec y,\sigma)$ is defined as follows.
\ \\

    \begin{algorithm}[H]\label{alg:SampGauss}
\DontPrintSemicolon
\SetAlgoLined
\KwIn{Matrix $\vec A \in \Z_q^{n \times m}$, a trapdoor $\mathsf{td}_{\vec A}$, an image $\vec y \in \Z_q^n$ and parameter $\sigma = O(\frac{q}{\sqrt{m}})$.}
\KwOut{Gaussian state $\ket{\psi_{\vec y}}$.}

Prepare the following superposition with parameter $q/\sigma > 0$:
    $$
 \sum_{\vec s \in \Z_q^n} \ket{\vec s} \otimes   \sum_{\vec e \in \Z_q^m} \rho_{q/\sigma}(\vec e) \ket{\vec e} \otimes \ket{\vec 0}
    $$\;

Apply the generalized Pauli operator $\vec Z_q^{-\vec y}$ on the first register, resulting in the state
    $$
 \sum_{\vec s \in \Z_q^n} \omega_q^{-\langle \vec s,\vec y \rangle} \ket{\vec s} \otimes   \sum_{\vec e \in \Z_q^m} \rho_{q/\sigma}(\vec e) \ket{\vec e} \otimes \ket{\vec 0}
    $$\;

Apply the unitary $U_{\vec A}: \ket{\vec s} \ket{\vec e} \ket{\vec 0} \rightarrow \ket{\vec s} \ket{\vec e} \ket{\vec s^\intercal \vec A + \vec e^\intercal \Mod{q}}$, resulting in the state
   $$
 \sum_{\vec s \in \Z_q^n} \sum_{\vec e \in \Z_q^m} \rho_{q/\sigma}(\vec e)  \,\omega_q^{-\langle \vec s,\vec y \rangle} \ket{\vec s}  \ket{\vec e} \ket{\vec s^\intercal \vec A + \vec e^\intercal \Mod{q}}
    $$\;
Coherently run $\mathsf{Invert}(\vec A,{\mathsf{td}_{\vec A}},\cdot)$ on the third register in order to uncompute the first and the second register, resulting in a state that is close in trace distance to the following state: 
    $$ \sum_{\vec s \in \Z_q^n}\sum_{\vec e \in \Z_q^m} \rho_{q/\sigma}(\vec e) \, \omega_q^{-\langle \vec s,\vec y \rangle} 
\ket{0}\ket{0} \ket{\vec s^\intercal \vec A + \vec e^\intercal \Mod{q}}
    $$\;
Discard the first two registers. Apply the (inverse) quantum Fourier transform and output the resulting state.
 \caption{$\mathsf{QSampGauss}(\vec A,\mathsf{td}_{\vec A},\vec y,\sigma)$}
\end{algorithm}
\ \\
Let us now prove the correctness of Algorithm \ref{alg:SampGauss}.

\begin{theorem}[Quantum Discrete Gaussian Sampler]
\label{lem:qdgs}
Let $n\in \N$, $q$ be a prime with $m \geq 2n \log q$ and $\sqrt{8m}< \sigma <q/\sqrt{8m}$.
Let $(\vec A,\mathsf{td}_{\vec A}) \leftarrow \GenTrap(1^n,1^m,q)$ be sampled as in \Cref{thm:gen-trap} and let $\vec y \in \Z_q^n$ be arbitrary. Then, with overwhelming probability, $\mathsf{QSampGauss}(\vec A,\mathsf{td}_{\vec A},\vec y,\sigma)$ in Algorithm \ref{alg:SampGauss} outputs a state which is within negligible trace distance of the (normalized variant of the) state,
$$
 \ket{\psi_{\vec y}} = \sum_{\substack{\vec x \in \Z_q^m\\ \vec A \vec x = \vec y \Mod{q}}}\rho_{\sigma}(\vec x) \ket{\vec x}.
$$
\end{theorem}
\begin{proof}
From \Cref{lem:full-rank} and \Cref{thm:gen-trap}, it follows that $(\vec A,\mathsf{td}_{\vec A}) \leftarrow \GenTrap(1^n,1^m,q)$ yields a matrix $\vec A \in \Z_q^{n \times m}$ be a matrix whose columns generate $\Z_q^n$ with overwhelming probability. Moreover, since $\sqrt{8m}< \sigma <q/\sqrt{8m}$, the inversion procedure
$\mathsf{Invert}(\vec A,{\mathsf{td}_{\vec A}},\cdot)$ from \Cref{thm:gen-trap} in Step $4$ in Algorithm \ref{alg:SampGauss} succeeds with overwhelming probability at generating the Gaussian state
 $$  \ket{\hat\psi_{\vec y}} =\sum_{\vec s \in \Z_q^n}\sum_{\vec e \in \Z_q^m} \rho_{q/\sigma}(\vec e) \, \omega_q^{-\langle \vec s,\vec y \rangle} 
\ket{\vec s^\intercal \vec A + \vec e^\intercal \Mod{q}}
$$
Applying the (inverse) quantum Fourier transform $\FT_q^\dag$, the claim then follows from \Cref{lem:duality}.
\end{proof}

\section{Quantum Goldreich-Levin Theorem for Large Fields}
\label{sec:QGL}

In this section, we give a proof of a quantum Goldreich-Levin theorem for large fields $\Z_q$.

\subsection{Post-Quantum Reductions and Quantum Rewinding}

We first review some recent work by Bitansky, Brakerski and Kalai~\cite{BBK22} that enables us to convert a wide range of classical reductions into post-quantum reductions (which allow for quantum auxiliary input) in a constructive manner. We first review some basic terminology from~\cite{BBK22}.

Let $\lambda \in \N$ be a parameter. A \emph{non-interactive assumption} $\mathsf{P} = (\mathsf{G},\mathsf{V},c)$ with respect to a set of polynomials $d(\lambda),n(\lambda)$ and $m(\lambda)$ is characterized as follows:
\begin{itemize}
    \item The generator $\mathsf{G}$ takes as input $1^\lambda$ and $r \in \bit^d$, and returns $x \in \bit^n$.
    \item The verifier $\mathsf{V}$ takes as input $1^\lambda$ and $(r,y) \in \bit^d \times \bit^m$, and returns a single bit output.
    \item $c(\lambda)$ is the threshold associated with the assumption.
\end{itemize}
Given a (possibly randomized) \emph{solver}, we characterize the \emph{advantage} in solving an assumption $\mathsf{P}$ in terms of the absolute distance between the solving probability (or, \emph{value}) and the threshold $c$; for example, for a \emph{decision assumption} $\mathsf{P}$ (with $m=1$) we characterize the value in solving $\mathsf{P}$ in terms of $\frac{1}{2} + \eps$, where the threshold is given by $c(\lambda) = \frac{1}{2}$ and $\eps > 0$ is corresponds to the \emph{advantage}.
We say that
a reduction is \emph{black-box} if it is oblivious to the
representation and inner workings of the solver that is being used. Moreover, we say that a reduction is \emph{non-adaptive} if all queries to the solver are known ahead of time. 
\par We use the following theorem.

\begin{theorem}[\cite{BBK22}, adapted from Theorem 7.1] \label{thm:BBK} 
    Let $c \in \mathbb{R}$. Suppose that there exists a classical reduction from solving a non-interactive assumption  ${\sf Q}$ to solving a non-interactive assumption ${\sf P}$ such that the following holds: if the ${\sf P}$-solver has advantage $\eps > 0$ then the ${\sf Q}$-solver has advantage $c$ (independent of $\eps$) with running time $\poly(1/\eps,c,\secparam)$. 
    \par Then, there exists a quantum reduction from solving $\mathsf{Q}$ to quantumly solving $\mathsf{P}$ such that the following holds: if the quantum ${\sf P}$-solver (with non-uniform quantum advice) has an advantage given by $\eps > 0$, then the ${\sf Q}$-solver has advantage $c$ (the same as the classical reduction) with running time $\poly(1/\eps,c,\secparam)$. 
\end{theorem}

\subsection{Goldreich-Levin Theorems for Large Fields}

The following result is implicit in the work of Dodis et al.~\cite{Dodis10}.

\begin{theorem}[Classical Goldreich-Levin Theorem for Finite Fields,
\cite{Dodis10}, Theorem 1]\label{thm:GL}
Let $q$ be a prime and $m \in \N$. Let $H = \{\vec x\in\Z_q^m \, : \, \|\vec x\| \leq \sigma \sqrt{m} \}$ be a subset of $ \Z_q^m$, for some $\sigma >0$. Let $f: H \rightarrow \bit^*$ be any (possibly randomized) function. Suppose there exists a distinguisher $\algo D$ that runs in time $T(\algo D)$ and has the property that
\begin{align*}
\vline \, \Pr \left[
 \mathcal{D}(\vec u,  \vec u^\intercal \vec x,\aux) =1
 \, : \, 
 \substack{
\vec u \rand \Z_q^m\\
\vec x \sim D_{\Z_q^m,\sigma}\\
\aux \leftarrow f(\vec x)
}
    \right] - \Pr \left[
 \mathcal{D}(\vec u, r,\aux) =1
 \, : \,
\substack{
\vec u \rand \Z_q^m, \, r \rand \Z_q\\
\vec x \sim D_{\Z_q^m,\sigma}\\
\aux \leftarrow f(\vec x)
}
\right]
\,\vline = \eps.
\end{align*}
Then, there exists a (classical) non-adaptive black-box extractor $\algo E$ whose running time is given by $T(\algo E) = T(\algo D) \cdot \poly(m,\sigma,1/\eps)$ and succeeds with probability at least
$$
\Pr \Big[
 \mathcal{E}\big(\aux\big) = \vec x \, : \, \aux \leftarrow f(\vec x)
    \Big] \geq  \frac{\eps^3}{512 \cdot m \cdot q^2}.
$$
\end{theorem}

Using the constructive post-quantum reduction from \Cref{thm:BBK}, we can convert \Cref{thm:GL} into a quantum Goldreich-Levin Theorem for finite fields, and obtain the following.

\begin{theorem}[Quantum Goldreich-Levin Theorem for Large Fields]\label{thm:QGL}
Let $q$ be a prime and $m \in \N$. Let $H = \{\vec x\in\Z_q^m \, : \, \|\vec x\| \leq \sigma \sqrt{m} \}$ be a subset of $ \Z_q^m$, for some $\sigma >0$. Let $\Phi: \algo L(\algo H_q^m) \rightarrow \algo L(\algo H_{\textsc{Aux}})$ be any $\CPTP$ map with auxiliary system $\algo H_{\textsc{Aux}}$. Suppose there exists a quantum distinguisher $\algo D$ that runs in time $T(\algo D)$ and has the property that
\begin{align*}
\vline \, \Pr \left[
 \mathcal{D}(\vec u,  \vec u^\intercal \vec x,\aux) =1
 \, : \, 
\substack{
\vec u \rand \Z_q^m\\
\vec x \sim D_{\Z_q^m,\sigma}\\
\aux\leftarrow \Phi(\proj{\vec x})
}
    \right] - \Pr \left[
 \mathcal{D}(\vec u, r,\aux) =1
 \, : \,
\substack{
\vec u \rand \Z_q^m, \, r \rand \Z_q\\
\vec x \sim D_{\Z_q^m,\sigma}\\
\aux \leftarrow \Phi(\proj{\vec x})
}
\right]
\,\vline = \eps.
\end{align*}
Then, there exists a quantum extractor $\algo E$ whose running time is given by $T(\algo E) = T(\algo D) \cdot \poly(m,\sigma,1/\eps)$ and that succeeds with probability at least
$$
\Pr \Big[ \mathcal{E}\big(\aux\big) = \vec x \, : \, \aux \leftarrow \Phi(\proj{\vec x})
    \Big] \geq  \poly(\eps,1/m,1/\sigma,1/q).
$$
\end{theorem}
\begin{proof}
The proof follows immediately by combining \Cref{thm:GL} and \Cref{thm:BBK}.
\end{proof}

\section{Definition: Key-Revocable Public-Key Encryption}

Let us now give a formal definition of key-revocable public-key encryption schemes.

\begin{definition}[Key-Revocable Public-Key Encryption]\label{def:key-revocable-PKE}
A key-revocable public-key encryption scheme consists of efficient algorithms $\left(\setup,\enc,\dec,\delete\right)$, where $\enc$ is a $\mathsf{PPT}$ algorithm and $\setup,\dec$ and $\delete$ are $\QPT$ algorithms defined as follows:
\begin{itemize}
    \item $\setup(1^{\secparam})$: given as input a security parameter $\secparam$, output a public key $\pk$, a master secret key $\msk$ and a quantum decryption key $\sk$. 
    \item $\enc(\pk,x)$: given a public key $\pk$ and plaintext $x \in \{0,1\}^{\ell}$, output a ciphertext $\ct$. 
    \item $\dec(\sk,\ct)$: given a decryption key $\sk$ and ciphertext $\ct$, output a message $y$. 
    \item $\delete\left(\pk,\msk,\sigma\right)$: given as input a master secret key $\msk$, a public key $\pk$ and quantum state $\sigma$, output $\valid$ or $\invalid$. 
\end{itemize}
\end{definition}

\paragraph{Correctness of Decryption.} For every $x \in \{0,1\}^{\ell}$, the following holds: 
$$\prob\left[ x \leftarrow \dec(\sk,\ct)\ :\ \substack{(\pk,\msk,\sk) \leftarrow \setup(1^{\secparam})\\ \ \\ \ct \leftarrow \enc(\pk,x)} \right] \geq 1 - \nu(\secparam),$$
where $\nu(\cdot)$ is a negligible function. 

\paragraph{Correctness of Revocation.} The following holds:
$$\prob\left[ \valid \leftarrow \delete\left(\pk,\msk,\sk \right)\ :\ (\pk,\msk,\sk) \leftarrow \setup(1^{\secparam}) \right] \geq 1 - \nu(\secparam),$$
where $\nu(\cdot)$ is a negligible function. 

\begin{remark}\label{remark:gentle-measurement}
Using the well-known ``Almost As Good As New Lemma'' (\Cref{lem:almost}), the procedure $\dec$ can be purified to obtain another quantum circuit $\widetilde{\dec}$ such that $\widetilde{\dec}(\sk,\ct)$ yields $(x,\sk')$ with probability at least $1-\nu(\secparam)$ and moreover, $\ct \leftarrow \enc(\pk,x)$ and $\TD(\sk',\sk) \leq \nu'(\secparam)$ with $\nu'(\secparam)$ is another negligible function. 
\end{remark}

\subsection{Security Definition}

Our security definition for key-revocable public-key encryption is as follows.

\begin{figure}[!htb]
   \begin{center} 
   \begin{tabular}{|p{14cm}|}
    \hline 
\begin{center}
\underline{$\expt_{\Sigma,\algo A}\left( 1^{\secparam},b \right)$}: 
\end{center}
\noindent {\bf Initialization Phase}:
\begin{itemize}
    \item The challenger runs $(\pk,\msk,\sk) \leftarrow \setup(1^{\secparam})$ and sends $(\pk,\sk)$ to $\adversary$. 
\end{itemize}
\noindent {\bf Revocation Phase}:
\begin{itemize}
    \item The challenger sends the message {\texttt{REVOKE}} to $\adversary$. 
    \item The adversary $\adversary$ returns a state $\sigma$.
    \item The challenger aborts if $\revoke(\pk,\msk,\sigma)$ outputs $\invalid$.
\end{itemize}
\noindent {\bf Guessing Phase}:
\begin{itemize}
    \item $\adversary$ submits a plaintext $x \in \{0,1\}^{\ell}$ to the challenger.
    \item If $b=0$: The challenger sends $\ct \leftarrow \enc(\pk,x)$ to $\adversary$. Else, if $b=1$, the challenger sends $\ct \xleftarrow{\$} {\cal C}$, where ${\cal C}$ is the ciphertext space of $\ell$ bit messages. 
    \item Output $b_{\adversary}$ if the output of $\adversary$ is $b_{\adversary}$.  
\end{itemize}
\ \\ 
\hline
\end{tabular}
    \caption{Security Experiment}
    \label{fig:securityexpt}
    \end{center}
\end{figure}

\begin{definition}
\label{def:krpke:security}
A key-revocable public-key encryption scheme $\Sigma =\left(\setup,\enc,\dec,\delete\right)$ is $(\epsilon,\delta)$-secure if, for every $\QPT$ adversary $\adversary$ with 
$$\prob[\invalid \leftarrow \expt_{\Sigma,\algo A}(1^\lambda,b)] \leq \delta(\secparam)
$$
for $b \in \bit$, it holds that
$$ \left| \Pr\left[ 1 \leftarrow \expt_{\Sigma,\algo A}(1^{\secparam},0) \right] - \Pr\left[ 1 \leftarrow \expt_{\Sigma,\algo A}(1^{\secparam},1) \right]\right|\leq  \eps(\secparam),$$
where $\expt_{\Sigma,\algo A}(1^{\secparam},b)$ is as defined in~\Cref{fig:securityexpt}. If $\delta(\secparam) = 1-1/\poly(\lambda)$ and $\eps(\secparam)=\negl(\lambda)$, we simply say the key-revocable public-key encryption scheme is secure. 
\end{definition}

\begin{remark}
    Our security definition is similar to the one proposed by Agrawal et al.~\cite{AKNYY23} in the context of public-key encryption with secure leasing. 
\end{remark}

\subsection{Key-Revocable Public-Key Fully Homomorphic Encryption}
A key-revocable public-key fully homomorphic encryption scheme defined for a class of functions ${\cal F}$, in addition to $(\keygen,\enc,\dec,\delete)$, consists of the following \textsf{PPT} algorithm:
\begin{itemize}
    \item $\eval(\pk,f,\ct)$: on input a public key $\pk$, function $f \in {\cal F}$, ciphertext $\ct$, outputs another ciphertext $\ct'$. 
\end{itemize}

\begin{remark}
Sometimes we allow $\keygen$ to additionally take as input different parameters associated with the implementations of the functions in ${\cal F}$. For example, we allow $\keygen$ to take as input a parameter $L$ in such a way that all the parameters in the system depend on $L$ and moreover, the homomorphic evaluation is only supported on circuits (in ${\cal F}$) of depth at most $L$. 
\end{remark}

\paragraph{Correctness of Evaluation and Decryption.}  For every $f \in {\cal F}$ with $\ell$-bit inputs, every $x \in \{0,1\}^{\ell}$, the following holds: 
$$\prob\left[ f(x) \leftarrow \dec(\sk,\ct')\ :\ \substack{(\pk,\msk,\sk) \leftarrow \setup(1^{\secparam})\\ \ \\ \ct \leftarrow \enc(\pk,x)\\ \ \\ \ct' \leftarrow \eval(\pk,f,\ct)} \right] \geq 1 - \nu(\secparam),$$
where $\nu(\cdot)$ is a negligible function. 

\paragraph{Correctness of Revocation.} Defined as before.

\paragraph{Security.} Defined as before (\Cref{def:krpke:security}).

\section{Key-Revocable Dual-Regev Encryption}\label{sec:dual-regev}
\noindent 
In this section, we present a construction of key-revocable public-key encryption from learnin with errors. Our construction essentially involves making the Dual Regev public-key encryption of Gentry, Peikert and Vaikuntanathan~\cite{cryptoeprint:2007:432} key revocable. 

\subsection{Construction}

We define our Dual-Regev construction below.

\begin{construction}[Key-Revocable Dual-Regev Encryption]\label{cons:dual-regev} Let $n \in \N$ be the security parameter and $m \in \N$. Let $q \geq 2$ be a prime and let $\alpha,\beta,\sigma >0$ be parameters. The key-revocable public-key scheme $\mathsf{RevDual} = (\keygen,\enc,\dec,\revoke)$ consists of the following $\QPT$ algorithms:
\begin{itemize}
\item $\keygen(1^\lambda) \rightarrow (\pk,\sk,\msk):$ sample $(\vec A \in \Z_q^{n \times m},\mathsf{td}_{\vec A}) \leftarrow \GenTrap(1^n,1^m,q)$ and generate a Gaussian superposition $(\ket{\psi_{\vec y}}, \vec y) \leftarrow \mathsf{GenGauss}(\vec A,\sigma)$ with
$$
\ket{\psi_{\vec y}} \,\,= \sum_{\substack{\vec x \in \Z_q^m\\ \vec A \vec x = \vec y}}\rho_{\sigma}(\vec x)\,\ket{\vec x},
$$
for some $\vec y \in \Z_q^n$. Output $\pk = (\vec A,\vec y)$, $\sk = \ket{\psi_{\vec y}}$ and $\msk = \mathsf{td}_{\vec A}$.
\item $\Enc(\pk,\mu) \rightarrow \ct:$ to encrypt a bit $\mu \in \bit$, sample a random vector $\vec s \rand \Z_q^n$ and errors $\vec e \sim D_{\Z^{m},\,\alpha q}$ and $e' \sim D_{\Z,\,\beta q}$, and output the ciphertext pair
$$
\ct = \left(\vec s^\intercal \vec A + \vec e^\intercal \Mod{q}, \vec s^\intercal \vec y + e' + \mu \cdot \lfloor \frac{q}{2} \rfloor \Mod{q} \right) \in \Z_q^m \times \Z_q.
$$
\item $\Dec(\sk,\ct) \rightarrow \bit:$ to decrypt $\ct$, 
apply the unitary $U: \ket{\vec x}\ket{0} \rightarrow \ket{\vec x}\ket{ \ct\cdot (-\vec x,1)^\intercal}$ on input $\ket{\psi_{\vec y}}\ket{0}$, where $\sk=\ket{\psi_{\vec y}}$, and measure the second register in the computational basis. Output $0$, if the measurement outcome
is closer to $0$ than to $\lfloor \frac{q}{2} \rfloor$,
and output $1$, otherwise. 

\item $\revoke(\msk,\pk,\rho) \rightarrow \{\top,\bot\}$: on input $\mathsf{td}_{\vec A} \leftarrow \msk$ and $(\vec A,\vec y) \leftarrow \pk$, apply the measurement $\{\ketbra{\psi_{\vec y}}{\psi_{\vec y}},I  - \ketbra{\psi_{\vec y}}{\psi_{\vec y}}\}$ onto the state $\rho$ using the procedure $\mathsf{QSampGauss}(\vec A,\mathsf{td}_{\vec A},\vec y,\sigma)$ in Algorithm \ref{alg:SampGauss}. Output $\top$, if the measurement is successful, and output $\bot$ otherwise.
\end{itemize}
\end{construction}

\paragraph{Correctness of Decryption.} Follows from the correctness of Dual-Regev public-key encryption. 

\paragraph{Correctness of Revocation.} This follows from~\Cref{lem:qdgs}.  

Let us now prove the security of our key-revocable Dual-Regev scheme in \Cref{cons:dual-regev}. Our first result concerns $(\negl(\lambda),\negl(\lambda))$-security, i.e., we assume that revocation succeeds with overwhelming probability.

\begin{theorem}\label{thm:security-Dual-Regev-high-revoke}
Let $n\in \N$ and $q$ be a prime modulus with $q=2^{o(n)}$ and $m \geq 2n \log q$, each parameterized by the security parameter $\lambda \in \N$. Let $\sqrt{8m}< \sigma <q/\sqrt{8m}$ and let $\alpha,\beta \in (0,1)$ be noise ratios chosen such that $\beta/\alpha =2^{o(n)}$ and $1/\alpha= 2^{o(n)} \cdot \sigma$.
Then, assuming the subexponential hardness of the $\mathsf{LWE}_{n,q,\alpha q}^m$
and $\mathsf{SIS}_{n,q,\sigma\sqrt{2m}}^m$ problems, the scheme
$\mathsf{RevDual} = (\keygen,\enc,\dec,\revoke)$ in \Cref{cons:dual-regev}
is a $(\negl(\lambda),\negl(\lambda))$-secure key-revocable public-key encryption scheme according to \Cref{def:krpke:security}.
\end{theorem}

To prove the stronger variant of $(\negl(\lambda),1-1/\poly(\lambda))$-security, i.e., where we do not make any requirements about the success probability of revocation, we need to invoke \Cref{thm:search-to-decision}. 

\begin{theorem}\label{thm:security-Dual-Regev}
Let $n\in \N$ and $q$ be a prime modulus with $q=2^{o(n)}$ and $m \geq 2n \log q$, each parameterized by the security parameter $\lambda \in \N$. Let $\sqrt{8m}< \sigma <q/\sqrt{8m}$ and let $\alpha,\beta \in (0,1)$ be noise ratios chosen such that $\beta/\alpha =2^{o(n)}$ and $1/\alpha= 2^{o(n)} \cdot \sigma$.
Assuming \Cref{thm:search-to-decision}, the scheme
$\mathsf{RevDual} = (\keygen,\enc,\dec,\revoke)$ in \Cref{cons:dual-regev}
is a $(\negl(\lambda),1-1/\poly(\lambda))$-secure key-revocable public-key encryption scheme according to \Cref{def:krpke:security}.
\end{theorem}

\paragraph{Guide for proving~\Cref{thm:security-Dual-Regev-high-revoke} and \Cref{thm:security-Dual-Regev}.} 
\begin{itemize}

    \item First, we prove a technical lemma (\Cref{lem:remove-revoke}) that helps us remove the condition that revocation succeeds when analyzing the advantage of a distinguisher. Our proof uses \emph{projective implementations} which allow us to estimate the success probability of quantum programs.
    \item The next step towards proving~\Cref{thm:security-Dual-Regev-high-revoke} is a search-to-decision reduction with quantum auxiliary input for the Dual-Regev scheme (\Cref{thm:search-to-decision-without-revoke}). Here, we show how to extract a short vector mapping $\bfA$ to $\bfy$ from an efficient adversary who has a non-negligible distinguishing advantage at distinguishing Dual-Regev ciphertexts from uniform. 

    \item Next, we state the \emph{Simultaneous Dual-Regev Extraction} conjecture in \Cref{thm:search-to-decision}, which is a strengthening of our search-to-decision reduction in \Cref{thm:search-to-decision-without-revoke}. Informally, it says that extraction of a short vector mapping $\bfA$ to $\bfy$ succeeds, even if we apply revocation on a separate register. We prove that \Cref{thm:search-to-decision} holds assuming $\LWE$/$\SIS$ in the special case when revocation succeeds with overwhelming probability. This is captured by \Cref{thm:search-to-decision-revoke-1}.
    
    \item Next, we prove technical lemma which exploits the search-to-reduction to extract two \emph{distinct} short vectors mapping $\bfA$ to $\bfy$. This is proven in~\Cref{sec:distinct-pre-images}. 
    \item Finally, we put all the pieces together in~\Cref{sec:dual-regev-proof} and show how to use the result from~\Cref{sec:distinct-pre-images} in order to break the $\SIS$ assumption. 
\end{itemize}

\subsection{Threshold Implementations}

In this section, we prove \Cref{lem:remove-revoke}. This is a useful ingredient in the security proofs behind our key-revocable Dual-Regev encryption scheme.

First, we review some recent techniques that allow us to measure the success probability of \emph{quantum programs}. In the classical setting, this task is fairly straightforward: simply execute a given program on samples from a \emph{test distribution}, and check how many times the program succeeds. Using standard concentration inequalities, one can then estimate the success probability to inverse polymonial precision. In the quantum realm, however, this task is non-trivial if the quantum program is run with respect to quantum auxiliary inputs. 

Inspired by the work of Marriott and Watrous~\cite{marriott2005quantum},
Zhandry~\cite{cryptoeprint:2020/1191} introduced the notion of projective implementations which allow us to accomplish this task efficiently. Below, we introduce some relevant definitions and results from the original work of Zhandry~\cite{cryptoeprint:2020/1191}, as well as subsequent follow-up works~\cite{ALLZZ21,CLLZ21,cryptoeprint:2022/884}. First, we discuss \emph{inefficient} measurement techniques for measuring the success probability of a quantum program. Next, we move onto \emph{efficient} measurement techniques that allow us to obtain such estimates approximately.

\paragraph{Inefficient measurements.}

Suppose we we have quantum program, say consisting of a quantum circuit and some quantum auxiliary input, and we wish to estimate its success probability. A natural starting point is to consider a two-outcome POVM $\algo P = (P,Q)$ over the two outcomes $0$ (success) and $1$ (failure). Zhandry~\cite{cryptoeprint:2020/1191} showed that for any such $\algo P$, there exists a natural projective measurement (called a \emph{projective implementation}) such that the post-measurement state corresponds precisely to an eigenvector of $P$. Moreover, there exists a projective measurement $\algo E$ that \emph{measures} the success probability with respect to $\algo P$ on some auxiliary input state; specifically,
\begin{itemize}
    \item $\algo E$ outputs a probability $p \in [0,1]$ (i.e., a real number) from the set of eigenvalues of $P$.
    \item The post-measurement state after obtaining outcome $p$ corresponds to an eigenvector of $P$ with eigenvalue $p$; similarly, it is an eigenvector of $Q = I - P$ with eigenvalue $1-p$.
\end{itemize}
The measurement $\algo E$ is projective in the following sense: whenever we apply the same measurement $\algo E$ on the post-measurement state, we obtain precisely the same outcome. The following theorem is implicit in \cite[Lemma 1]{cryptoeprint:2020/1191}, but we rely on the presentation from \cite[Theorem 2.5]{cryptoeprint:2022/884}.

\begin{theorem}[Projective implementation]\label{thm:pi} Let $\algo P = (P,Q)$ be a two-outcome POVM and let $\algo D$ be the distribution over the eigenvalues of $P$. Then, there exists a projective measurement $\algo E = \{E_p\}_{p \in \algo D}$ with index set $\algo D$ such that: for every quantum state $\rho$, where we let $\rho_p = E_p \rho E_p$ denote the sub-normalized post-measurement state after measuring $\rho$ via $E_p$, it holds that
\begin{itemize}
    \item For every $p \in \algo D$, the state $\rho_p$ is an eigenvector of $P$ with eigenvaue $p$, and
    \item the probability of $\rho$ when measured with respect to $P$ is equal to $\Tr[P\rho] = \sum_{p \in \algo D}\Tr[P\rho_p]$.
\end{itemize}
    
\end{theorem}

\begin{remark}
Suppose that $\algo P = (P,Q)$ is a two-outcome POVM and that $P$ has an eigenbasis $\{\ket{\psi_i}\}$ with associated eigenvalues $\{\lambda_i\}$. Because $P$ and $Q$ commute, they share a common eigenbasis. In this case, there exists a natural measurement $\algo E$ that corresponds to a projective implementation of the POVM $\algo P$; namely, for any input $\ket{\psi}$, which can express as $\ket{\psi} = \sum_i \alpha_i \ket{\psi_i}$, the measurement $\algo E = \{ E_{\lambda_i}\}$ will result in outcome $\lambda_i$ and a leftover eigenstate $\ket{\psi_i}$ with probabiliy $|\alpha_i|^2$.
\end{remark}

Next, we use a generalization of projective implementations introduced in~\cite{ALLZZ21}. Rather than estimating the success probability directly, we can instead measure whether it is above or below a certain threshold. This gives rise to the following notion of \emph{threshold implementations}.

\begin{theorem}[Threshold implementation]\label{def:TI}
Let $\gamma \in (0,1)$ be a parameter and let $\algo P = (P,Q)$ be a two-outcome POVM, where $P$ has an eigenbasis $\{\ket{\psi_i}\}$ with associated eigenvalues $\{\lambda_i\}$. Then, there exists a projective threshold implementation $(\mathsf{TI}_\gamma(\algo P), \mathrm{I}- \mathsf{TI}_\gamma(\algo P))$ such that
\begin{itemize}
    \item $\mathsf{TI}_\gamma(\algo P)$ projects a quantum state into the subspace spanned by $\{\ket{\psi_i}\}$ whose eigenvalues $\lambda_i$ satisfy the property $\lambda_i \leq \gamma$. 

    \item $\mathrm{I}-\mathsf{TI}_\gamma(\algo P)$ projects a quantum state into the subspace spanned by $\{\ket{\psi_i}\}$ whose eigenvalues $\lambda_i$ satisfy the property $\lambda_i > \gamma$. 
\end{itemize}
\end{theorem}
The proof of the theorem above follows directly from \Cref{thm:pi} by considering the projective measurements $\mathsf{TI}_\gamma(\algo P) = \sum_{i: \lambda_i\leq \gamma} E_{\lambda_i}$ and $\mathrm{I} -\mathsf{TI}_\gamma(\algo P) = \mathrm{I} -\sum_{i: \lambda_i > \gamma} E_{\lambda_i}$.

Finally, we also use the following \emph{symmetric} variant of threshold implementations which were considered in \cite[Theorem 2.6]{cryptoeprint:2022/884}. Here, the projective measurement determines whether the success probability is either close to $1/2$ or far from $1/2$.

\begin{theorem}[Symmetric threshold implementation]
Let $\gamma \in (0,1/2)$ be a parameter and let $\algo P = (P,Q)$ be a two-outcome POVM, where $P$ has an eigenbasis $\{\ket{\psi_i}\}$ with associated eigenvalues $\{\lambda_i\}$. Then, there exists a projective threshold implementation $(\mathsf{STI}_\gamma(\algo P), \mathrm{I}- \mathsf{STI}_\gamma(\algo P))$ such that
\begin{itemize}
    \item $\mathsf{STI}_\gamma(\algo P)$ projects a quantum state into the subspace spanned by $\{\ket{\psi_i}\}$ whose eigenvalues $\lambda_i$ satisfy the property $|\lambda_i - \frac{1}{2}| \leq \gamma$. 

    \item $\mathrm{I}-\mathsf{STI}_\gamma(\algo P)$ projects a quantum state into the subspace spanned by $\{\ket{\psi_i}\}$ whose eigenvalues $\lambda_i$ satisfy the property $|\lambda_i - \frac{1}{2}| > \gamma$. 
\end{itemize}
\end{theorem}
The proof of the theorem above follows directly from \Cref{thm:pi} by considering the projective measurements $\mathsf{STI}_\gamma(\algo P) = \sum_{i: |\lambda_i - \frac{1}{2}| \leq \gamma} E_{\lambda_i}$ and $\mathrm{I} -\mathsf{STI}_\gamma(\algo P) = \mathrm{I} -\sum_{i: |\lambda_i - \frac{1}{2}| > \gamma} E_{\lambda_i}$.

\paragraph{Efficient measurements.}

The quantum measurements we described above can, in general, not be implemented efficiently. However, Zhandry~\cite{cryptoeprint:2020/1191} showed that there exist so-called efficient \emph{approximate implementations} which allow one to obtain approximate estimates of the success probability of a quantum program.
In this section, we review some basic definitions and results that allow us to perform such measurmenents efficiently.

\begin{definition}[Mixture of projective measurements]
Let $\algo P = \{\algo P_i\}_{i \in \algo I}$ be a collection of binary outcome projective measurements $\algo P_i = (P_i,Q_i)$ over the same Hilbert space $\algo H$, and suppose that $P_i$ corresponds to outcome $1$ and $Q_i$ corresponds to outcome $0$. Let $D$ be a distribution over the the index set $\algo I$. Then, $\algo P_D = (P_D,Q_D)$ is the following mixture of pojective measurements:
$$
P_D = \sum_{i \in \algo I} \Pr[i \leftarrow D] \, P_i \quad\quad \text{ and } \quad\quad Q_D = \sum_{i \in \algo I} \Pr[i \leftarrow D] \, Q_i.
$$    
\end{definition}

The following result is adapted from~\cite[Theorem 6.2]{cryptoeprint:2020/1191} and~\cite[Corollary 1]{ALLZZ21}.

\begin{lemma}[Approximate threshold implementation]\label{lem:ati}
Let $\cP_D = (P_D,Q_D)$ be a binary outcome POVM over Hilbert space $\Hs$ that is a mixture of projective measurements over some distribution $D$. Let $\eps,\delta,\gamma \in (0,1)$. Then, there exists an efficient binary-outcome quantum algorithm $\sati_{\cP,D,\gamma}^{\epsilon, \delta}$, interpreted as the POVM element corresponding to outcome $1$,  such that the following holds:
\begin{itemize}
    \item For all quantum states $\rho$, $\Tr[\sati_{\cP,D,\gamma-\epsilon}^{\epsilon, \delta} \, \rho] \geq \Tr[\mathsf{TI}_\gamma(\algo P_D) \,\rho] - \delta$.

    \item For all quantum states $\rho$, it holds that
    $\Tr[\mathsf{TI}_{\gamma - 2 \eps}(\algo P_D) \,\rho'] \geq 1 - 2\delta$, where $\rho'$ is the post-measurement state which results from applying the measurement $\sati_{\cP,D, \gamma}^{\epsilon, \delta}$ to $\rho$.

    \item The expected running time to implement $\sati_{\cP,D, \gamma}^{\epsilon, \delta}$ is proportional to $\poly(1/\eps,\log(1/\delta))$, the time it takes to implement $P_D$, and the time it takes to sample from $D$. 
\end{itemize}
\end{lemma}

Finally, we use the following 
\emph{symmetric} version of the approximate threshold implementation \Cref{lem:ati} which is a variant of \cite[Theorem 2.8]{cryptoeprint:2022/884}.

\begin{lemma}[Symmetric approximate threshold implementation]\label{lem:sym-ati}
Let $\cP_D = (P_D,Q_D)$ be a binary outcome POVM over Hilbert space $\Hs$ that is a mixture of projective measurements over some distribution $D$. Let $\gamma \in (0,1/2)$ and $\eps \in (0,\gamma/2)$, and let $\delta\in (0,1)$. Let $D$ be a distribution. Then, there exists an efficient binary-outcome quantum algorithm $\mathsf{SATI}_{\cP,D,\gamma}^{\epsilon, \delta}$, interpreted as the POVM element corresponding to outcome $1$,  such that the following holds:
\begin{itemize}
    \item For all quantum states $\rho$, $\Tr[\mathsf{SATI}_{\cP,D,\gamma-\epsilon}^{\epsilon, \delta} \, \rho] \geq \Tr[\mathsf{STI}_\gamma(\algo P_D) \,\rho] - \delta$.

    \item For all quantum states $\rho$, it holds that
    $\Tr[\mathsf{STI}_{\gamma - 2 \eps}(\algo P_D) \,\rho'] \geq 1 - 2\delta$, where $\rho'$ is the post-measurement state which results from applying the measurement $\mathsf{SATI}_{\cP,D, \gamma}^{\epsilon, \delta}$ to $\rho$.

    \item The expected running time to implement $\mathsf{SATI}_{\cP,D, \gamma}^{\epsilon, \delta}$ is proportional to $\poly(1/\eps,\log(1/\delta))$, the time it takes to implement $P_D$, and the time it takes to sample from $D$. 
\end{itemize}
\end{lemma}

\paragraph{Useful Lemma.} We are now ready to prove an important lemma. Roughly speaking, the lemma says the following. Suppose we have a bipartite state $\rho$ on two registers $\textsc{R}$ 
and $\textsc{Aux}$ with the guarantee that conditioned on a binary outcome POVM succeeding on $\textsc{R}$, given the register $\textsc{Aux}$, a distinguisher can successfully distinguish two distributions ${\cal D}_0$ and ${\cal D}_1$. The lemma states that there is a distinguisher that can distinguish ${\cal D}_0$ and ${\cal D}_1$ {\em regardless of the outcome of the binary-outcome POVM on $\textsc{R}$}. 

\begin{lemma}\label{lem:remove-revoke}
Let $\lambda \in \N$ be a parameter and let $\rho_{\textsc{R},\textsc{Aux}}$ be a quantum state on systems $\textsc{R}$ and \textsc{Aux} of at most $\poly(\lambda)$ many qubits. Let $D_0,D_1$ be two efficiently samplable distributions with support $\algo X$. Let $\algo D$ be a $\QPT$ algorithm. Suppose that the following two properties hold:
\begin{itemize}
    \item A (possibly inefficient) two-outcome POVM $\algo M = \{M_1 ,M_0\}$ succeeds on system \textsc{R} with probability at least 
    $$
    \Tr[(M_1 \otimes I_{\textsc{Aux}}) \rho] \geq \frac{1}{p(\lambda)}
    $$ 
    for some polynomial $p(\lambda)$.

    \item the algorithm $\algo D$ succeeds at distinguishing $D_0$ from $D_1$ with advantage
    $$
    \left|\Pr\left[\algo D(x,\textsc{Aux})=b \, : \, 
    \substack{
    b \rand \bit\\
    x \sim D_b\\
    1 \leftarrow \algo M(\textsc{R})
    }\right] - \frac{1}{2}\right|\geq \frac{1}{q(\lambda)}.
    $$
    for some polynomial $q(\lambda)$ conditioned on the measurement $\algo M$ succeeding on register $\textsc{R}$. 
\end{itemize}
Then, there exists a $\QPT$ algorithm $\tilde{\algo D}$ and a polynomial $\mu(\lambda)$ such that $\tilde{\algo D}$ succeeds at distinguishing $D_0$ and $D_1$ with advantage at least $1/\mu(\lambda)$ on the reduced system alone, i.e.
    $$
   \left|\Pr\left[\tilde{\algo D}(x,\textsc{Aux})=b \, : \, 
    \substack{
    b \rand \bit\\
    x \sim D_b
    }\right] - \frac{1}{2}\right|
   \geq \frac{1}{\mu(\lambda)},
    $$
 where system \textsc{Aux} corresponds to the reduced state   $\rho_{\textsc{Aux}} = \Tr_{\textsc{R}}[\rho_{\textsc{R},\textsc{Aux}}]$.
\end{lemma}
\begin{proof}
Consider the binary outcome POVM $\algo P = ( P_{(D_0,D_1)}, Q_{(D_0,D_1)})$ with $ Q_{(D_0,D_1)} = I - P_{(D_0,D_1)}$ which is the following mixture of projective measurents such that 
$$
P_{(D_0,D_1)} = \frac{\Pi_0 + \Pi_1}{2}
$$
where $\Pi_0,\Pi_1$ are mixtures of two-outcome POVMs $\{\algo P_x\}$ that correspond to running $\algo D$ on samples $x$ from $D_0,D_1$ and system \textsc{Aux}, and then measuring whether the output is $0$ or $1$, i.e.
$$
\Pi_0 = \sum_{x \in \algo X} \Pr[x \leftarrow D_0] \, \algo P_x \quad\quad \text{ and } \quad\quad \Pi_1 = \sum_{x \in \algo X} \Pr[x \leftarrow D_1] \,\algo P_x.
$$
Let $\algo P$ have an eigenbasis $\{ \ket{\psi_i}\}$ with eigenvalues $\{ \lambda_i\}$.
Without loss of generality we can assume that $\rho_{R,\textsc{Aux}}$ is a pure state $\ket{\psi} = \sum_{i} \alpha_i \ket{\psi}$ in systems $R$ and \textsc{Aux}. Moreover, we can write $\ket{\psi}$ as
$$
\ket{\psi}=\sum_{i: \, |\lambda_i - \frac{1}{2}| \geq \frac{1}{q} } \alpha_i \ket{\psi_i} + \sum_{i: \, |\lambda_i - \frac{1}{2}| < \frac{1}{q}} \alpha_i \ket{\psi_i}.
$$
Let $\eps = 1/8p$, $\delta = 2^{-\lambda}$ and $\gamma = 1/2p$ be parameters.
Consider the following distinguisher $\tilde{\algo D}$:
\begin{itemize}
    \item Run the efficient approximate threshold implementation $\mathsf{SATI}_{\cP, (D_0,D_1),\gamma}^{\epsilon, \delta}$ from~\Cref{lem:sym-ati} on system \textsc{Aux} for the binary-outcome POVM given by $\algo P$.
    \item If the outcome of $\mathsf{SATI}_{\cP,(D_0,D_1), \gamma}^{\epsilon, \delta}$ is $1$, then run $\algo D$ on the post-measurement system \textsc{Aux}, and output whatever $\algo D$ outputs. Otherwise, output a random bit.
\end{itemize}
Let us now analyze the success probability of $\tilde{\algo D}$. Because the two-outcome POVM $\algo M$ succeeds on system $\textsc{R}$ with probability at least $\frac{1}{p}$ and because $\algo D$ succeeds with advantage at least $\frac{1}{q}$ on system \textsc{Aux} conditioned on $\algo M$ outputting $1$, we have
that $\ket{\psi}$ has weight at least $\frac{1}{p}$ on eigenvectors with eigenvalues $\lambda_i$ such that $|\lambda_i - \frac{1}{2}| \geq \frac{1}{q}$. In other words,
$$
\sum_{i: \, |\lambda_i - \frac{1}{2}| \geq \frac{1}{q}} |\alpha_i|^2 \geq \frac{1}{p}. 
$$
Therefore, the probability that $\mathsf{SATI}_{\cP,(D_0,D_1), \gamma}^{\epsilon, \delta}$ outputs $1$ on system \textsc{Aux} is at least
$$
\Tr\left[\mathsf{SATI}_{\cP,(D_0,D_1), \gamma}^{\epsilon, \delta}(\textsc{Aux})\right] \geq \frac{1}{p} - 2 \delta = O\left(1/p\right).
$$
Moreover, the post-measurement state in system $\tilde{\textsc{Aux}}$ after getting outcome $1$ has weight $1-2\delta$ on eigenvectors $\bracketsC{\ket{\psi_i}}$ such that $|\lambda_i - \frac{1}{2}| > \gamma - 2\epsilon$. Therefore, with probability at least $1-2\delta$, $\algo D$ has an advantage of at least $\gamma - 2 \epsilon$ at outputting the correct bit when run on the collapsed post-measurement system $\tilde{\textsc{Aux}}$.

However, if the measurement $\mathsf{SATI}_{\cP,(D_0,D_1), \gamma}^{\epsilon, \delta}$ on system \textsc{Aux} fails and outputs $0$, then $\tilde{\algo D}$ succeeds with probability $1/2$. Therefore, with overwhelming probability, $\tilde{\algo D}$ has advantage at least
\begin{align*}
& \vline \Pr\left[\tilde{\algo D}(x,\textsc{Aux})=b \, : \, 
    \substack{
    b \rand \bit\\
    x \sim D_b
    }\right] - \frac{1}{2} \,\vline\\
    & = \, \vline \Pr\left[\algo D(x,\tilde{\textsc{Aux}})=b \, : \, 
    \substack{
    b \rand \bit\\
    x \sim D_b
    }\right] \cdot \Tr\left[\mathsf{SATI}_{\cP,(D_0,D_1), \gamma}^{\epsilon, \delta}(\textsc{Aux})\right]\\
    &\quad\quad + \frac{1}{2} \left(1-\Tr\left[\mathsf{SATI}_{\cP,(D_0,D_1), \gamma}^{\epsilon, \delta}(\textsc{Aux})\right] \right) - \frac{1}{2} \,\vline\\
    &= \Tr\left[\mathsf{SATI}_{\cP,(D_0,D_1), \gamma}^{\epsilon, \delta}(\textsc{Aux})\right] \cdot \,\,\vline\Pr\left[\algo D(x,\tilde{\textsc{Aux}})=b \, : \, 
    \substack{
    b \rand \bit\\
    x \sim D_b
    }\right] - \frac{1}{2} \,\vline\\
    &\geq \left(1/p - 2\delta \right) \cdot (\gamma - 2\eps) \,\geq \, 1/\poly(\lambda).
\end{align*}
Finally, we remark that the running time of the distinguisher $\tilde{\algo D}$ is proportional to the running time of $\algo D$ and $\poly(1/\eps,\log(1/\delta))$, and hence it is efficient.
\end{proof}

\subsection{Simultaneous Search-to-Decision Reduction with Quantum Auxiliary Input}\label{sec:simult:srchtodecision}

Our first result concerns distinguishers with quantum auxiliary input that can distinguish between Dual-Regev samples and uniformly random samples with high probability. In \Cref{thm:search-to-decision-without-revoke}, we give a search-to-decision reduction: we show that such distinguishers can be converted into a quantum extractor that can obtain a Dual-Regev secret key with overwhelming probability. We then state a strenghtening of this extraction property (which we call \emph{Simultaneous Dual-Regev Extraction}) in \Cref{thm:search-to-decision}. Informally, this property states that extraction is possible even if additionally require that a \emph{revocation} procedure succeeds on a separate register. 

While we do not know how to prove \Cref{thm:search-to-decision} under standard assumptions, we prove that \emph{Simultaneous Dual-Regev Extraction} holds assuming $\LWE$/$\SIS$ in the special case when revocation succeeds with overwhelming probability. This is captured by \Cref{thm:search-to-decision-revoke-1}.

\paragraph{Search-to-decision reduction.} We first show the following result.

\begin{theorem}[Search-to-Decision Reduction with Quantum Auxiliary Input]\label{thm:search-to-decision-without-revoke}
Let $n\in \N$ and $q$ be a prime modulus with $q=2^{o(n)}$ and let $m \geq 2n \log q$, each parameterized by the security parameter $\lambda \in \N$. Let $\sqrt{8m}< \sigma <q/\sqrt{8m}$ and let $\alpha,\beta \in (0,1)$ be noise ratios with $\beta/\alpha = 2^{o(n)}$ and $1/\alpha= 2^{o(n)} \cdot \sigma$. Let
$\algo A = \{(\algo A_{\lambda,\vec A,\vec y},\nu_\lambda)\}_{\lambda \in \N}$
be any non-uniform quantum algorithm consisting of a family of polynomial-sized quantum circuits
$$\Bigg\{\algo A_{\lambda,\vec A,\vec y}: \algo L(\algo H_q^m \otimes \algo H_{B_\lambda}) \allowbreak\rightarrow \algo L(\algo H_{R_\lambda} \otimes \algo H_{\textsc{aux}_\lambda})\Bigg\}_{\vec A \in \Zq^{n \times m}, \,\vec y \in \Z_q^n}$$
and polynomial-sized advice states $\nu_\lambda \in \algo D(\algo H_{B_\lambda})$ which are independent of $\vec A$.
Then,
assuming the quantum hardness of the $\LWE_{n,q,\alpha q}^m$ assumption, the following holds for every $\QPT$ distinguisher $\mathcal{D}$. Suppose that there exists a function $\eps(\lambda) = 1/\poly(\lambda)$ such that
\begin{align*}
\vline \, \prob\left[1 \leftarrow \mathsf{SearchToDecisionExpt}^{\adversary,\algo D}(1^{\secparam},0)\right] - \prob\left[1 \leftarrow \mathsf{SearchToDecisionExpt}^{\adversary,\algo D}(1^{\secparam},1)\right]\vline
= \eps(\lambda).
\end{align*}
\begin{figure}[!htb]
   \begin{center} 
   \begin{tabular}{|p{14cm}|}
    \hline 
\begin{center}
\underline{$\mathsf{SearchToDecisionExpt}^{\adversary,\algo D}\left( 1^{\secparam},b \right)$}: 
\end{center}
\begin{itemize}
\item If $b=0$: output $\lwe.\dist^{\adversary,\algo D}\left( 1^{\secparam} \right)$ defined in \Cref{fig:dist-lwe-tilde}.
\item If $b=1$: output $\unif.\dist^{\adversary,\algo D}\left( 1^{\secparam} \right)$ defined in \Cref{fig:dist-unif-tilde}.
\end{itemize}
\ \\
\hline
\end{tabular}
    \caption{The experiment $\mathsf{SearchToDecisionExpt}^{\adversary,\algo D}\left( 1^{\secparam},b \right)$.}
    \label{fig:search-to-decision}
    \end{center}
\end{figure}

\noindent Then, there exists a quantum extractor $\algo E$ that takes as input $\vec A$, $\vec y$ and system $\vec{\textsc{Aux}}$ of the state $\rho_{\reg R,\vec{\textsc{Aux}}}$ and outputs a short vector in the coset $\Lambda_q^{\vec y}(\vec A)$ in time $\poly(\lambda,m,\sigma,q,1/\eps)$ such that
\begin{align*}
\Pr\left[\substack{
\mathcal{E}(\vec A,\vec y,\rho_{\vec{\textsc{Aux}}}) = \vec x\vspace{1mm}\\ \bigwedge \vspace{0.1mm}\\ {\vec x} \,\,\in\,\, \Lambda_q^{\vec y}(\vec A) \,\cap\, \algo B^m(\vec 0,\sigma \sqrt{\frac{m}{2}})} \, : \, \substack{
 \vec A \rand \Z_q^{n \times m}\\
(\ket{\psi_{\vec y}},\vec y) \leftarrow \mathsf{GenGauss}(\vec A,\sigma)\\
\rho_{\reg R,\vec{\textsc{Aux}}} \leftarrow \algo A_{\lambda,\vec A,\vec y}(\proj{\psi_{\vec y}} \otimes \nu_\lambda)
    }\right] \geq 1/\poly(\lambda).
\end{align*}
\end{theorem}

\begin{proof}
Let $\lambda \in \N$ be the security parameter and let $\algo A = \{(\algo A_{\lambda,\vec A,\vec y},\nu_\lambda)\}_{\vec A \in \Z_q^{n \times m}}$
be a non-uniform quantum algorithm. Suppose that $\mathcal{D}$ is a $\QPT$ distinguisher with advantage $\eps = 1/\poly(\lambda)$.

To prove the claim, we consider the following sequence of hybrid distributions.

\begin{description}
  
\item $\hybrid_0$: This is the distribution $\lwe.\dist^{\adversary,\algo D}\left( 1^{\secparam}\right)$ in \Cref{fig:dist-lwe-tilde}.

\begin{figure}[!htb]
   \begin{center} 
   \begin{tabular}{|p{14cm}|}
    \hline 
\begin{center}
\underline{$\lwe.\dist^{\adversary,\algo D}\left( 1^{\secparam}\right)$}: 
\end{center}
\begin{enumerate}
    \item Sample $\vec A \rand \Z_q^{n \times m}$.
   \item Generate $(\ket{\psi_{\vec y}},\vec y) \leftarrow \mathsf{GenGauss}(\vec A,\sigma)$.
   \item Generate $\rho_{R, \,\vec{\textsc{Aux}}} \leftarrow \algo A_{\lambda,\vec A,\vec y}(\proj{\psi_{\vec y}} \otimes \nu_\lambda)$.
    \item Sample $\vec s \rand \Z_q^n, \vec e \sim D_{\Z^{m},\alpha q}$ and $e'\sim D_{\Z,\beta q}$.
   
    \item Generate $\rho_{\reg R,\vec{\textsc{Aux}}} \leftarrow \algo A_{\lambda,\vec A,\vec y}(\proj{\psi_{\vec y}} \otimes \nu_\lambda)$.

    \item Run $b' \leftarrow \mathcal{D}(\vec A,\vec y,\vec s^\intercal \vec A+ \vec e^\intercal,\vec s^\intercal\vec y + e',\rho_{\textsc{Aux}})$ on the reduced state. Output $b'$.

\end{enumerate}
\ \\
\hline
\end{tabular}
    \caption{The distribution $\lwe.\dist^{\adversary,\algo D}\left( 1^{\secparam}\right)$.}
    \label{fig:dist-lwe-tilde}
    \end{center}
\end{figure}

\item $\hybrid_1$: This is the following distribution:
\begin{enumerate}
        \item Sample a random matrix $\vec A \rand \Z_q^{n \times m}$.

        \item Sample $\vec s \rand \Zq^n$, $\vec e \sim D_{\Z^m,\alpha q}$ and $e'\sim D_{\Z,\beta q}$.
        \item \rc{Sample a Gaussian vector $\vec x_0 \sim D_{\Z_q^m,\frac{\sigma}{\sqrt{2}}}$ and let $\vec y = \vec A \cdot \vec x_0 \Mod{q}$.}

        \item Run \rc{$\algo A_{\lambda,\vec A,\vec y}(\proj{\vec x_0} \otimes \nu_\lambda)$} to generate a state $\rho_{\reg R ,\vec{\textsc{aux}}}$ in systems $\reg R$ and $\vec{\textsc{aux}}$.
        
        \item Run the distinguisher $\mathcal{D}(\vec A,\vec y,\vec s^\intercal \vec A+ \vec e^\intercal,\vec s^\intercal\vec y + e',\rho_{\textsc{aux}})$ on the reduced state $\rho_{\vec{\textsc{aux}}}$.
    \end{enumerate}

\item $\hybrid_2:$ This is the following distribution:
\begin{enumerate}
        \item Sample a uniformly random matrix $\vec A \rand \Z_q^{n \times m}$.

        \item Sample $\vec s \rand \Zq^n$, $\vec e \sim D_{\Z^m,\alpha q}$ and $e'\sim D_{\Z,\beta q}$. \rc{Let $\vec u = \vec A^\intercal \vec s+ \vec e$.}
        
        \item Sample a Gaussian vector $\vec x_0 \sim D_{\Z_q^m,\frac{\sigma}{\sqrt{2}}}$ and let $\vec y = \vec A \cdot \vec x_0 \Mod{q}$.
        
        \item Run $ \algo A_{\lambda,\vec A,\vec y}(\proj{\vec x_0} \otimes \nu_\lambda)$ to generate a state $\rho_{R ,\vec{\textsc{aux}}}$ in systems $\reg R$ and $\vec{\textsc{aux}}$.
        
        \item Run the distinguisher 
        \rc{$\mathcal{D}(\vec A,\vec y,\vec u,\vec u^\intercal\vec x_0 + e',\rho_{\textsc{aux}})$} on the reduced state $\rho_{\vec{\textsc{aux}}}$.
    \end{enumerate}

 \item $\hybrid_3:$ This is the following distribution:
\begin{enumerate}
        \item Sample a uniformly random matrix $\vec A \rand \Z_q^{n \times m}$.

        \item \rc{Sample $\vec u \rand \Zq^m$} and $e'\sim D_{\Z,\beta q}$.
        
        \item Sample a Gaussian vector $\vec x_0 \sim D_{\Z_q^m,\frac{\sigma}{\sqrt{2}}}$ and let $\vec y = \vec A \cdot \vec x_0 \Mod{q}$.
        
        \item Run $ \algo A_{\lambda,\vec A,\vec y}(\proj{\vec x_0} \otimes \nu_\lambda)$ to generate a state $\rho_{R ,\vec{\textsc{aux}}}$ in systems $\reg R$ and $\vec{\textsc{aux}}$.
        
        \item Run the distinguisher 
        $\mathcal{D}(\vec A,\vec y,\vec u,\vec u^\intercal\vec x_0 + e',\rho_{\textsc{aux}})$ on the reduced state $\rho_{\vec{\textsc{aux}}}$.
    \end{enumerate}   

\item $\hybrid_4$: This is the following distribution:
\begin{enumerate}
        \item Sample a uniformly random matrix $\vec A \rand \Z_q^{n \times m}$.

        \item Sample $\vec u \rand \Zq^m$ and \rc{$r \rand \Z_q$}.
        
        \item Sample a Gaussian vector $\vec x_0 \sim D_{\Z_q^m,\frac{\sigma}{\sqrt{2}}}$ and let $\vec y = \vec A \cdot \vec x_0 \Mod{q}$.
        
        \item Run $ \algo A_{\lambda,\vec A,\vec y}(\proj{\vec x_0} \otimes \nu_\lambda)$ to generate a state $\rho_{R ,\vec{\textsc{aux}}}$ in systems $\reg R$ and $\vec{\textsc{aux}}$.
        
        \item Run the distinguisher 
        $\mathcal{D}(\vec A,\vec y,\vec u,\rc{r},\rho_{\textsc{aux}})$ on the reduced state $\rho_{\vec{\textsc{aux}}}$.
    \end{enumerate}

\item $\hybrid_5$: This is the distribution $\unif.\dist^{\adversary,\algo D}\left( 1^{\secparam}\right)$ in \Cref{fig:dist-unif-tilde}.

\begin{figure}[!htb]
   \begin{center} 
   \begin{tabular}{|p{14cm}|}
    \hline 
\begin{center}
\underline{$\unif.\dist^{\adversary,\algo D}\left( 1^{\secparam} \right)$}: 
\end{center}
\begin{enumerate}   \item Sample $\vec A \rand \Z_q^{n \times m}$.
    \item Sample $\vec u \rand \Z_q^{m}$ and $r \rand \Z_q$.
    \item Run $(\ket{\psi_{\vec y}},\vec y) \leftarrow \mathsf{GenGauss}(\vec A,\sigma)$.

    \item Generate $\rho_{R,\vec{\textsc{Aux}}} \leftarrow \algo A_{\lambda,\vec A,\vec y}(\proj{\psi_{\vec y}} \otimes \nu_\lambda)$.
    
    \item Run $b' \leftarrow \mathcal{D}(\vec A,\vec y,\vec u,r,\rho_{\textsc{Aux}})$ on the reduced state. Output $b'$. 
\end{enumerate}
\ \\
\hline
\end{tabular}
    \caption{The distribution $\unif.\dist^{\adversary,\algo D}\left( 1^{\secparam} \right)$.}
    \label{fig:dist-unif-tilde}
    \end{center}
\end{figure}
\end{description}
We now show the following:
\begin{claim} Assuming $\LWE_{n,q,\alpha q}^m$, the hybrids $\hybrid_0$ and $\hybrid_1$ are computationally indistinguishable,
$$
\hybrid_0 \, \approx_c \,\hybrid_1.
$$
\end{claim}
\begin{proof}
Here, we invoke the \emph{Gaussian-collapsing property} in \Cref{thm:Gauss-collapsing} which states that the following samples are indistinguishable under $\LWE_{n,q,\alpha q}^m$,
$$
\Big(\vec  A \rand \Z_q^{n \times m},\,\, \ket{\psi_{\vec y}}=\sum_{\substack{\vec x \in \Z_q^m\\ \vec A \vec x = \vec y}}\rho_{\sigma}(\vec x) \,\ket{\vec x}, \,\,\vec y\in \Z_q^n \Big)\,\, \approx_c \,\, \Big(\vec  A \rand \Z_q^{n \times m}, \,\,\ket{\vec x_0},\,\, \vec A \cdot \vec x_0 \,\in \Z_q^n\Big)
$$
where $(\ket{\psi_{\vec y}},\vec y) \leftarrow \mathsf{GenGauss}(\vec A,\sigma)$ and where $\vec x_0 \sim D_{\Z_q^m,\frac{\sigma}{\sqrt{2}}}$ is a sample from the discrete Gaussian distribution. Because $\algo A_{\lambda,\vec A,\vec y}$ is a family efficient quantum algorithms, this implies that
$$
\algo A_{\lambda,\vec A,\vec y}(\proj{\psi_{\vec y}} \otimes \nu_\lambda)
\quad \approx_c
\quad \algo A_{\lambda,\vec A,\vec y}(\proj{\vec x_0} \otimes \nu_\lambda),
$$
for any polynomial-sized advice state $\nu_\lambda \in \algo D(\algo H_{B_\lambda})$ which is independent of $\vec A$.
\end{proof}

\begin{claim} Hybrids $\hybrid_1$ and $\hybrid_2$ are statistically indistinguishable. In other words,
$$
\hybrid_1 \, \approx_s \,\hybrid_2.
$$
\end{claim}
\begin{proof} Here, we invoke the \emph{noise flooding} property in \Cref{lem:shifted-gaussian} to argue that $\vec e^\intercal \vec x_0 \ll e'$ holds with overwhelming probability for our choice of parameters.
Therefore, the distributions in $\hybrid_1$ and $\hybrid_2$ are computationally indistinguishable.
\end{proof}

\begin{claim} Assuming $\LWE_{n,q,\alpha q}^m$, the hybrids $\hybrid_2$ and $\hybrid_3$ are computationally indistinguishable,
$$
\hybrid_2 \, \approx_c \,\hybrid_3.
$$
\end{claim}
\begin{proof}
This follows from the $\LWE_{n,q,\alpha q}^m$ assumption since the reduction can sample $\vec x_0 \sim D_{\Z^m,\frac{\sigma}{\sqrt{2}}}$ itself and generate $\rho_{R ,\vec{\textsc{aux}}} \leftarrow \algo A_{\lambda,\vec A,\vec y}(\proj{\vec x_0} \otimes \nu_\lambda)$ on input $\vec A \in \Z_q^{n \times m}$ and $\nu_\lambda$.
\end{proof}
Finally, we show the following:
\begin{claim} Assuming $\LWE_{n,q,\alpha q}^m$, the hybrids $\hybrid_4$ and $\hybrid_5$ are computationally indistinguishable,
$$
\hybrid_4 \, \approx_c \,\hybrid_5.
$$
\end{claim}
\begin{proof}
Here, we invoke the \emph{Gaussian-collapsing property} in \Cref{thm:Gauss-collapsing} again.
\end{proof}

\par Recall that $\hybrid_0$ and $\hybrid_5$ can be distinguished with probability $\eps=1/\poly(\secparam)$. We proved that the hybrids $\hybrid_0$ and $\hybrid_3$ are computationally indistinguishable and moreover, hybrids $\hybrid_4$ and $\hybrid_5$ are computationally indistinguishable. As a consequence, it holds that hybrids $\hybrid_3$ and $\hybrid_4$ can be distinguished with probability at least $\eps - \negl(\secparam)$.
\par We leverage this to obtain a Goldreich-Levin reduction. Consider the following distinguisher.

\begin{figure}[!htb]
   \begin{center} 
   \begin{tabular}{|p{14cm}|}
    \hline 
\begin{center}
\underline{$\tilde{\algo D}\big(\vec A,\vec y,\vec u,v,\rho\big)$}: 
\end{center}
Input: $\vec A \in \Zq^{n \times m}$, $\vec y \in \Z_q^n$, $\vec u \in \Z_q^n$, $v \in \Z_q$ and $\rho \in L(\algo H_{\vec{\textsc{Aux}}})$.\\
Output: A bit $b' \in \bit$.
\vspace{3mm}\\
\textbf{Procedure:}
\begin{enumerate}
   \item Sample $e' \sim D_{\Z,\beta q}$.
   
   \item Output $b' \leftarrow \algo D\big(\vec A,\vec y,\vec u, v + e',\rho\big)$.
\end{enumerate}
\ \\
\hline
\end{tabular}
    \caption{The distinguisher $\tilde{\algo D}\big(\vec A,\vec y,\vec u,v,\rho\big)$.}
    \label{fig:tilde-dist}
    \end{center}
\end{figure}
Note that $r + e' \Mod{q}$ is uniform whenever $r \rand \Z_q$ and $e' \sim D_{\Z,\beta q}$. Therefore, our previous argument shows that there exists a negligible function $\eta$ such that:
\begin{align*}
&\vline\Pr \left[
\tilde{\mathcal{D}}(\vec A,\vec y,\vec u,\vec u^\intercal \vec x_0,\rho_{\vec{\textsc{aux}}}) = 1
 \, : \, \substack{
  \vec A \rand \Z_q^{n \times m}, \,
\vec u \rand \Z_q^{m}\\
\vec x_0 \sim D_{\Z_q^m,\frac{\sigma}{\sqrt{2}}}, \, \vec y \leftarrow \vec A \cdot \vec x_0 \Mod{q}\\
\rho_{R \,\vec{\textsc{aux}}} \leftarrow \algo A_{\lambda,\vec A,\vec y}(\proj{\vec x_0} \otimes \nu_\lambda)
    }\right] \nonumber\\
&- \Pr \left[
\tilde{\mathcal{D}}(\vec A,\vec y,\vec u,\vec r,\rho_{\vec{\textsc{aux}}}) = 1
 \, : \, \substack{
  \vec A \rand \Z_q^{n \times m}\\
\vec u \rand \Z_q^{m}, \, r \rand \Z_q\\
\vec x_0 \sim D_{\Z_q^m,\frac{\sigma}{\sqrt{2}}}, \, \vec y \leftarrow \vec A \cdot \vec x_0 \Mod{q}\\
\rho_{R \,\vec{\textsc{aux}}} \leftarrow \algo A_{\lambda,\vec A,\vec y}(\proj{\vec x_0} \otimes \nu_\lambda)
    }\right] 
\vline \geq \eps -  \eta(\lambda).   
\end{align*}

From \Cref{thm:QGL}, it follows that there exists a Goldreich-Levin extractor $\algo E$ running in time $T(\algo E) =\poly(\lambda,n,m,\sigma,q,1/\eps)$ that outputs a short vector in $\Lambda_q^{\vec y}(\vec A)$ with probability at least
$$
\Pr \left[\substack{
\mathcal{E}(\vec A,\vec y,\rho_{\vec{\textsc{Aux}}}) = \vec x\vspace{1mm}\\ \bigwedge \vspace{0.1mm}\\ {\vec x} \,\,\in\,\, \Lambda_q^{\vec y}(\vec A) \,\cap\, \algo B^m(\vec 0,\sigma \sqrt{\frac{m}{2}})}  \, : \, \substack{
 \vec A \rand \Z_q^{n \times m}, \, \vec x_0 \sim D_{\Z_q^m,\frac{\sigma}{\sqrt{2}}}\\
 \vec y \leftarrow \vec A \cdot \vec x_0 \Mod{q}\\
\rho_{\reg R,\vec{\textsc{Aux}}} \leftarrow \algo A_{\lambda,\vec A,\vec y}(\proj{\vec x_0} \otimes \nu_\lambda)
    }\right] \geq \poly(\eps,1/q).
$$
Assuming the $\LWE_{n,q,\alpha q}^m$ assumption, we can invoke the Gaussian-collapsing property in \Cref{thm:Gauss-collapsing} once again which implies that the quantum extractor $\algo E$ satisfies
$$
\Pr \left[\substack{
\mathcal{E}(\vec A,\vec y,\rho_{\vec{\textsc{Aux}}}) = \vec x\vspace{1mm}\\ \bigwedge \vspace{0.1mm}\\ {\vec x} \,\,\in\,\, \Lambda_q^{\vec y}(\vec A) \,\cap\, \algo B^m(\vec 0,\sigma \sqrt{\frac{m}{2}})}\, : \, \substack{
 \vec A \rand \Z_q^{n \times m}\\
(\ket{\psi_{\vec y}},\vec y) \leftarrow \mathsf{GenGauss}(\vec A,\sigma)\\
\rho_{\reg R,\vec{\textsc{Aux}}} \leftarrow \algo A_{\lambda,\vec A,\vec y}(\proj{\psi_{\vec y}} \otimes \nu_\lambda)
    }\right] \geq \poly(\eps,1/q).
$$
This proves the claim.
\end{proof}

\paragraph{Simultaneous sarch-to-decision reduction.}
Next, we give a strengthening of our result in \Cref{thm:search-to-decision-without-revoke} and state a \emph{simultaneous} search-to-decision reduction with quantum auxiliary input which holds even if additionally require that a \emph{revocation} procedure succeeds on a separate register.

To formalize the notion that revocation is applied on a separate register, we introduce a procedure called $\ineffrevoke$ which is defined in \Cref{fig:ineff-revoke}. In \Cref{fig:ineff-Qsampgauss}, we also introduce an inefficient variant of $\mathsf{QSampGauss}$ from \Cref{sec:QSampGauss} which does not require a trapdoor.

\begin{figure}[!htb]
   \begin{center} 
   \begin{tabular}{|p{11cm}|}
    \hline 
\begin{center}
\underline{$\ineffrevoke(\vec A,\vec y,\sigma,\rho_{\reg R})$}: 
\end{center}
Input: $\vec A \in \Zq^{n \times m}$, $\vec y \in \Z_q^n$ and $\rho \in L(\algo H_{\reg R})$.\\
Output: Accept ($\top$) or reject ($\bot$).
\vspace{3mm}\\
\textbf{Procedure:}
\begin{enumerate}
   \item Apply the (inefficient) projective measurement 
   $$
   \big\{\proj{\psi_{\vec y}}, I - \proj{\psi_{\vec y}}\big\}
   $$
   where $\ket{\psi_{\vec y}}$ is the Gaussian coset state
    $$ \ket{\psi_{\vec y}} = \sum_{\substack{\vec x \in \Z_q^m:\\ \vec A \vec x= \vec y \Mod{q}}} \rho_\sigma(\vec x) \ket{\vec x}.
    $$
   
   \item If the measurement succeeds, output $\top$. Else, output $\bot$.
\end{enumerate}
\ \\
\hline
\end{tabular}
    \caption{The procedure $\ineffrevoke(\vec A,\vec y,\sigma,\rho_{\reg R})$.}
    \label{fig:ineff-revoke}
    \end{center}
\end{figure}

\begin{figure}[!htb]
   \begin{center} 
   \begin{tabular}{|p{11cm}|}
    \hline 
\begin{center}
\underline{$\mathsf{IneffQSampGauss}(\vec A,\vec y,\sigma)$}: 
\end{center}
Input: $\vec A \in \Zq^{n \times m}$, $\vec y \in \Z_q^n$ and $\sigma > 0$.\\
Output: Gaussian coset state $\ket{\psi_{\vec y}}$.
\vspace{3mm}\\
\textbf{Procedure:}
\begin{enumerate}
   \item Prepare the Gaussian coset state
    $$ \ket{\psi_{\vec y}} = \sum_{\substack{\vec x \in \Z_q^m:\\ \vec A \vec x= \vec y \Mod{q}}} \rho_\sigma(\vec x) \ket{\vec x}.
    $$
   
   \item Output $\ket{\psi_{\vec y}}$.
\end{enumerate}
\ \\
\hline
\end{tabular}
    \caption{The procedure $\mathsf{IneffQSampGauss}(\vec A,\vec y,\sigma)$.}
    \label{fig:ineff-Qsampgauss}
    \end{center}
\end{figure}

\noindent We use the following conjecture. We refer the reader to the introduction for an informal explanation of the conjecture below.

\begin{conjecture}\label{conjecture-I}
Let $\lambda \in \N$. Then, there exist parameters (each parameterized by $\lambda$) such that $n\in \N$, $q$ is a prime with $q=2^{o(n)}$, $m \geq 2n \log q$, $\sqrt{8m}< \sigma <q/\sqrt{8m}$, $\alpha \in (0,1)$ with $1/\alpha= 2^{o(n)} \cdot \sigma$ for which the following holds: for any $\QPT$ $\algo A$ and $\QPT$ $\algo C$ (and for a fixed algorithm $\algo B$), and for any fixed $\poly(\lambda)$-sized quantum auxiliary input $\nu_\lambda$ (which depends on $\lambda$):
$$
\vline \, \Pr[1 \leftarrow \expt_{\algo A,\algo B,\algo C}(1^{\secparam},0)]  - \Pr[1 \leftarrow \expt_{\algo A,\algo B,\algo C}(1^{\secparam},1)] \, \vline \leq \negl(\lambda),
$$
where $\expt_{\algo A,\algo B,\algo C}(1^{\secparam},b)$ is the asymmetric cloning experiment in \Cref{fig:Expt-Conjecture1}.
\end{conjecture}

\begin{figure}
   \begin{center} 
   \begin{tabular}{|p{14cm}|}
    \hline 
\begin{center}
\underline{$\expt_{\algo A,\algo B,\algo C}(1^{\secparam},b)$}: 
\end{center}
\begin{enumerate}
        \item The challenger samples a random matrix $\vec A \rand \Z_q^{n \times m}$ and a Gaussian vector $\vec x_0 \sim D_{\Z_q^m,\frac{\sigma}{\sqrt{2}}}$ and lets $\vec y= \vec A \cdot \vec x_0 \Mod{q}$. Then, the challenger runs $\ket{\psi_{\vec y}} \leftarrow \mathsf{IneffQSampGauss}(\vec A,\vec y,\sigma)$ and sends $(\vec A,\vec y,\ket{\psi_{\vec y}})$ to $\algo A$.

        \item $\algo A$ receives $(\vec A,\vec y,\ket{\psi_{\vec y}})$ together with auxiliary input $\nu_\lambda$ and generates a state $\rho_{\reg B \reg C}$ in systems $\reg{BC}$, and sends $\reg B$ to $\algo B$ and $\reg C$ to $\algo C$.

        \item The challenger sends $(\vec A,\vec y,\sigma)$ to $\algo B$ and, depending on the value of $b$, the challenger sends the following to $\algo C$:
        \begin{itemize}
            \item if $b=0$: the challenger samples $\vec s \rand \Zq^n$, $\vec e \sim D_{\Z^m,\alpha q}$, lets $\vec u = \vec A^\intercal \vec s + \vec e$, and sends $(\vec A,\vec y,\vec u,\vec u^\intercal \vec x_0 \Mod{q})$ to $\algo C$.

            \item if $b=0$: the challenger samples a uniformly random vector $\vec u \rand \Z_q^m$ and sends $(\vec A,\vec y,\vec u,\vec u^\intercal \vec x_0 \Mod{q})$ to $\algo C$.
        \end{itemize}

        \item $\algo B$ receives as input $(\vec A,\vec y,\sigma,\reg B)$ and runs $\ineffrevoke(\vec A,\vec y,\sigma,\reg B)$. Next, $\algo B$ outputs $\top$, if the measurement succeeds, else $\algo B$ outputs $\bot$.
        
       \item $\algo C$ receives as input $(\vec A,\vec y,\vec u,\vec u^\intercal \vec x_0 \Mod{q},\reg C)$ and outputs a bit $b'$.

       \item The challenger outputs $1$, if $\algo B$ outputs $\top$ and $\algo C$ outputs $b'=b$. This is also the outcome of the experiment.
    \end{enumerate}\\
\hline
\end{tabular}
    \caption{The experiment for Conjecture 1.}
    \label{fig:Expt-Conjecture1}
    \end{center}
\end{figure}

\begin{theorem}[Simultaneous Search-To-Decision Reduction with Quantum Auxiliary Input]\label{thm:search-to-decision}
Let $\lambda\in \N$ be the security parameter. Let $n\in \N$, $q$ be a prime with $q=2^{o(n)}$, $m \geq 2n \log q$, $\sqrt{8m}< \sigma <q/\sqrt{8m}$, and let $\alpha,\beta \in (0,1)$ with $\beta/\alpha = 2^{o(n)}$ with $1/\alpha= 2^{o(n)} \cdot \sigma$. Let
$\algo A = \{(\algo A_{\lambda,\vec A,\vec y},\nu_\lambda)\}_{\lambda \in \N}$
be any non-uniform quantum algorithm consisting of a family of polynomial-sized quantum circuits
$$\Bigg\{\algo A_{\lambda,\vec A,\vec y}: \algo L(\algo H_q^m \otimes \algo H_{B_\lambda}) \allowbreak\rightarrow \algo L(\algo H_{R_\lambda} \otimes \algo H_{\textsc{aux}_\lambda})\Bigg\}_{\vec A \in \Zq^{n \times m}, \,\vec y \in \Z_q^n}$$
and polynomial-sized advice states $\nu_\lambda \in \algo D(\algo H_{B_\lambda})$ which are independent of $\vec A$.
Then, assuming \Cref{conjecture-I} is true, the following holds for every $\QPT$ distinguisher $\mathcal{D}$. Suppose there exists a function $\eps(\lambda) = 1/\poly(\lambda)$ such that
\begin{align*}
&\vline \, \prob\left[ 1 \leftarrow \mathsf{SimultSearchToDecisionExpt}^{\adversary,\algo D}(1^{\secparam},0)\right]-\\
&\,\prob\left[ 1 \leftarrow \mathsf{SimultSearchToDecisionExpt}^{\adversary,\algo D}(1^{\secparam},1)\right]  \,\vline = \eps(\lambda).
\end{align*}
\begin{figure}[!htb]
   \begin{center} 
   \begin{tabular}{|p{14cm}|}
    \hline 
\begin{center}
\underline{$\mathsf{SimultSearchToDecisionExpt}^{\adversary,\algo D}\left( 1^{\secparam},b \right)$}: 
\end{center}
\begin{itemize}
\item If $b=0$: output $\simult.\lwe.\dist^{\adversary,\algo D}\left( 1^{\secparam} \right)$ defined in \Cref{fig:lwe-dist}.
\item If $b=1$: output $\simult.\unif.\dist^{\adversary,\algo D}\left( 1^{\secparam} \right)$ defined in \Cref{fig:unif.dist}.
\end{itemize}
\ \\
\hline
\end{tabular}
    \caption{The experiment $\mathsf{SimultSearchToDecisionExpt}^{\adversary,\algo D}\left( 1^{\secparam},b \right)$.}
    \label{fig:simult-search-to-decision}
    \end{center}
\end{figure}

\noindent Then, there exists a quantum extractor $\algo E$ that takes as input $\vec A$, $\vec y$ and system $\vec{\textsc{Aux}}$ of the state $\rho_{\reg R,\vec{\textsc{Aux}}}$ and outputs a short vector in the coset $\Lambda_q^{\vec y}(\vec A)$ in time $\poly(\lambda,m,\sigma,q,1/\eps)$ such that
\begin{align*}
&\Pr \left[\substack{
\ineffrevoke(\vec A,\vec y,\sigma,\reg R)  = \top\vspace{0.2mm}\\ \bigwedge \vspace{0.2mm}\\ \mathcal{E}(\vec A,\vec y,\reg{Aux}) \,\,\in\,\, \Lambda_q^{\vec y}(\vec A) \,\cap\, \algo B^m(\vec 0,\sigma \sqrt{\frac{m}{2}})}
 \, : \, \substack{
 \vec A \rand \Z_q^{n \times m}\\
(\ket{\psi_{\vec y}},\vec y) \leftarrow \mathsf{GenGauss}(\vec A,\sigma)\\
\rho_{\reg R,\vec{\textsc{Aux}}} \leftarrow \algo A_{\lambda,\vec A,\vec y}(\proj{\psi_{\vec y}} \otimes \nu_\lambda)
    }\right] \, \geq\, \poly(\eps,1/q).
\end{align*}
\end{theorem}

\begin{figure}[!htb]
   \begin{center} 
   \begin{tabular}{|p{14cm}|}
    \hline 
\begin{center}
\underline{$\simult.\lwe.\dist^{\adversary,\algo D}\left( 1^{\secparam}\right)$}: 
\end{center}
\begin{enumerate}
   \item Sample $ \vec A \rand \Z_q^{n \times m}$.
   \item Generate $(\ket{\psi_{\vec y}},\vec y) \leftarrow \mathsf{GenGauss}(\vec A,\sigma)$.
   \item Generate $\rho_{\reg R, \,\vec{\textsc{Aux}}} \leftarrow \algo A_{\lambda,\vec A,\vec y}(\proj{\psi_{\vec y}} \otimes \nu_\lambda)$.
    \item Sample $\vec s \rand \Z_q^n, \vec e \sim D_{\Z^{m},\alpha q}$ and $e'\sim D_{\Z,\beta q}$.
   
    \item Generate $\rho_{\reg R,\vec{\textsc{Aux}}} \leftarrow \algo A_{\lambda,\vec A,\vec y}(\proj{\psi_{\vec y}} \otimes \nu_\lambda)$.

   \item Run $\ineffrevoke(\vec A,\vec y,\sigma,\cdot)$ on system $\reg R$. If it outputs $\top$, continue. Otherwise, output $\invalid$.
    
    \item Run $b' \leftarrow \mathcal{D}(\vec A,\vec y,\vec s^\intercal \vec A+ \vec e^\intercal,\vec s^\intercal\vec y + e',\cdot)$ on system \textsc{Aux}. Output $b'$.

\end{enumerate}
\ \\
\hline
\end{tabular}
    \caption{The distribution $\simult.\lwe.\dist^{\adversary,\algo D}\left( 1^{\secparam}\right)$.}
    \label{fig:lwe-dist}
    \end{center}
\end{figure}

\begin{figure}[!htb]
   \begin{center} 
   \begin{tabular}{|p{14cm}|}
    \hline 
\begin{center}
\underline{$\simult.\unif.\dist^{\adversary,\algo D}\left( 1^{\secparam} \right)$}: 
\end{center}
\begin{enumerate}
   \item Sample $ \vec A \rand \Z_q^{n \times m}$.
   \item Generate $(\ket{\psi_{\vec y}},\vec y) \leftarrow \mathsf{GenGauss}(\vec A,\sigma)$.

\item Sample $\vec u \rand \Z_q^{m}$ and $r \rand \Z_q$.
   
   \item Generate $\rho_{R, \vec{\textsc{Aux}}} \leftarrow \algo A_{\lambda,\vec A,\vec y}(\proj{\psi_{\vec y}} \otimes \nu_\lambda)$.

\item Run $\ineffrevoke(\vec A,\vec y,\sigma,\cdot)$ on system $\reg R$. If it outputs $\top$, continue. Otherwise, output $\invalid$.

\item Run $b' \leftarrow \mathcal{D}(\vec A,\vec y,\vec u,r,\cdot)$ on system \textsc{Aux}. Output $b'$.
\end{enumerate}
\ \\
\hline
\end{tabular}
    \caption{The distribution $\simult.\unif.\dist^{\adversary,\algo D}\left( 1^{\secparam} \right)$.}
    \label{fig:unif.dist}
    \end{center}
\end{figure}

\begin{proof}
Let $\lambda \in \N$ be the security parameter and let $\algo A = \{(\algo A_{\lambda,\vec A,\vec y},\nu_\lambda)\}_{\vec A \in \Z_q^{n \times m}}$
be a non-uniform quantum algorithm. Suppose that $\mathcal{D}$ is a $\QPT$ distinguisher with advantage $\eps = 1/\poly(\lambda)$.

To prove the claim, we consider the following sequence of hybrid distributions.

\begin{description}
  
\item $\hybrid_0$: This is the distribution $\simult.\lwe.\dist^{\adversary,\algo D}\left( 1^{\secparam} \right)$ in \Cref{fig:lwe-dist}.

\item $\hybrid_1$: This is the following distribution:
\begin{enumerate}
        \item Sample a random matrix $\vec A \rand \Z_q^{n \times m}$.

        \item \rc{Sample a Gaussian vector $\vec x_0 \sim D_{\Z_q^m,\frac{\sigma}{\sqrt{2}}}$ and let $\vec y= \vec A \cdot \vec x_0 \Mod{q}$.}

        \item Run $\ket{\psi_{\vec y}} \leftarrow \mathsf{IneffQSampGauss}(\vec A,\vec y,\sigma)$.

        \item Sample $\vec s \rand \Zq^n$, $\vec e \sim D_{\Z^m,\alpha q}$ and $e'\sim D_{\Z,\beta q}$. \rc{Let $\vec u = \vec A^\intercal \vec s+ \vec e$}.

        \item Run $\algo A_{\lambda,\vec A,\vec y}(\proj{\psi_{\vec y}} \otimes \nu_\lambda)$ to generate a state $\rho_{\reg R ,\vec{\textsc{aux}}}$ in systems $\reg R$ and $\vec{\textsc{aux}}$.

        \item Run $\ineffrevoke(\vec A,\vec y,\sigma,\cdot)$ on $\reg R$. If it outputs $\top$, continue. Otherwise, output $\invalid$.
        
       \item Run the distinguisher 
        \rc{$\mathcal{D}(\vec A,\vec y,\vec u,\vec u^\intercal\vec x_0 + e',\cdot)$} on system \textsc{Aux}. Output $b'$.
    \end{enumerate}

\item $\hybrid_2:$ This is the following distribution:
\begin{enumerate}
        \item Sample a random matrix $\vec A \rand \Z_q^{n \times m}$.

        \item Sample a Gaussian vector $\vec x_0 \sim D_{\Z_q^m,\frac{\sigma}{\sqrt{2}}}$ and let $\vec y= \vec A \cdot \vec x_0 \Mod{q}$.

        \item Run $\ket{\psi_{\vec y}} \leftarrow \mathsf{IneffQSampGauss}(\vec A,\vec y,\sigma)$.

        \item Sample  \rc{$\vec u \rand \Z_q^m$ and $e'\sim D_{\Z,\beta q}$}.

        \item Run $\algo A_{\lambda,\vec A,\vec y}(\proj{\psi_{\vec y}} \otimes \nu_\lambda)$ to generate a state $\rho_{\reg R ,\vec{\textsc{aux}}}$ in systems $\reg R$ and $\vec{\textsc{aux}}$.

        \item Run $\ineffrevoke(\vec A,\vec y,\sigma,\cdot)$ on $\reg R$. If it outputs $\top$, continue. Otherwise, output $\invalid$.
        
       \item Run the distinguisher 
    $\mathcal{D}(\vec A,\vec y,\vec u,\vec u^\intercal\vec x_0 + e',\cdot)$ on system \textsc{Aux}. Output $b'$.
    \end{enumerate}

\item $\hybrid_3$: This is the distribution $\simult.\unif.\dist^{\adversary,\algo D}\left( 1^{\secparam} \right)$ defined in \Cref{fig:unif.dist}.

\end{description}
We now show the following:

\begin{claim} Hybrids $\hybrid_0$ and $\hybrid_1$ are statistically indistinguishable. In other words,
$$
\hybrid_0 \, \approx_s \,\hybrid_1.
$$
\end{claim}
\begin{proof} Here, we invoke the \emph{noise flooding} property in \Cref{lem:shifted-gaussian} to argue that $\vec e^\intercal \vec x_0 \ll e'$ holds with overwhelming probability for our choice of parameters.
Therefore, the distributions in $\hybrid_0$ and $\hybrid_1$ are computationally indistinguishable.
\end{proof}

\begin{claim} Assuming \Cref{conjecture-I} holds for our choice of parameters, the hybrids $\hybrid_1$ and $\hybrid_2$ are computationally indistinguishable,
$$
\hybrid_1 \, \approx_c \,\hybrid_2.
$$
\end{claim}
\begin{proof}
This follows directly from \Cref{conjecture-I} which allows us to invoke a variant of $\LWE$ assumption, even if the procedure $\ineffrevoke$ is applied on a separate register. Here, we rely on the fact that during the reduction, we can simply sample $e'\sim D_{\Z,\beta q}$ to produce an identically distributed challenge distribution.
\end{proof}

\par Recall that $\hybrid_0$ and $\hybrid_3$ can be distinguished with probability $\eps=1/\poly(\secparam)$. We proved that the hybrids $\hybrid_0$ and $\hybrid_2$ are computationally indistinguishable. As a consequence, it holds that hybrids $\hybrid_2$ and $\hybrid_3$ can be distinguished with probability at least $\eps - \negl(\secparam)$.
\par We leverage this to obtain a Goldreich-Levin reduction. Consider the following distinguisher.

\begin{figure}[!htb]
   \begin{center} 
   \begin{tabular}{|p{14cm}|}
    \hline 
\begin{center}
\underline{$\tilde{\algo D}\big(\vec A,\vec y,\vec u,v,\rho\big)$}: 
\end{center}
Input: $\vec A \in \Zq^{n \times m}$, $\vec y \in \Z_q^n$, $\vec u \in \Z_q^n$, $v \in \Z_q$ and $\rho \in L(\algo H_{\vec{\textsc{Aux}}})$.\\
Output: A bit $b' \in \bit$.
\vspace{3mm}\\
\textbf{Procedure:}
\begin{enumerate}
   \item Sample $e' \sim D_{\Z,\beta q}$.
   
   \item Output $b' \leftarrow \algo D\big(\vec A,\vec y,\vec u, v + e',\rho\big)$.
\end{enumerate}
\ \\
\hline
\end{tabular}
    \caption{The distinguisher $\tilde{\algo D}\big(\vec A,\vec y,\vec u,v,\rho\big)$.}
    \label{fig:tilde-dist}
    \end{center}
\end{figure}
Note that $r + e' \Mod{q}$ is uniform whenever $r \rand \Z_q$ and $e' \sim D_{\Z,\beta q}$. Therefore, our previous argument shows that there exists a negligible function $\eta$ such that:
\begin{align*}
\,&\vline\Pr \left[
\substack{
\ineffrevoke(\vec A,\vec y,\sigma,\reg R) = \top\\
\ \\
\bigwedge\\
\ \\
\tilde{\mathcal{D}}(\vec A,\vec y,\vec u, \vec u^\intercal \vec x_0, \textsc{Aux}) =1
}
\, : \, \substack{
\vec A \rand \Z_q^{n \times m}, \vec u \rand \Z_q^m\\
\vec x_0 \sim D_{\Z_q^m,\frac{\sigma}{\sqrt{2}}}, \, \vec y = \vec A \vec x_0 \text{ (mod $q$)}\\
\ket{\psi_{\vec y}} \leftarrow \mathsf{IneffQSampGauss}(\vec A,\vec y,\sigma)\\
\rho_{\reg R,\vec{\textsc{Aux}}} \leftarrow \algo A_{\lambda,\vec A,\vec y}(\proj{\psi_{\vec y}} \otimes \nu_\lambda)
}
\right] -\\
&\Pr \left[
\substack{
\ineffrevoke(\vec A,\vec y,\sigma,\reg R) = \top\\
\ \\
\bigwedge\\
\ \\
\tilde{\mathcal{D}}(\vec A,\vec y,\vec u,r, \textsc{Aux}) =1
}
\, : \, \substack{
\vec A \rand \Z_q^{n \times m}\\
\vec u \rand \Z_q^m, r \rand \Z_q\\
\vec x_0 \sim D_{\Z_q^m,\frac{\sigma}{\sqrt{2}}}, \, \vec y = \vec A \vec x_0 \text{ (mod $q$)}\\
\ket{\psi_{\vec y}} \leftarrow \mathsf{IneffQSampGauss}(\vec A,\vec y,\sigma)\\
\rho_{\reg R,\vec{\textsc{Aux}}} \leftarrow \algo A_{\lambda,\vec A,\vec y}(\proj{\psi_{\vec y}} \otimes \nu_\lambda)
}
\right]  \vline \geq \eps - \eta(\lambda).
\end{align*}

From \Cref{thm:QGL}, it follows that there exists a Goldreich-Levin extractor $\algo E$ running in time $T(\algo E) =\poly(\lambda,n,m,\sigma,q,1/\eps)$ that outputs a short vector in $\Lambda_q^{\vec y}(\vec A)$ with probability at least
$$
\Pr \left[\substack{
\mathcal{E}(\vec A,\vec y,\rho_{\vec{\textsc{Aux}}}) = \vec x\vspace{1mm}\\ \bigwedge \vspace{0.1mm}\\ {\vec x} \,\,\in\,\, \Lambda_q^{\vec y}(\vec A) \,\cap\, \algo B^m(\vec 0,\sigma \sqrt{\frac{m}{2}})}\, : \, \substack{
 \vec A \rand \Z_q^{n \times m}\\
\vec x_0 \sim D_{\Z_q^m,\frac{\sigma}{\sqrt{2}}}, \, \vec y = \vec A \vec x_0 \text{ (mod $q$)}\\
\ket{\psi_{\vec y}} \leftarrow \mathsf{IneffQSampGauss}(\vec A,\vec y,\sigma)\\
\rho_{\reg R,\vec{\textsc{Aux}}} \leftarrow \algo A_{\lambda,\vec A,\vec y}(\proj{\psi_{\vec y}} \otimes \nu_\lambda)
    }\right] \geq \poly(\eps,1/q).
$$
This proves the claim.
\end{proof}

Finally, we give a proof of the \emph{simultaneous} search-to-decision reduction (\Cref{thm:search-to-decision}) from standard assumptions without \Cref{conjecture-I} in the special case when revocation succeeds with overwhelming probability.

\begin{theorem}\label{thm:search-to-decision-revoke-1}
Let $n\in \N$. Let $q$ be a prime with $q=2^{o(n)}$, $m \geq 2n \log q$, $\sqrt{8m}< \sigma <q/\sqrt{8m}$, and let $\alpha,\beta \in (0,1)$ with $\beta/\alpha = 2^{o(n)}$ with $1/\alpha= 2^{o(n)} \cdot \sigma$. Let
$\algo A = \{(\algo A_{\lambda,\vec A,\vec y},\nu_\lambda)\}_{\lambda \in \N}$
be any non-uniform quantum algorithm consisting of a family of polynomial-sized quantum circuits
$$\Bigg\{\algo A_{\lambda,\vec A,\vec y}: \algo L(\algo H_q^m \otimes \algo H_{B_\lambda}) \allowbreak\rightarrow \algo L(\algo H_{R_\lambda} \otimes \algo H_{\textsc{aux}_\lambda})\Bigg\}_{\vec A \in \Zq^{n \times m}, \,\vec y \in \Z_q^n}$$
and polynomial-sized advice states $\nu_\lambda \in \algo D(\algo H_{B_\lambda})$ which are independent of $\vec A$ such that
$$
\Pr\left[
\ineffrevoke(\vec A,\vec y,\sigma,\rho_{\reg R})  = \top
 \, : \, \substack{
 \vec A \rand \Z_q^{n \times m}\\
(\ket{\psi_{\vec y}},\vec y) \leftarrow \mathsf{GenGauss}(\vec A,\sigma)\\
\rho_{\reg R,\vec{\textsc{Aux}}} \leftarrow \algo A_{\lambda,\vec A,\vec y}(\proj{\psi_{\vec y}} \otimes \nu_\lambda)
    }\right] = 1-\nu(\lambda),
$$
for some negligle function $\nu(\lambda)$.
Then,
assuming the quantum hardness of the $\LWE_{n,q,\alpha q}^m$ assumption, the following holds for every $\QPT$ distinguisher $\mathcal{D}$. Suppose that there exists a function $\eps(\lambda) = 1/\poly(\lambda)$ such that
\begin{align*}
&\vline \, \prob\left[ 1 \leftarrow \mathsf{SimultSearchToDecisionExpt}^{\adversary,\algo D}(1^{\secparam},0)\right]-\\
&\,\prob\left[ 1 \leftarrow \mathsf{SimultSearchToDecisionExpt}^{\adversary,\algo D}(1^{\secparam},1)\right]  \,\vline = \eps(\lambda).
\end{align*}

\noindent Then, there exists a quantum extractor $\algo E$ that takes as input $\vec A$, $\vec y$ and system $\vec{\textsc{Aux}}$ of the state $\rho_{\reg R,\vec{\textsc{Aux}}}$ and outputs a short vector in the coset $\Lambda_q^{\vec y}(\vec A)$ in time $\poly(\lambda,m,\sigma,q,1/\eps)$ such that
\begin{align*}
&\Pr \left[\substack{
\ineffrevoke(\vec A,\vec y,\sigma,\reg R)  = \top\vspace{0.2mm}\\ \bigwedge \vspace{0.2mm}\\ \mathcal{E}(\vec A,\vec y,\reg{Aux}) \,\,\in\,\, \Lambda_q^{\vec y}(\vec A) \,\cap\, \algo B^m(\vec 0,\sigma \sqrt{\frac{m}{2}})}
 \, : \, \substack{
 \vec A \rand \Z_q^{n \times m}\\
(\ket{\psi_{\vec y}},\vec y) \leftarrow \mathsf{GenGauss}(\vec A,\sigma)\\
\rho_{\reg R,\vec{\textsc{Aux}}} \leftarrow \algo A_{\lambda,\vec A,\vec y}(\proj{\psi_{\vec y}} \otimes \nu_\lambda)
    }\right] \, \geq\, \poly(\eps,1/q).
\end{align*}
\end{theorem}

\begin{proof}
By assumption, there exists an adversary $(\algo A,\algo D)$ such that $\mathsf{Adv}(\algo A,\algo D) = \eps(\lambda)$, where
\begin{align*}
\mathsf{Adv}(\algo A,\algo D) = \,&\vline\Pr \left[
\substack{
\ineffrevoke(\vec A,\vec y,\sigma,\reg R) = \top\\
\ \\
\bigwedge\\
\ \\
\mathcal{D}(\vec A,\vec y,\vec s^\intercal \vec A + \vec e^\intercal, \vec s^\intercal \vec y + e', \textsc{Aux}) =1
}
\, : \, \substack{
\vec A \rand \Z_q^{n \times m}, \,\vec s \rand \Z_q^n\\
\vec e \sim D_{\Z^{m},\,\alpha q}, \, e' \sim D_{\Z,\,\beta q}\\
(\ket{\psi_{\vec y}},\vec y) \leftarrow \mathsf{GenGauss}(\vec A,\sigma)\\
\rho_{\reg R,\vec{\textsc{Aux}}} \leftarrow \algo A_{\lambda,\vec A,\vec y}(\proj{\psi_{\vec y}} \otimes \nu_\lambda)
}
\right] -\\
&\Pr \left[
\substack{
\ineffrevoke(\vec A,\vec y,\sigma,\reg R) = \top\\
\ \\
\bigwedge\\
\ \\
\mathcal{D}(\vec A,\vec y,\vec u,r, \textsc{Aux}) =1
}
\, : \, \substack{
\vec A \rand \Z_q^{n \times m}\\
\vec u \rand \Z_q^m, r \rand \Z_q\\
(\ket{\psi_{\vec y}},\vec y) \leftarrow \mathsf{GenGauss}(\vec A,\sigma)\\
\rho_{\reg R,\vec{\textsc{Aux}}} \leftarrow \algo A_{\lambda,\vec A,\vec y}(\proj{\psi_{\vec y}} \otimes \nu_\lambda)
}
\right]  \vline.
\end{align*}
We can now invoke \Cref{lem:remove-revoke} to argue that there exists a $\QPT$ distinguisher $\tilde{\algo D}$ (that internally runs $\algo D$) and succeeds on the reduced system \reg{Aux} alone, i.e.
\begin{align*}
&\vline\Pr \left[
\tilde{\mathcal{D}}(\vec A,\vec y,\vec s^\intercal \vec A + \vec e^\intercal, \vec s^\intercal \vec y + e', \rho_{\textsc{Aux}}) =1
\, : \, \substack{
\vec A \rand \Z_q^{n \times m}, \,\vec s \rand \Z_q^n\\
\vec e \sim D_{\Z^{m},\,\alpha q}, \, e' \sim D_{\Z,\,\beta q}\\
(\ket{\psi_{\vec y}},\vec y) \leftarrow \mathsf{GenGauss}(\vec A,\sigma)\\
\rho_{\reg R,\vec{\textsc{Aux}}} \leftarrow \algo A_{\lambda,\vec A,\vec y}(\proj{\psi_{\vec y}} \otimes \nu_\lambda)
}
\right] -\\
&\Pr \left[
\tilde{\mathcal{D}}(\vec A,\vec y,\vec u,r, \rho_{\textsc{Aux}}) =1
\, : \,  \substack{
\vec A \rand \Z_q^{n \times m}\\
\vec u \rand \Z_q^m, r \rand \Z_q\\
(\ket{\psi_{\vec y}},\vec y) \leftarrow \mathsf{GenGauss}(\vec A,\sigma)\\
\rho_{\reg R,\vec{\textsc{Aux}}} \leftarrow \algo A_{\lambda,\vec A,\vec y}(\proj{\psi_{\vec y}} \otimes \nu_\lambda)
}
\right]  \vline = \bar{\eps}(\lambda),
\end{align*}
for some $\bar{\eps}=1/\poly(\lambda)$.
In other words, the $\QPT$ algorithm $\tilde{\algo D}$ can successfully predict whether it has received a Dual-Regev sample or a uniformly random sample.
Therefore, we can now invoke \Cref{thm:search-to-decision-without-revoke} to argue there exists an extractor $\algo E$ that takes as input $\vec A$, $\vec y$ and system $\vec{\textsc{Aux}}$ of the state $\rho_{\reg R,\vec{\textsc{Aux}}}$ and outputs a short vector in the coset $\Lambda_q^{\vec y}(\vec A)$ in time $\poly(\lambda,m,\sigma,q,1/\bar{\eps})$ such that
\begin{align*}
\Pr \left[\substack{
\mathcal{E}(\vec A,\vec y,\rho_{\vec{\textsc{Aux}}}) = \vec x\vspace{1mm}\\ \bigwedge \vspace{0.1mm}\\ {\vec x} \,\,\in\,\, \Lambda_q^{\vec y}(\vec A) \,\cap\, \algo B^m(\vec 0,\sigma \sqrt{\frac{m}{2}})} \, : \, \substack{
 \vec A \rand \Z_q^{n \times m}\\
(\ket{\psi_{\vec y}},\vec y) \leftarrow \mathsf{GenGauss}(\vec A,\sigma)\\
\rho_{\reg R,\vec{\textsc{Aux}}} \leftarrow \algo A_{\lambda,\vec A,\vec y}(\proj{\psi_{\vec y}} \otimes \nu_\lambda)
    }\right] \geq \poly(\bar{\eps},1/q).
\end{align*}
Recall also that, by assumption, revocation succeeds with overwhelming probability, i.e.,
$$
\Pr\left[
\ineffrevoke(\vec A,\vec y,\sigma,\rho_{\reg R})  = \top
 \, : \, \substack{
 \vec A \rand \Z_q^{n \times m}\\
(\ket{\psi_{\vec y}},\vec y) \leftarrow \mathsf{GenGauss}(\vec A,\sigma)\\
\rho_{\reg R,\vec{\textsc{Aux}}} \leftarrow \algo A_{\lambda,\vec A,\vec y}(\proj{\psi_{\vec y}} \otimes \nu_\lambda)
    }\right] = 1-\negl(\lambda).
$$
Using Bonferroni's inequality, we can argue that
\begin{align*}
&\Pr \left[\substack{
\ineffrevoke(\vec A,\vec y,\sigma,\reg R)  = \top\vspace{0.2mm}\\ \bigwedge \vspace{0.2mm}\\ \mathcal{E}(\vec A,\vec y,\reg{Aux}) \,\,\in\,\, \Lambda_q^{\vec y}(\vec A) \,\cap\, \algo B^m(\vec 0,\sigma \sqrt{\frac{m}{2}})}
 \, : \, \substack{
 \vec A \rand \Z_q^{n \times m}\\
(\ket{\psi_{\vec y}},\vec y) \leftarrow \mathsf{GenGauss}(\vec A,\sigma)\\
\rho_{\reg R,\vec{\textsc{Aux}}} \leftarrow \algo A_{\lambda,\vec A,\vec y}(\proj{\psi_{\vec y}} \otimes \nu_\lambda)
    }\right]\\
&\geq \Pr \left[
\ineffrevoke(\vec A,\vec y,\sigma,\rho_{\reg R})= \top  \, : \, \substack{
 \vec A \rand \Z_q^{n \times m}\\
(\ket{\psi_{\vec y}},\vec y) \leftarrow \mathsf{GenGauss}(\vec A,\sigma)\\
\rho_{\reg R,\vec{\textsc{Aux}}} \leftarrow \algo A_{\lambda,\vec A,\vec y}(\proj{\psi_{\vec y}} \otimes \nu_\lambda)
    } \right]\\
    & \quad + \Pr \left[
\mathcal{E}(\vec A,\vec y,\rho_{\reg{Aux}}) \,\,\in\,\, \Lambda_q^{\vec y}(\vec A) \,\cap\, \algo B^m(\vec 0,\sigma \sqrt{m/2})  \, : \, \substack{
 \vec A \rand \Z_q^{n \times m}\\
(\ket{\psi_{\vec y}},\vec y) \leftarrow \mathsf{GenGauss}(\vec A,\sigma)\\
\rho_{\reg R,\vec{\textsc{Aux}}} \leftarrow \algo A_{\lambda,\vec A,\vec y}(\proj{\psi_{\vec y}} \otimes \nu_\lambda)
    }\right] -  1\\
&\geq \Pr \left[
\mathcal{E}(\vec A,\vec y,\rho_{\reg{Aux}}) \,\,\in\,\, \Lambda_q^{\vec y}(\vec A) \,\cap\, \algo B^m(\vec 0,\sigma \sqrt{m/2})  \, : \, \substack{
 \vec A \rand \Z_q^{n \times m}\\
(\ket{\psi_{\vec y}},\vec y) \leftarrow \mathsf{GenGauss}(\vec A,\sigma)\\
\rho_{\reg R,\vec{\textsc{Aux}}} \leftarrow \algo A_{\lambda,\vec A,\vec y}(\proj{\psi_{\vec y}} \otimes \nu_\lambda)
    }\right] -  \negl(\lambda)\\
&\geq \poly(\eps,1/q).
\end{align*}
This proves the claim.
    
\end{proof}

\subsection{Distinct Pair Extraction}\label{sec:distinct-pre-images}

The following lemma allows us to analyze the probability of simultaneously extracting two distinct preimages in terms of the success probability of revocation and the success probability of extracting a preimage from the adversary's state.

\begin{lemma}[Distinct pair extraction]
\label{clm:distinct:preimageext}
Let $\rho \in \mathcal{D}({\cal H}_{X} \otimes {\cal H}_{Y})$ be an any density matrix, for some Hilbert spaces ${\cal H}_{X}$ and ${\cal H}_{Y}$. Let $\ket{\psi} = \sum_{x \in \algo S} \alpha_x \ket{x}\in {\cal H}_{X}$ be any state supported on a subset ${\cal S} \subseteq \algo X$, and let $\vec{\Pi}= \ketbra{\psi}{\psi}$ denote its associated projection.
Let $\vec \Pi_{\algo S}$ be the projector onto $\algo S$ with
$$
\vec \Pi_{\algo S} = \sum_{x \in \algo S} \proj{x}.
$$
Let $\cal{E}: \algo L(\algo H
_Y) \rightarrow \algo L(\algo H_{X'})$ be any $\CPTP$ map of the form 
$$
\cal{E}_{Y \rightarrow X'}(\sigma) = \Tr_{E} \left[V_{Y \rightarrow X'E}\, \sigma V_{Y \rightarrow X'E}^\dag\right], \quad \forall \sigma \in \algo D(\algo H_{Y}),
$$
for some isometry $V_{Y \rightarrow X'E}$. Consider the measurement specified by
$$\vec{\Gamma} = \sum_{x,x' \in \setS:x \neq x'} \ketbra{x}{x}_{X} \otimes V_{Y \rightarrow X'E}^{\dagger} (\ketbra{x'}{x'}_{X'} \otimes I_E) V_{Y \rightarrow X'E}.$$
Let $\rho_X = \Tr_Y[\rho_{XY}]$ denote the reduced state. Then, it holds that
$$\Tr[\vec{\Gamma} \rho] \,\geq \, \left(1- \max_{x \in \setS} |\alpha_x|^2 \right) \cdot \Tr[\vec \Pi \rho_X]\cdot \Tr\left[ \vec \Pi_{\algo S} \, \algo E_{Y \rightarrow X'}(\sigma) \right],$$
where $\sigma = \Tr[(\vec \Pi \otimes I)\rho]^{-1} \cdot  \Tr_X[(\vec \Pi \otimes I)\rho]$ is a reduced state in system $Y$.
\end{lemma}
\begin{proof}
\noindent 
Because the order in which we apply $\vec \Gamma$ and $(\vec \Pi \otimes I)$ does not matter, we have the inequality
\begin{align}\label{ineq:gamma}
\Tr\left[\vec\Gamma \rho \right]
\geq \Tr\left[ \left(\vec \Pi \otimes I \right) \vec \Gamma \rho  \right] =  \Tr\left[ \left(\vec \Pi \otimes I \right) \vec \Gamma \rho \left(\vec \Pi \otimes I \right)  \right]
 = \Tr\left[ \vec \Gamma  (\vec \Pi \otimes I )  \rho \left(\vec \Pi \otimes I \right) \right]. 
\end{align}
Notice also that $(\vec \Pi \otimes I)\rho(\vec \Pi \otimes I)$ lies in the image of $(\vec \Pi \otimes I)$ with $\vec{\Pi}= \ketbra{\psi}{\psi}$, and thus
\begin{align}\label{ineq:sigma}
(\vec \Pi \otimes I)\rho(\vec \Pi \otimes I) = \Tr[(\vec \Pi \otimes I)\rho] \cdot ( \ketbra{\psi}{\psi} \otimes \sigma),
\end{align}
for some $\sigma \in \algo D(\algo H_Y)$.
Putting everything together, we get that
\begin{eqnarray*}
\Tr\left[\vec\Gamma \rho \right]
& \geq & \Tr\left[ \vec \Gamma  (\vec \Pi \otimes I )  \rho \left(\vec \Pi \otimes I \right) \right] \quad\quad\,\,\,\quad\quad\quad\quad\quad\quad\quad\quad\,\,\,\,\,\,\text{(using inequality \eqref{ineq:gamma})}\\
& = &  \Tr[(\vec \Pi \otimes I)\rho] \cdot \Tr\left[\vec \Gamma\left( \ketbra{\psi}{\psi} \otimes \sigma \right) \right] \quad\quad\quad\quad\quad\quad\quad\,\,  \text{(using equation \eqref{ineq:sigma})}\\
& = &  \Tr[\vec \Pi \rho_X] \cdot \Tr\left[ \sum_{x,x'  \in \setS: x \neq x'} \ketbra{x}{x}_X \otimes V_{Y \rightarrow X'E}^{\dagger} \left(  \ketbra{x'}{x'}_{X'} \otimes I_E \right) V_{Y \rightarrow X'E} \left( \ketbra{\psi}{\psi} \otimes \sigma \right) \right] \\
& = &  \Tr[\vec \Pi \rho_X] \cdot \sum_{x' \in \setS}\left( \sum_{x \in \setS:x \neq x'} |\langle x|\psi \rangle|^2 \right)\Tr\left[  V_{Y \rightarrow X'E}^{\dagger} (\ketbra{x'}{x'}_{X'} \otimes I_E)  V_{Y \rightarrow X'E} \, \sigma\right] \\
& = & \Tr[\vec \Pi \rho_X] \cdot \sum_{x' \in \setS}\left(1-|\alpha_{x'}|^2 \right)\Tr\left[  (\ketbra{x'}{x'}_{X'} \otimes I_E)  V_{Y \rightarrow X'E} \,\sigma \, V_{Y \rightarrow X'E}^{\dagger}\right] \\
& \geq & \Tr[\vec \Pi \rho_X] \cdot \left(1- \max_{x \in \setS}|\alpha_x|^2 \right)\cdot \sum_{x' \in \setS}\Tr\left[  (\ketbra{x'}{x'}_{X'} \otimes I_E)  V_{Y \rightarrow X'E} \,\sigma \, V_{Y \rightarrow X'E}^{\dagger}\right] \\
& = &  \Tr[\vec \Pi \rho_X] \cdot \left(1- \max_{x \in \setS} |\alpha_x|^2 \right)\cdot \sum_{x' \in \setS}\Tr\left[ \ketbra{x'}{x'}_{X'} \Tr_E \left[V_{Y \rightarrow X'E} \,\sigma \, V_{Y \rightarrow X'E}^{\dagger}\right] \right] \\
& = & \Tr[\vec \Pi \rho_X] \cdot \left(1- \max_{x \in \setS
} |\alpha_x|^2 \right)\cdot \Tr\left[ \vec \Pi_{\algo S} \, \algo E_{Y \rightarrow X'}(\sigma) \right].
\end{eqnarray*}
This proves the claim.
\end{proof}

\subsection{Proof of Theorem~\ref{thm:security-Dual-Regev-high-revoke}}  \label{sec:dual-regev-proof}

\begin{proof} 
Let $\algo A$ be a $\QPT$ adversary and suppose that
$$ \left| \Pr\left[ 1 \leftarrow \expt_{\algo A}(1^{\secparam},0) \right] - \Pr\left[ 1 \leftarrow \expt_{\algo A}(1^{\secparam},1) \right] \right| =  \eps(\secparam),$$
for some $\eps(\lambda)$ with respect to $\expt_{\algo A}(1^{\secparam},b)$ in \Cref{fig:Dual-Regev-security}. We show that $\eps(\lambda)$ is negligible.

\begin{figure}
   \begin{center} 
   \begin{tabular}{|p{14cm}|}
    \hline 
\begin{center}
\underline{$\expt_{\adversary}(1^{\secparam},b)$}: 
\end{center}
\begin{enumerate}
    \item The challenger samples $(\vec A \in \Z_q^{n \times m},\mathsf{td}_{\vec A}) \leftarrow \GenTrap(1^n,1^m,q)$ and generates
$$
\ket{\psi_{\vec y}} \,\,= \sum_{\substack{\vec x \in \Z_q^m\\ \vec A \vec x = \vec y \Mod{q}}}\rho_{\sigma}(\vec x)\,\ket{\vec x},
$$
for some $\vec y \in \Z_q^n$, by running $(\ket{\psi_{\vec y}}, \vec y) \leftarrow \mathsf{GenGauss}(\vec A,\sigma)$. The challenger lets $\msk \leftarrow \mathsf{td}_{\vec A}$ and $\pk \leftarrow (\vec A,\vec y)$ and sends  $\sk \leftarrow \ket{\psi_{\vec y}}$ to the adversary $\mathcal A$.
\item $\mathcal A$ generates a (possibly entangled) bipartite state $\rho_{R,\vec{\textsc{aux}}}$ in systems $\algo H_R \otimes \algo H_{\textsc{Aux}}$ with $\algo H_R= \algo H_q^m$, returns system $R$ and holds onto the auxiliary system $\textsc{Aux}$.
\item The challenger runs $\revoke(\pk,\msk,\rho_R)$, where $\rho_R$ is the reduced state in system $R$. If the outcome is $\top$, the game continues. Otherwise, output \textsf{Invalid}.

\item $\algo A$ submits a plaintext bit $\mu \in \bit$.

\item The challenger does the following depending on $b \in \bit$:
\begin{itemize}
    \item if $b=0$: the challenger samples a vector $\vec s \rand \Z_q^n$ and errors $\vec e \sim D_{\Z^{m},\,\alpha q}$ and $e' \sim D_{\Z,\,\beta q}$, and sends a Dual-Regev encryption of $\mu \in \bit$ to $\algo A$:
$$
\ct = \left(\vec s^\intercal \vec A + \vec e^\intercal, \vec s^\intercal \vec y + e' + \mu \cdot \lfloor \frac{q}{2} \rfloor\right) \in \Z_q^m \times \Z_q.
$$

\item if $b=1$: the challenger samples $\vec u \rand \Z_q^m$ and $r \rand \Z_q$ uniformly at random and sends the following pair to $\algo A$:
$$
(\vec u,r) \in \Z_q^m \times \Z_q.
$$
\end{itemize}

\item $\algo A$ returns a bit $b' \in \bit$.
\end{enumerate}\\
\hline
\end{tabular}
    \caption{The key-revocable security experiment according to \Cref{def:krpke:security}.}
    \label{fig:Dual-Regev-security}
    \end{center}
\end{figure}

Suppose for the sake of contradiction that $\epsilon(\lambda) $ is non-negligible. Using the equivalence between prediction advantage and distinguishing advantage, we can write
$$ 2 \cdot \left| \Pr\left[ b \leftarrow \expt_{\algo A}(1^{\secparam},b) \, : \, b \rand \bit\right] - \frac{1}{2}\right| =  \eps(\secparam).$$
We show that we can use $\algo A$ to break the $\mathsf{SIS}_{n,q,\sigma\sqrt{2m}}^m$ problem. Without loss of generality, we assume that $\algo A$ submits the plaintext $x=0$. By the assumption that revocation succeeds with overwhelming probability and since $\epsilon(\lambda) \geq 1/\poly(\lambda)$, we can use \Cref{thm:search-to-decision-revoke-1} to argue that
there exists a quantum Goldreich-Levin extractor $\algo E$ that takes as input $\vec A$, $\vec y$ and system $\vec{\textsc{Aux}}$ of the state $\rho_{R,\vec{\textsc{Aux}}}$ and outputs a short vector in the coset $\Lambda_q^{\vec y}(\vec A)$ in time $\poly(\lambda,m,\sigma,q,1/\eps)$ such that
\begin{align*}
&\Pr \left[\substack{
\ineffrevoke(\vec A,\vec y,\sigma,\reg R)  = \top\vspace{0.2mm}\\ \bigwedge \vspace{0.2mm}\\ \mathcal{E}(\vec A,\vec y,\reg{Aux}) \,\,\in\,\, \Lambda_q^{\vec y}(\vec A) \,\cap\, \algo B^m(\vec 0,\sigma \sqrt{\frac{m}{2}})}
 \, : \, \substack{
 \vec A \rand \Z_q^{n \times m}\\
(\ket{\psi_{\vec y}},\vec y) \leftarrow \mathsf{GenGauss}(\vec A,\sigma)\\
\rho_{R,\vec{\textsc{Aux}}} \leftarrow \algo A_{\lambda,\vec A,\vec y}(\proj{\psi_{\vec y}})
    }\right] \, \geq\, \poly(\eps,1/q).
\end{align*}
Here, we rely on the  correctness of $\GenTrap$ in \Cref{thm:gen-trap} and $\mathsf{QSampGauss}$ in \Cref{lem:qdgs}.

\noindent Consider the following procedure in Algorithm \ref{alg:SIS-Solver}.\ \\

\begin{algorithm}[H]\label{alg:SIS-Solver}
\DontPrintSemicolon
\SetAlgoLined
\KwIn{Matrix $\vec A \in \Z_q^{n \times m}$.}
\KwOut{Vector $\vec x \in \Z^m$.}

Generate a Gaussian state $(\ket{\psi_{\vec y}}, \vec y) \leftarrow \mathsf{GenGauss}(\vec A,\sigma)$ with
$$
\ket{\psi_{\vec y}} \,\,= \sum_{\substack{\vec x \in \Z_q^m\\ \vec A \vec x = \vec y \Mod{q}}}\rho_{\sigma}(\vec x)\,\ket{\vec x}
$$
for some vector $\vec y \in \Z_q^n$.\;

Run $\mathcal A$ to generate a bipartite state $\rho_{R \,\vec{\textsc{Aux}}}$ in systems $\algo H_{\reg R} \otimes \algo H_{\textsc{aux}}$ with $\algo H_{\reg R}= \algo H_q^m$.\;    

Measure system $\reg R$ in the computational basis, and let $\vec x_0 \in \Z_q^n$ denote the outcome.\;

Run the quantum Goldreich-Levin extractor $\algo E(\vec A,\vec y,\rho_{\vec{\textsc{aux}}})$ from \Cref{thm:search-to-decision}, where $\rho_{\vec{\textsc{Aux}}}$ is the reduced state in system $\algo H_{\textsc{Aux}}$, and let $\vec x_1 \in \Z_q^n$ denote the outcome.\;

Output the vector $\vec x = \vec x_1 - \vec x_0$.
 \caption{$\mathsf{SIS \_ Solver}(\vec A)$}
\end{algorithm}
\ \\
\ \\
To conclude the proof, we show that $\mathsf{SIS \_ Solver}(\vec A)$ in Algorithm \ref{alg:SIS-Solver} breaks the $\mathsf{SIS}_{n,q,\sigma\sqrt{2m}}^m$ problem whenever $\eps(\lambda) = 1/\poly(\lambda)$. In order to guarantee that $\mathsf{SIS \_ Solver}(\vec A)$ is successful, we use the distinct pair extraction result of \Cref{clm:distinct:preimageext}. This allows us to analyze the probability of simultaneously extracting two distinct short pre-images $\vec x_0\neq\vec x_1$ such that $\vec A \vec x_0 = \vec y = \vec A \vec x_1 \Mod{q}$ -- both in terms of the success probability of revocation and the success probability of extracting a pre-image from the adversary's state $\rho_{\vec{\textsc{Aux}}}$ in system $\algo H_{\textsc{Aux}}$. Assuming that $\vec x_0,\vec x_1$ are distinct short pre-images such that $\|\vec x_0\| \leq \sigma \sqrt{\frac{m}{2}}$ and $\|\vec x_1\| \leq \sigma \sqrt{\frac{m}{2}}$, it then follows that the vector $\vec x = \vec x_1 - \vec x_0$ output by $\mathsf{SIS \_ Solver}(\vec A)$ has norm at most $\sigma\sqrt{2m}$, and thus yields a solution to $\mathsf{SIS}_{n,q,\sigma\sqrt{2m}}^m$.

We remark that the state $\ket{\psi_{{\vec y}}}$ prepared by
Algorithm \ref{alg:SIS-Solver} is not normalized for ease of notation. Note that the tail bound in \Cref{lem:gaussian-tails} implies that (the normalized variant of) $\ket{\psi_{{\vec y}}}$ is within negligible trace distance of the state
with support $\{\vec x \in \Z_q^m : \|\vec x\| \leq \sigma \sqrt{\frac{m}{2}}\}$. Therefore, for the sake of \Cref{clm:distinct:preimageext}, we can assume that $\ket{\psi_{{\vec y}}}$ is a normalized state of the form
  $$
  \ket{\psi_{{\vec y}}} =
    \left(\displaystyle\sum_{\substack{\vec z \in \Z_q^m,\|\vec z\| \leq \sigma \sqrt{\frac{m}{2}}\\
\vec A \vec z = \vec y \Mod{q}}} \rho_{\frac{\sigma}{\sqrt{2}}}(\vec z) \right)^{-\frac{1}{2}}
    \sum_{\substack{\vec x \in \Z_q^m,\|\vec x\| \leq \sigma \sqrt{\frac{m}{2}}\\
\vec A \vec x = \vec y \Mod{q}}}
  \rho_{\sigma}(\vec x)
\ket{\vec x}.
$$
Before we analyze Algorithm \ref{alg:SIS-Solver}, we first make two technical remarks. First, since $\sigma \geq \omega(\sqrt{\log m})$, it follows from \Cref{lem:max:ampl:bound}
that, for any full-rank $\vec A \in \Z_q^{n \times m}$ and for any $\vec y \in \Z_q^n$, we have
$$
\max_{\substack{\vec x \in \Z_q^m, \, \|\vec x\| \leq \sigma \sqrt{\frac{m}{2}}\\
\vec A \vec x = \vec y \Mod{q}}}
\left\{
\frac{\rho_{\frac{\sigma}{\sqrt{2}}}(\vec x)}{\displaystyle\sum_{\substack{\vec z \in \Z_q^m,\|\vec z\| \leq \sigma \sqrt{\frac{m}{2}}\\
\vec A \vec z = \vec y \Mod{q}}} \rho_{\frac{\sigma}{\sqrt{2}}}(\vec z)} \right\} \,\,\leq \,\,  2^{-\Omega(m)}.
$$
Second, we can replace the procedure $\revoke(\vec A,\mathsf{td}_{\vec A},\vec y,\rho_R)$
by an (inefficient) projective measurement $\{\proj{\psi_{\vec y}}, I - \proj{\psi_{\vec y}}\}$, since they produce statistically close outcomes. This follows from the fact that 
$\revoke(\vec A,\mathsf{td}_{\vec A},\vec y,\rho_{\reg R})$ applies the procedure $\mathsf{QSampGauss}$ in Algorithm \ref{alg:SampGauss} as a subroutine, which is correct with overwhelming probability acccording to \Cref{lem:qdgs}.

Let us now analyze the success probability of Algorithm \ref{alg:SIS-Solver}. Putting everything together, we get
\begin{align*}
&\prob\left[
\substack{
\vec x \leftarrow {\sf SIS}\_{\sf Solver}(\vec A)\vspace{1mm}\\
\bigwedge \vspace{0.5mm}\\
\vec x \neq \vec 0 \, \text{ s.t.} \,\, \|\vec x\| \leq \sigma \sqrt{2m}
}\ \, : \,\ \vec A \rand \Zq^{n \times m}\right]  \\
& \geq 
\left( 1 -\max_{\substack{\vec x \in \Z_q^m, \, \|\vec x\| \leq \sigma \sqrt{\frac{m}{2}}\\
\vec A \vec x = \vec y \Mod{q}}}
\left\{
\frac{\rho_{\frac{\sigma}{\sqrt{2}}}(\vec x)}{\displaystyle\sum_{\substack{\vec z \in \Z_q^m,\|\vec z\| \leq \sigma\sqrt{\frac{m}{2}}\\
\vec A \vec z = \vec y \Mod{q}}} \rho_{\frac{\sigma}{\sqrt{2}}}(\vec z)} \right\} \right)\\
&\quad\cdot \Pr \left[
\ineffrevoke({\vec A},{\vec y},\rho_R) = \top \,\, : \,\, 
\substack{
 \vec A \rand \Zq^{n \times m} \text{ s.t. }{\vec A}\text{ is full-rank}\\
(\ket{\psi_{\vec y}},\vec y) \leftarrow \mathsf{GenGauss}(\vec A,\sigma)\\
\rho_{\reg R,\vec{\textsc{Aux}}} \leftarrow \algo A_{\lambda,\vec A,\vec y}(\proj{\psi_{\vec y}})} \right]\\
&\quad\cdot 
\Pr\left[
 \mathcal{E}\big(\vec A,\vec y,\rho_{\vec{\textsc{Aux}}}\big) \,\,\in\,\, \Lambda_q^{\vec y}(\vec A) \cap \algo B^m(\vec 0,\sigma \sqrt{m/2}) \,\, : \,\,
 \substack{
  \vec A \rand \Zq^{n \times m} \text{ s.t. }{\vec A}\text{ is full-rank}\\
(\ket{\psi_{\vec y}},\vec y) \leftarrow \mathsf{GenGauss}(\vec A,\sigma)\\
\rho_{\reg R,\vec{\textsc{Aux}}} \leftarrow \algo A_{\lambda,\vec A,\vec y}(\proj{\psi_{\vec y}})\\
\top \leftarrow \ineffrevoke({\vec A},{\vec y},\rho_{\reg R})
}\right] \\
& \geq \left( 1 - 2^{-\Omega(m)} \right) \cdot  \Pr \left[\substack{
\ineffrevoke(\vec A,\vec y,\sigma,\reg R)  = \top\vspace{0.2mm}\\ \bigwedge \vspace{0.2mm}\\ \mathcal{E}(\vec A,\vec y,\reg{Aux}) \,\,\in\,\, \Lambda_q^{\vec y}(\vec A) \,\cap\, \algo B^m(\vec 0,\sigma \sqrt{\frac{m}{2}})}
 \, : \, \substack{
\vec A \rand \Zq^{n \times m} \text{ s.t. }{\vec A}\text{ is full-rank}\\
(\ket{\psi_{\vec y}},\vec y) \leftarrow \mathsf{GenGauss}(\vec A,\sigma)\\
\rho_{\reg R,\vec{\textsc{Aux}}} \leftarrow \algo A_{\lambda,\vec A,\vec y}(\proj{\psi_{\vec y}})
    }\right] \\
& \geq \left( 1 - 2^{-\Omega(m)} \right)  \cdot \big(  \poly(\eps,1/q) - q^{-n}\big) \quad \geq \quad \poly(\eps,1/q).
\end{align*}
In the last line, we applied the simultaneous search-to-decision reduction from \Cref{thm:search-to-decision-revoke-1} and \Cref{lem:full-rank}.
Therefore, $\mathsf{SIS\_ Solver}(\vec A)$ in Algorithm \ref{alg:SIS-Solver} runs in time $\poly(q,1/\eps)$ and solves $\mathsf{SIS}_{n,q,\sigma\sqrt{2m}}^m$ whenever $\eps = 1/\poly(\lambda)$. Therefore, we conclude that $\eps(\lambda)$ must be negligible.
\end{proof}

\subsection{Proof of Theorem~\ref{thm:security-Dual-Regev}}

\begin{proof}
The proof is the same as in \Cref{thm:security-Dual-Regev-high-revoke}, except that we invoke 
\Cref{thm:search-to-decision} instead of \Cref{thm:search-to-decision-revoke-1} to argue that simultaneous extraction succeeds with sufficiently high probability.
\end{proof}
\newpage
\section{Key-Revocable Dual-Regev Encryption with Classical Revocation}

Recall that our key-revocable Dual-Regev public-key encryption scheme in \Cref{cons:dual-regev} requires that a quantum state is \emph{returned} as part of revocation. In this section, we present a Dual-Regev public-key encryption scheme with \emph{classical key revocation} (\Cref{cons:dual-regev-classical-revoc}).

To prove security, we follow a similar proof as in our previous construction from \Cref{sec:dual-regev}. As an additional technical ingredient, we rely on a strong type of \emph{collapsing} property Ajtai hash recently proven by Bartusek, Khurana and Poremba~\cite[Theorem 5.5]{bartusek2023publiclyverifiable}, which relies on the hardness of subexponential $\LWE$ and $\SIS$.

\subsection{Definition: Public-Key Encryption with Classical Key Revocation}

Let us now give a formal definition of a public-key encryption scheme with classical key revocation.

\begin{definition}[Key-Revocable Public-Key Encryption with Classical Revocation]\label{def:classical-key-revocable-PKE}

Let $\lambda \in \N$ be the security parameter. A key-revocable public-key encryption scheme with classical revocation consists efficient algorithms $\left(\setup,\enc,\dec,\Delete,\delete\right)$, where $\enc$ and $\revoke$ are $\mathsf{PPT}$ algorithms, and $\setup,\dec$ and $\Delete$ are $\QPT$ algorithms defined as follows:
\begin{itemize}
    \item $\setup(1^{\secparam})$: given as input a security parameter $\secparam$, output a public key $\pk$, a master secret key $\msk$ and a quantum decryption key $\sk$. 
    \item $\enc(\pk,x)$: given a public key $\pk$ and plaintext $x \in \{0,1\}^{\ell}$, output a ciphertext $\ct$. 
    \item $\dec(\sk,\ct)$: given a decryption key $\sk$ and ciphertext $\ct$, output a message $y$. 
    \item $\Delete(\sk)$: given a quantum decryption key, it outputs a classical certificate $\pi$.
    \item $\delete\left(\pk,\msk,\pi\right)$: given as input a master secret key $\msk$, a public key $\pk$ and a certificate $\pi$, output $\valid$ or $\invalid$. 
\end{itemize}
\end{definition}

\paragraph{Correctness of Decryption.} For every $x \in \{0,1\}^{\ell}$, the following holds: 
$$\prob\left[ x \leftarrow \dec(\sk,\ct)\ :\ \substack{(\pk,\msk,\sk) \leftarrow \setup(1^{\secparam})\\ \ \\ \ct \leftarrow \enc(\pk,x)} \right] \geq 1 - \nu(\secparam),$$
where $\nu(\cdot)$ is a negligible function. 

\paragraph{Correctness of Revocation.} The following holds:
$$\prob\left[ \valid \leftarrow \delete\left(\pk,\msk,\pi \right)\ :\ 
\substack{
(\pk,\msk,\sk) \leftarrow \setup(1^{\secparam})\\
\pi \leftarrow \Delete(\sk)}\right] \geq 1 - \nu(\secparam),$$
where $\nu(\cdot)$ is a negligible function. 

Our security definition for key-revocable public-key encryption is as follows.

\begin{figure}[!htb]
   \begin{center} 
   \begin{tabular}{|p{14cm}|}
    \hline 
\begin{center}
\underline{$\expt_{\Sigma,\algo A}\left( 1^{\secparam},b \right)$}: 
\end{center}
\noindent {\bf Initialization Phase}:
\begin{itemize}
    \item The challenger runs $(\pk,\msk,\sk) \leftarrow \setup(1^{\secparam})$ and sends $(\pk,\sk)$ to $\adversary$. 
\end{itemize}
\noindent {\bf Revocation Phase}:
\begin{itemize}
    \item The challenger sends the message {\texttt{REVOKE}} to $\adversary$. 
    \item The adversary $\adversary$ returns a certificate $\pi$.
    \item The challenger aborts if $\revoke(\pk,\msk,\pi)$ outputs $\invalid$.
\end{itemize}
\noindent {\bf Guessing Phase}:
\begin{itemize}
    \item $\adversary$ submits a plaintext $x \in \{0,1\}^{\ell}$ to the challenger.
    \item If $b=0$: The challenger sends $\ct \leftarrow \enc(\pk,x)$ to $\adversary$. Else, if $b=1$, the challenger sends $\ct \xleftarrow{\$} {\cal C}$, where ${\cal C}$ is the ciphertext space of $\ell$ bit messages. 
    \item Output $b_{\adversary}$ if the output of $\adversary$ is $b_{\adversary}$.  
\end{itemize}
\ \\ 
\hline
\end{tabular}
    \caption{Security Experiment for Classical Key Revocation}
    \label{fig:securityexpt-classical}
    \end{center}
\end{figure}

\begin{definition}
\label{def:krpke:security-classical}
A public-key encryption scheme $\Sigma =\left(\setup,\enc,\dec,\Delete,\delete\right)$ with classical key revocation is $(\eps,\delta)$-secure if, for every $\QPT$ adversary $\adversary$ with 
$$\prob[\invalid \leftarrow \expt_{\Sigma,\algo A}(1^\lambda,b)] \leq \delta(\secparam)
$$
for $b \in \bit$, it holds that
$$ \left| \Pr\left[ 1 \leftarrow \expt_{\Sigma,\algo A}(1^{\secparam},0) \right] - \Pr\left[ 1 \leftarrow \expt_{\Sigma,\algo A}(1^{\secparam},1) \right]\right|\leq  \eps(\secparam),$$
where $\expt_{\Sigma,\algo A}(1^{\secparam},b)$ is as defined in~\Cref{fig:securityexpt-classical}. If $\delta(\secparam) = 1-1/\poly(\lambda)$ and $\eps(\secparam)=\negl(\lambda)$, we simply say the key-revocable public-key encryption scheme is secure. 
\end{definition}

\subsection{Construction.}

\paragraph{Preliminaries.} Let us first introduce some additional background on lattice trapdoors.

\begin{definition}[$\beta$-good trapdoor]\label{def:beta-good-lattice-td}
Let $\vec A \in \Z_q^{n \times m}$ be a matrix. A $\beta$-good lattice trapdoor $\vec T \in \Z^{m \times m}$ for the matrix $\vec A$ is a short basis for the lattice $\Lambda_q^\perp(\vec A)$ with the following properties:
\begin{enumerate}
    \item Each column vector of $\vec T$ is in the right kernel of $\vec A$ (mod $q$), i.e. $\vec A \cdot \vec T = \vec 0 \Mod{q}$.

    \item Each column vector of $\vec T = [\vec t_1, \dots, \vec t_m]$ is short, i.e. $\| \vec t_i \| \leq \beta$, for all $i \in [m]$.

    \item The matrix $\vec T$ has full rank $m$ over $\mathbb{R}$.
\end{enumerate}   
\end{definition}

We make use of the following specific variant of a lattice trapdoor for the $q$-ary lattice $\Lambda_q^\perp(\vec A)$ due to Alwen and Peikert~\cite{cryptoeprint:2008/521}.

\begin{lemma}[AP lattice trapdoor, \cite{cryptoeprint:2008/521}, Lemma 3.5]\label{lem:gen-trap-AP} Let  $n,m \in \N$ and $q \in \N$ be a prime such that $m = \lceil 6 n \log q\rceil$. There exists a randomized algorithm with the following properties:
\begin{itemize}
    \item $\GenTrap(1^n,1^m,q)$: on input $1^n, 1^m$ and $q$, returns a matrix $\vec A \in \Z_q^{n \times m}$ and a short trapdoor matrix $\vec T_{\vec A} \in \Z^{m \times m}$ such that the distribution of $\vec A$ is negligibly (in the parameter $n$) close to uniform and $\vec T_{\vec A}$ is a  $(20 n \log q)$-good lattice trapdoor with all but negligible probability.
    \end{itemize}
\end{lemma}

Let us now present our Dual-Regev encryption scheme with classical key revocation

\paragraph{Parameters.} Let $\lambda \in \N$ be the security parameter. Let $n,m \in \N$ and $q \in \N$ be a prime with $q= 2^{o(n)}$ such that $m = \lceil 6 n \log q\rceil$ and $\nu = 64m^2$, each parameterized by $\lambda$. Choose a parameter $\sigma \in \mathbb{R}$ such that $\sqrt{8m}< \sigma <q/\sqrt{8m}$ and $\sigma/\nu = 2^{o(n)}$. Moreover, let $\alpha,\beta \in (0,1)$ be noise ratios with the property that $\beta/\alpha =2^{o(n)}$ such that $1/\alpha= 2^{o(n)} \cdot \sigma$.

\begin{construction}[Dual-Regev Encryption with Classical Key Revocation]\label{cons:dual-regev-classical-revoc} Let $\lambda \in \N$ be the security parameter. The key-revocable scheme $\mathsf{CRevDual} = (\keygen,\enc,\dec,\mathsf{Delete},\mathsf{Revoke})$ with classical revocation consists of the following $\QPT$ algorithms:
\begin{itemize}
\item $\keygen(1^\lambda) \rightarrow (\pk,\sk,\msk):$ sample $(\vec A \in \Z_q^{n \times m},\vec T_{\vec A} \in \Z^{m \times m}) \leftarrow \GenTrap(1^n,1^m,q)$ using the procedure in \Cref{lem:gen-trap-AP}, sample a random vector $\vec v \rand \bit^m$, generate a Gaussian superposition $(\ket{\psi_{\vec y}}, \vec y \in \Z_q^n) \leftarrow \mathsf{GenGauss}(\vec A,\sigma)$, and let
$$
\ket{\psi_{\vec y}^{\vec v}} = \vec Z_q^{\lfloor \frac{q}{\nu} \rfloor \cdot \vec v} \ket{\psi_{\vec y}}= \sum_{\substack{\vec x \in \Z_q^m\\ \vec A \vec x = \vec y}}\rho_{\sigma}(\vec x)\, \omega_q^{\langle \vec x,\lfloor \frac{q}{\nu} \rfloor \cdot \vec v \rangle} 
 \, \ket{\vec x}.
$$
Output $\pk = (\vec A,\vec y)$, $\sk = \ket{\psi_{\vec y}^{\vec v}} $ and $\msk = (\vec T_{\vec A},\vec v)$.
\item $\Enc(\pk,\mu) \rightarrow \ct:$ to encrypt a bit $\mu \in \bit$, sample a random vector $\vec s \rand \Z_q^n$ and errors $\vec e \sim D_{\Z^{m},\,\alpha q}$ and $e' \sim D_{\Z,\,\beta q}$, and output the ciphertext pair
$$
\ct = \left(\vec s^\intercal \vec A + \vec e^\intercal \Mod{q}, \vec s^\intercal \vec y + e' + \mu \cdot \lfloor \frac{q}{2} \rfloor \Mod{q} \right) \in \Z_q^m \times \Z_q.
$$
\item $\Dec(\sk,\ct) \rightarrow \bit:$ to decrypt $\ct$, 
apply the unitary $U: \ket{\vec x}\ket{0} \rightarrow \ket{\vec x}\ket{ \ct\cdot (-\vec x,1)^\intercal}$ on input $\ket{\psi_{\vec y}^{\vec v}}\ket{0}$, where $\sk=\ket{\psi_{\vec y}^{\vec v}}$, and measure the second register in the computational basis. Output $0$, if the measurement outcome
is closer to $0$ than to $\lfloor \frac{q}{2} \rfloor$,
and output $1$, otherwise.\footnote{This procedure can be purified via the Gentle Measurement Lemma.} 

\item $\mathsf{Delete}(\rho) \rightarrow \vec w$: apply the $m$-qudit $q$-ary quantum Fourier transform to the state $\rho$, and measure in the computational basis to obtain an outcome $\vec w \in \Z_q^m$.

\item $\mathsf{Revoke}(\msk,\vec w) \rightarrow \{ \top,\bot\}$: to verify the certificate $\vec w \in \Z_q^m$, do the following:
\begin{enumerate}
    \item Parse $(\vec T_{\vec A} \in \Z^{m \times m},\vec v \in \bit^m) \leftarrow \msk$. 
    \item Compute $\vec c^\intercal = \vec w^\intercal \cdot \vec T_{\vec A} \Mod{q}$.

    \item Compute $\vec d^\intercal = \vec c^\intercal \cdot \vec T_{\vec A}^{-1}$, where $\vec T_{\vec A}^{-1}$ is the inverse matrix of $\vec T_{\vec A}$ over $\mathbb{R}$.\footnote{The matrix $\vec T_{\vec A} \in \Z^{m \times m}$ output by $\GenTrap(1^n,1^m,q)$ has full rank $m$ over  $\mathbb{R}$ with overwhelming probability.}

    \item For $i \in [m]$, let $v'_i = 0$, if $d_i \in [-q/\sigma,q/\sigma]$, and let $v'_i = 1$, otherwise.
    \item Output $\top$,  if $\vec v' = \vec v$, and output $\bot$ otherwise.
\end{enumerate}
\end{itemize}
\end{construction}

\paragraph{Correctness of verification.}

\begin{lemma} Let $\lambda \in \N$ be the security parameter. Then, by our choice of parameters.
$$\prob\left[ \top \leftarrow \mathsf{Revoke}(\msk,\vec w)\ :\ \substack{(\pk,\msk,\sk) \leftarrow \keygen(1^{\secparam})\\ \ \\ 
\vec w \leftarrow \mathsf{Delete}(\sk)}\right] \geq 1 - \negl(\secparam)$$
\end{lemma}
\begin{proof}
Recall that $\keygen(1^\lambda)$ samples $(\vec A \in \Z_q^{n \times m},\vec T \in \Z^{m \times m}) \leftarrow \GenTrap(1^n,1^m,q)$ and $\vec v \rand \bit^m$, generates a Gaussian superposition $(\ket{\psi_{\vec y}}, \vec y \in \Z_q^n) \leftarrow \mathsf{GenGauss}(\vec A,\sigma)$ and assigns $\pk = (\vec A,\vec y)$, $\sk = \ket{\psi_{\vec y}^{\vec v}} $ and $\msk = (\vec T,\vec v)$, where
$$
\ket{\psi_{\vec y}^{\vec v}} = \vec Z_q^{\lfloor \frac{q}{\nu} \rfloor \cdot \vec v} \ket{\psi_{\vec y}}= \sum_{\substack{\vec x \in \Z_q^m\\ \vec A \vec x = \vec y}}\rho_{\sigma}(\vec x)\, \omega_q^{\langle \vec x,\lfloor \frac{q}{\nu} \rfloor \cdot \vec v \rangle} 
 \, \ket{\vec x}.
$$
From \Cref{lem:gen-trap-AP} and \Cref{lem:full-rank}, it follows that $\vec A$ is full rank and that $\vec T$ is a  $(20 n \log q)$-good lattice trapdoor  with overwhelming probability.
Moreover, from \Cref{lem:duality}, it follows that $\vec w \leftarrow \mathsf{Del}(\sk)$ results in a vector $\vec w^\intercal = \hat{\vec s}^\intercal \vec A + \lfloor \frac{q}{\nu} \rfloor \cdot \vec v^\intercal  + \hat{\vec e}^\intercal \in \Z_q^m$ with $\hat{\vec e} \sim D_{\Z_q^m,q/\sigma}$. Therefore, we have that with overwhelming probability $\|\hat{\vec e}\|_2 \leq \sqrt{m}q/\sigma$ and thus
\begin{align*}
\|\lfloor \frac{q}{\nu} \rfloor \cdot \vec v^\intercal \cdot \vec T  + \hat{\vec e}^\intercal \cdot \vec T\|_\infty &\leq  \lfloor \frac{q}{\nu} \rfloor \cdot \|\vec v^\intercal \vec T\|_2 + \|\hat{\vec e}^\intercal \vec T\|_2 \\
&\leq \lfloor \frac{q}{\nu} \rfloor \cdot \|\vec v\|_2 \cdot \|\vec T\|_2 + \|\hat{\vec e}\|_2 \cdot \|\vec T\|_2\\
&\leq \lfloor \frac{q}{\nu} \rfloor \cdot \sqrt{m} \cdot \|\vec T\|_2 + \sqrt{m} \cdot \frac{q}{\sigma}\cdot \|\vec T\|_2\\
&\leq \left(\lfloor \frac{q}{\nu} \rfloor + \frac{q}{\sigma} \right) \cdot m \cdot \|\vec T\| \, \leq \, q/4.
\end{align*}
Consequently, $\vec c^\intercal = \vec w^\intercal \cdot \vec T = \lfloor \frac{q}{\nu} \rfloor \cdot \vec v^\intercal \cdot \vec T  + \hat{\vec e}^\intercal \cdot \vec T\,\in \, \Z^m \cap (-\frac{q}{4},\frac{q}{4}]^m$ and thus 
$$
\vec d^\intercal = \vec c^\intercal \cdot \vec T^{-1} = \lfloor \frac{q}{\nu} \rfloor \cdot \vec v^\intercal  + \hat{\vec e}^\intercal.$$
By our choice of parameters, we have that $\lfloor \frac{q}{\nu} \rfloor \gg \frac{q}{\sigma}$, and thus the procedure $\mathsf{Revoke}(\msk,\vec w)$ successfully decodes $\vec v$ from $\vec d$, and outputs $\top$ with overwhelming probability.

\end{proof}

\subsection{Proof of Security.}

We first discuss some relevant background which we use for the security proof.

\paragraph{Everlasting collapsing property of the Ajtai hash function.} We need the following result from Bartusek, Khurana and Poremba~\cite[Theorem 5.5]{bartusek2023publiclyverifiable}. Informally, it says that once the adversary outputs a short pre-image (given as input either a Gaussian superposition of pre-images, or a single measured pre-image), then the adversary cannot tell which state it initially received even if it is allowed to be unbounded for the rest of the experiment. This property is called \emph{certified-everlasting target-collapsing}\footnote{It's called ``everlasting'' because the adversary can be unbounded once a pre-image is presented and it's called ``target-collapsing'' because it is a weakening of collapsing, where the image is honestly generated by the challenger.} and it can be shown assuming the hardness of $\LWE$/$\SIS$.

\begin{theorem}[Certified-everlasting target-collapsing property of the Ajtai hash, \cite{bartusek2023publiclyverifiable}]\label{thm:ev-target-collapse-ajtai}
Let $n\in \N$ and $q$ be a prime modulus with $q=2^{o(n)}$ and $m \geq 2n \log q$, each parameterized by the security parameter $\lambda \in \N$. Let $q/\sqrt{8m} > \sigma > \sqrt{8m}$ and let $\alpha \in (0,1)$ be a noise ratio such that $1/\alpha= 2^{o(n)} \cdot \sigma$.
Then, assuming the hardness of the $\mathsf{LWE}_{n,q,\alpha q}^m$ and $\SIS_{n,q,\sigma\sqrt{2m}}^m$, it holds for any adversary
$\algo A = (\algo A_0,\algo A_1)$ consisting of a $\QPT$ algorithm $\algo A_0$ and an unbounded algorithm $\algo A_1$:
$$
|\Pr\left[\mathsf{EvTargetCollapseExpt}_{\algo A,\lambda}(0) = 1 \right] - \Pr\left[\mathsf{EvTargetCollapseExpt}_{\algo A,\lambda}(1) = 1\right]| \leq \negl(\lambda).
$$
Here, the experiment $\mathsf{EvTargetCollapseExpt}_{\algo A,\lambda}(b)$ is defined as follows:

\begin{enumerate}

    \item The challenger generates a Gaussian superposition $(\ket{\psi_{\vec y}}, \vec y \in \Z_q^n) \leftarrow \mathsf{GenGauss}(\vec A,\sigma)$, with
    $$
    \ket{\psi_{\vec y}}= \sum_{\substack{\vec x \in \Z_q^m\\ \vec A \vec x = \vec y}}\rho_{\sigma}(\vec x) \, \ket{\vec x}.
    $$
    \item If $b=0$, the challenger does nothing. Else, if $b=1$, the challenger measures the state in the computational basis. Next, the challenger sends the resulting state together with $y$ to $\algo A_0$.
    \item $\algo A_{0}$ sends a classical certificate $\vec w \in \Z_q^m$ to the challenger and initializes $\algo A_{1}$ with its residual (internal) state.
    \item The challenger checks if $\vec w$ satisfies $\vec A \cdot \vec w = \vec y \Mod{q}$ and $\| \vec w \| \leq \sigma \sqrt{m/2}$. If true, $\algo A_{1}$ is run until it outputs a bit $b'$. Otherwise, $b' \gets \{0,1\}$ is sampled uniformly at random. The output of the experiment is $b'$.

    
\end{enumerate}
\end{theorem}

We immediately obtain result.

\begin{theorem}[Uncertainty principle for Gaussian states]\label{cor:uncertainty}

Let $n\in \N$ and $q$ be a prime modulus with $q=2^{o(n)}$ and $m \geq 2n \log q$, each parameterized by the security parameter $\lambda \in \N$. Let $q/\sqrt{8m} > \sigma > \sqrt{8m}$ and let $\alpha \in (0,1)$ be a noise ratio such that $1/\alpha= 2^{o(n)} \cdot \sigma$, and let $\nu >0$.
Then, assuming the subexponential hardness of the $\mathsf{LWE}_{n,q,\alpha q}^m$, it holds for any adversary
$\algo A = (\algo A_0,\algo A_1)$ consisting of a $\QPT$ algorithm $\algo A_0$ and an unbounded algorithm $\algo A_1$:
$$
\Pr[\mathsf{UncertaintyExpt}_{\algo A}(1^\lambda) =1] \leq \negl(\lambda),
$$
where $\mathsf{UncertaintyExpt}_{\algo A}(1^\lambda)$ is the following experiment between a challenger and the adversary $\algo A$:
\begin{enumerate}
    \item The challenger samples a random matrix $\vec A \rand \Z_q^{n \times m}$ and vector $\vec v \rand \bit^m$, generates a Gaussian superposition $(\ket{\psi_{\vec y}}, \vec y \in \Z_q^n) \leftarrow \mathsf{GenGauss}(\vec A,\sigma)$ and prepares
$$
\ket{\psi_{\vec y}^{\vec v}} = \vec Z_q^{\lfloor \frac{q}{\nu} \rfloor \cdot \vec v} \ket{\psi_{\vec y}}= \sum_{\substack{\vec x \in \Z_q^m\\ \vec A \vec x = \vec y}}\rho_{\sigma}(\vec x)\, \omega_q^{\langle \vec x,\lfloor \frac{q}{\nu} \rfloor \cdot \vec v \rangle} 
 \, \ket{\vec x}.
$$
Then, the challenger sends $(\vec A,\vec y,\ket{\psi_{\vec y}^{\vec v}})$ to $\algo A_0$.
\item $\algo A_0$ generates a vector $\vec x_0 \in \Z_q^m$ and a polynomial-sized quantum state $\rho_{\textsc{Aux}}$. Then, $\algo A_0$ sends $\vec x_0$ to the challenger and forwards $(\vec A,\vec y,\rho_{\textsc{Aux}})$ to $\algo A_1$.

\item $\algo A_1$ generates a vector $\vec v' \in \Z_q^m$ on input $(\vec A,\vec y,\rho_{\textsc{Aux}})$ and sends it to the challenger.

\item The challenger checks if $\vec v' = \vec v$ and whether the vector $\vec x_0$ satisfies $\vec A \cdot \vec x_0 = \vec y \Mod{q}$ and $\| \vec x_0 \| \leq \sigma \sqrt{m/2}$. If true, the challenger outputs $1$, otherwise the challenger outputs $0$. The bit output by the challenger also denotes the outcome of the experiment.
\end{enumerate}
\end{theorem}

\begin{proof}
Let $\algo A = (\algo A_0,\algo A_1)$ be an adversary consisting of a $\QPT$ algorithm $\algo A_0$ and an unbounded algorithm $\algo A_1$. We consider the following two hybrids:

\begin{description}
  
\item $\hybrid_0$: This is the original experiment $\mathsf{UncertaintyExpt}_{\algo A}(1^\lambda)$.

\item $\hybrid_1$: This is the same experiment as before, except that the challenger additionally measures the state $\ket{\psi_{\vec y}^{\vec v}}$ in Step $1$ before sending it to $\algo A_0$.
\end{description}

To complete the proof, we show the following:

\begin{claim}
$$
|\Pr\left[\hybrid_0 = 1\right]-\Pr\left[\hybrid_1=1\right]| \leq \negl(\lambda).
$$
\end{claim}
\begin{proof}
This follows from the certified-everlasting target-collapsing property of the Ajtai hash in \Cref{thm:ev-target-collapse-ajtai}. Suppose for the sake of contradiction that the advantage is $1/\poly(\lambda)$. We consider the following adversary $\algo B = (\algo B_0,\algo B_1)$ consisting of a $\QPT$ algorithm $\algo B_0$ and an unbounded algorithm $\algo B_1$ that distinguishes $\mathsf{EvTargetCollapseExpt}_{\algo B,\lambda}(1^\lambda,b)$ for $b \in \bit$ with advantage $1/\poly(\lambda)$.
\begin{itemize}
    \item $\algo B_0$ receives as input $(\vec A,\vec y,\ket{\psi_{\vec y}})$ from the challenger, samples $\vec v \rand \bit^m$, generates the state
    $\ket{\psi_{\vec y}^{\vec v}} = \vec Z_q^{\lfloor \frac{q}{\nu} \rfloor \cdot \vec v} \ket{\psi_{\vec y}}$
and then forwards $(\vec A,\vec y,\ket{\psi_{\vec y}^{\vec v}})$ to $\algo A_0$.

\item When $\algo A_0$ outputs a vector $\vec x_0 \in \Z_q^m$ and a polynomial-size advice state $\rho_{\textsc{Aux}}$, $\algo B_0$ forwards $\vec x_0$ to the challenger, and initializes $\algo B_1$ with $\rho_{\textsc{Aux}}$.

\item $\algo B_1$ runs $\algo A_1$ with auxuliary input $\rho_{\textsc{Aux}}$. When $\algo A_1$ outputs a vector $\vec v'$, the procedure $\algo B_1$ checks if $\vec v' = \vec v$. If true, $\algo B_1$ outputs $1$. Else, $\algo B_1$ outputs $0$ otherwise.
\end{itemize}
Note that if $b=0$, the state $\ket{\psi_{\vec y}^{\vec v}}$ prepared by $\algo B_0$ corresponds to  
$$
\ket{\psi_{\vec y}^{\vec v}}= \sum_{\substack{\vec x \in \Z_q^m\\ \vec A \vec x = \vec y}}\rho_{\sigma}(\vec x)\, \omega_q^{\langle \vec x,\lfloor \frac{q}{\nu} \rfloor \cdot \vec v \rangle} 
 \, \ket{\vec x},
$$
and thus matches the correct input from $\hybrid_0$.
Moreover, if $b=1$, the state is collapsed and thus matches the correct input from $\hybrid_1$. Crucially, the Pauli-Z operator $\vec Z_q^{\lfloor \frac{q}{\nu} \rfloor \cdot \vec v}$ has no effect on the collapsed state and merely introduces a global phase.
Therefore,
$\algo B = (\algo B_0,\algo B_1)$ distinguishes $\mathsf{EvTargetCollapseExpt}_{\algo B,\lambda}(1^\lambda,b)$ for $b \in \bit$ with advantage at least $1/\poly(\lambda)$, which results in the desired contradiction.
 \end{proof}
Because the vector $\vec v$ is completely hidden from the view of the adversary in $\hybrid_1$, the probability that $\algo A_1$ guesses it correctly is at most $2^{-m}$, which is negligible. This proves the claim.
\end{proof}

\paragraph{Simultaneous search-to-decision reduction with classical revocation}

We give a strengthening of our result in \Cref{thm:search-to-decision-without-revoke} and state a \emph{simultaneous} search-to-decision reduction with quantum auxiliary input which holds even if additionally require that a \emph{classical revocation} procedure succeeds. This is essentially a classical revocation analogue of \Cref{thm:search-to-decision}.

To prove the theorem, we require the following analogue of \Cref{conjecture-I}.

\begin{conjecture}\label{conjecture-II}
Let $\lambda \in \N$. Then, there exist parameters (each parameterized by $\lambda$) such that $n\in \N$, $q$ is a prime with $q=2^{o(n)}$, $m \geq 2n \log q$, $\sqrt{8m}< \sigma <q/\sqrt{8m}$, $\alpha \in (0,1)$ with $1/\alpha= 2^{o(n)} \cdot \sigma$ and $\nu = 64m^2$ for which the following holds: for any $\QPT$ $\algo A$ and $\QPT$ $\algo C$ (and for a fixed algorithm $\algo B$), and for any $\poly(\lambda)$-sized quantum auxiliary input $\tau_\lambda$ (which depends on $\lambda$):
$$
\vline \, \Pr[1 \leftarrow \expt_{\algo A,\algo B,\algo C}(1^{\secparam},0)]  - \Pr[1 \leftarrow \expt_{\algo A,\algo B,\algo C}(1^{\secparam},1)] \, \vline \leq \negl(\lambda),
$$
where $\expt_{\algo A,\algo B,\algo C}(1^{\secparam},b)$ is the cloning experiment in \Cref{fig:Expt-Conjecture1}.
\end{conjecture}

\begin{figure}
   \begin{center} 
   \begin{tabular}{|p{14cm}|}
    \hline 
\begin{center}
\underline{$\expt_{\algo A,\algo B,\algo C}(1^{\secparam},b)$}: 
\end{center}
\begin{enumerate}
        \item The challenger samples $(\vec A \in \Z_q^{n \times m},\vec T_{\vec A}) \leftarrow \GenTrap(1^n,1^m,q)$ and a Gaussian vector $\vec x_0 \sim D_{\Z_q^m,\frac{\sigma}{\sqrt{2}}}$ and lets $\vec y= \vec A \cdot \vec x_0 \Mod{q}$. Then, the challenger runs $\ket{\psi_{\vec y}} \leftarrow \mathsf{QSampGauss}(\vec A,\vec T_{\vec A},\vec y,\sigma)$, samples $\vec v \rand \bit^m$ and sends $(\vec A,\vec y,\ket{\psi_{\vec y}^{\vec v}})$ to $\algo A$, where $\ket{\psi_{\vec y}^{\vec v}} = \vec Z_q^{\lfloor \frac{q}{\nu} \rfloor \cdot \vec v} \ket{\psi_{\vec y}}$.

        \item $\algo A$ receives $(\vec A,\vec y,\ket{\psi_{\vec y}^{\vec v}})$ together with auxiliary input $\tau_\lambda$ and generates a state $\rho_{\reg B \reg C}$ in systems $\reg{BC}$ where $\reg B$ is classical, and sends $\reg B$ to $\algo B$ and $\reg C$ to $\algo C$.

        \item The challenger sends $(\vec T_{\vec A},\vec v,\sigma)$ to $\algo B$ and, depending on the value of $b$, the challenger sends the following to $\algo C$:
        \begin{itemize}
            \item if $b=0$: the challenger samples $\vec s \rand \Zq^n$, $\vec e \sim D_{\Z^m,\alpha q}$, lets $\vec u = \vec A^\intercal \vec s + \vec e$, and sends $(\vec A,\vec y,\vec u,\vec u^\intercal \vec x_0 \Mod{q})$ to $\algo C$.

            \item if $b=0$: the challenger samples a uniformly random vector $\vec u \rand \Z_q^m$ and sends $(\vec A,\vec y,\vec u,\vec u^\intercal \vec x_0 \Mod{q})$ to $\algo C$.
        \end{itemize}

        \item $\algo B$ receives as input $(\vec T_{\vec A},\vec v,\sigma,\vec w)$, where $\vec w$ is the content of the classical register $\reg B$, and applies $\mathsf{Revoke}(\vec T_{\vec A},\vec v,\vec w)$.  
        $\algo B$ outputs $\top$, if it succeeds, else outputs $\bot$.
        
       \item $\algo C$ receives as input $(\vec A,\vec y,\vec u,\vec u^\intercal \vec x_0 \Mod{q},\reg C)$ and outputs a bit $b'$.

       \item The challenger outputs $1$, if $\algo B$ outputs $\top$ and $\algo C$ outputs $b'=b$. This is also the outcome of the experiment.
    \end{enumerate}\\
\hline
\end{tabular}
    \caption{The experiment for Conjecture 2.}
    \label{fig:Expt-Conjecture2}
    \end{center}
\end{figure}

We prove the following statement.

\begin{theorem}[Simultaneous Search-to-Decision Reduction with Classical Revocation]\label{thm:search-to-decision-classical}
Let $\lambda\in \N$ be the security parameter. Let $n\in \N$, $q$ be a prime with $q=2^{o(n)}$, $m \geq 6n \log q$, $\sqrt{8m}< \sigma <q/\sqrt{8m}$, $\alpha,\beta \in (0,1)$ with $\beta/\alpha = 2^{o(n)}$, $1/\alpha= 2^{o(n)} \cdot \sigma$ and $\nu = 64m^2$, such that the following holds: Let
$\algo A = \{(\algo A_{\lambda,\vec A,\vec y},\tau_\lambda)\}_{\lambda \in \N}$
be any non-uniform quantum algorithm consisting of a family of polynomial-sized quantum circuits
$$\Bigg\{\algo A_{\lambda,\vec A,\vec y}: \algo L(\algo H_q^m \otimes \algo H_{B_\lambda}) \allowbreak\rightarrow \algo L(\algo H_{R_\lambda} \otimes \algo H_{\textsc{aux}_\lambda})\Bigg\}_{\vec A \in \Zq^{n \times m}, \,\vec y \in \Z_q^n}$$
and polynomial-sized advice states $\tau_\lambda \in \algo D(\algo H_{B_\lambda})$ which are independent of $\vec A$.
Then, assuming \Cref{conjecture-II} is true, the following statement holds for every $\QPT$ distinguisher $\mathcal{D}$. Suppose there exists a function $\eps(\lambda) = 1/\poly(\lambda)$ such that
\begin{align*}
&\vline \, \prob\left[ 1 \leftarrow \mathsf{SimultSearchToDecisionExpt}^{\adversary,\algo D}(1^{\secparam},0)\right]-\\
&\,\prob\left[ 1 \leftarrow \mathsf{SimultSearchToDecisionExpt}^{\adversary,\algo D}(1^{\secparam},1)\right]  \,\vline = \eps(\lambda).
\end{align*}
\begin{figure}[!htb]
   \begin{center} 
   \begin{tabular}{|p{14cm}|}
    \hline 
\begin{center}
\underline{$\mathsf{SimultSearchToDecisionExpt}^{\adversary,\algo D}\left( 1^{\secparam},b \right)$}: 
\end{center}
\begin{itemize}
\item If $b=0$: output $\simult.\lwe.\dist^{\adversary,\algo D}\left( 1^{\secparam} \right)$ defined in \Cref{fig:lwe-dist-classical}.
\item If $b=1$: output $\simult.\unif.\dist^{\adversary,\algo D}\left( 1^{\secparam} \right)$ defined in \Cref{fig:unif.dist-classical}.
\end{itemize}
\ \\
\hline
\end{tabular}
    \caption{The experiment $\mathsf{SimultSearchToDecisionExpt}^{\adversary,\algo D}\left( 1^{\secparam},b \right)$.}
    \label{fig:simult-search-to-decision-classical}
    \end{center}
\end{figure}

\noindent Then, there exists a quantum extractor $\algo E$ that takes as input $\vec A$, $\vec y$ and system $\vec{\textsc{Aux}}$ of the state $\rho_{\reg R,\vec{\textsc{Aux}}}$ and outputs a short vector in the coset $\Lambda_q^{\vec y}(\vec A)$ in time $\poly(\lambda,m,\sigma,q,1/\eps)$ such that
\begin{align*}
&\Pr \left[\substack{
\mathsf{Revoke}(\vec T_{\vec A},\vec v,\vec w) = \top \\ 
\ \\
\bigwedge\\
\ \\
\mathcal{E}(\vec A,\vec y,\rho_{\vec{\textsc{Aux}}}) \, \in \, \Lambda_q^{\vec y}(\vec A) \,\cap\, \algo B^m(\vec 0,\sigma \sqrt{\frac{m}{2}})}
 \, : \, \substack{
 (\vec A, \vec T_{\vec A}) \leftarrow \GenTrap(1^n,1^m,q)\\
(\ket{\psi_{\vec y}}, \vec y) \leftarrow \mathsf{GenGauss}(\vec A,\sigma)\\
\vec v \rand \bit^m, \, \, \ket{\psi_{\vec y}^{\vec v}} \leftarrow \vec Z_q^{\lfloor \frac{q}{\nu} \rfloor \cdot \vec v} \ket{\psi_{\vec y}}\\
\proj{\vec w} \otimes \rho_{\textsc{Aux}} \leftarrow \algo A_{\lambda,\vec A,\vec y}(\proj{\psi_{\vec y}^{\vec v}} \otimes \tau_\lambda)
    }\right]
\geq 1/\poly(\eps,1/q).
\end{align*}
\end{theorem}

\begin{figure}[!htb]
   \begin{center} 
   \begin{tabular}{|p{14cm}|}
    \hline 
\begin{center}
\underline{$\simult.\lwe.\dist^{\adversary,\algo D}\left( 1^{\secparam}\right)$}: 
\end{center}
\begin{enumerate}
   \item Sample $(\vec A, \vec T_{\vec A}) \leftarrow \GenTrap(1^n,1^m,q)$.
   \item Generate $(\ket{\psi_{\vec y}},\vec y) \leftarrow \mathsf{GenGauss}(\vec A,\sigma)$.

   \item Sample $\vec v \rand \bit^m$ and compute $\ket{\psi_{\vec y}^{\vec v}} \leftarrow \vec Z_q^{\lfloor \frac{q}{\nu} \rfloor \cdot \vec v} \ket{\psi_{\vec y}}$.
   \item Generate $\proj{\vec w} \otimes \rho_{\textsc{Aux}}  \leftarrow \algo A_{\lambda,\vec A,\vec y}(\proj{\psi_{\vec y}^{\vec v}} \otimes \tau_\lambda)$.
    \item Sample $\vec s \rand \Z_q^n, \vec e \sim D_{\Z^{m},\alpha q}$ and $e'\sim D_{\Z,\beta q}$.
   
   \item Run $\Revoke(\vec T_{\vec A},\vec w)$ on $\vec w$. If it outputs $\top$, continue. Otherwise, output $\invalid$.
    
    \item Run $b' \leftarrow \mathcal{D}(\vec A,\vec y,\vec s^\intercal \vec A+ \vec e^\intercal,\vec s^\intercal\vec y + e',\rho_{\reg{Aux}})$. Output $b'$.

\end{enumerate}
\ \\
\hline
\end{tabular}
    \caption{The distribution $\simult.\lwe.\dist^{\adversary,\algo D}\left( 1^{\secparam}\right)$.}
    \label{fig:lwe-dist-classical}
    \end{center}
\end{figure}

\begin{figure}[!htb]
   \begin{center} 
   \begin{tabular}{|p{14cm}|}
    \hline 
\begin{center}
\underline{$\simult.\unif.\dist^{\adversary,\algo D}\left( 1^{\secparam} \right)$}: 
\end{center}
\begin{enumerate}
   \item Sample $(\vec A, \vec T_{\vec A}) \leftarrow \GenTrap(1^n,1^m,q)$.
   \item Generate $(\ket{\psi_{\vec y}},\vec y) \leftarrow \mathsf{GenGauss}(\vec A,\sigma)$.

   \item Sample $\vec v \rand \bit^m$ and compute $\ket{\psi_{\vec y}^{\vec v}} \leftarrow \vec Z_q^{\lfloor \frac{q}{\nu} \rfloor \cdot \vec v} \ket{\psi_{\vec y}}$.
   \item Generate $\proj{\vec w} \otimes \rho_{\textsc{Aux}}  \leftarrow \algo A_{\lambda,\vec A,\vec y}(\proj{\psi_{\vec y}^{\vec v}} \otimes \tau_\lambda)$.
    \item Sample $\vec s \rand \Z_q^n, \vec e \sim D_{\Z^{m},\alpha q}$ and $e'\sim D_{\Z,\beta q}$.

\item Sample $\vec u \rand \Z_q^{m}$ and $r \rand \Z_q$.

   \item Run $\Revoke(\vec T_{\vec A},\vec w)$ on $\vec w$. If it outputs $\top$, continue. Otherwise, output $\invalid$.
    
    \item Run $b' \leftarrow \mathcal{D}(\vec A,\vec y,\vec u,r,\rho_{\reg{Aux}})$. Output $b'$.
\end{enumerate}
\ \\
\hline
\end{tabular}
    \caption{The distribution $\simult.\unif.\dist^{\adversary,\algo D}\left( 1^{\secparam} \right)$.}
    \label{fig:unif.dist-classical}
    \end{center}
\end{figure}

\begin{proof}
Let $\lambda \in \N$ be the security parameter and let $\algo A = \{(\algo A_{\lambda,\vec A,\vec y},\tau_\lambda)\}_{\vec A \in \Z_q^{n \times m}}$
be a non-uniform quantum algorithm. Suppose that $\mathcal{D}$ is a $\QPT$ distinguisher with advantage $\eps = 1/\poly(\lambda)$.

To prove the claim, we consider the following sequence of hybrid distributions.

\begin{description}
  
\item $\hybrid_0$: This is the distribution $\simult.\lwe.\dist^{\adversary,\algo D}\left( 1^{\secparam} \right)$ in \Cref{fig:lwe-dist-classical}.

\item $\hybrid_1$: This is the following distribution:
\begin{enumerate}
        \item Sample $(\vec A, \vec T_{\vec A}) \leftarrow \GenTrap(1^n,1^m,q)$.

        \item \rc{Sample a Gaussian vector $\vec x_0 \sim D_{\Z_q^m,\frac{\sigma}{\sqrt{2}}}$ and let $\vec y= \vec A \cdot \vec x_0 \Mod{q}$.}

        \item Run $\ket{\psi_{\vec y}} \leftarrow \mathsf{QSampGauss}(\vec A,\vec T_{\vec A},\vec y,\sigma)$.

        \item Sample $\vec v \rand \bit^m$ and compute $\ket{\psi_{\vec y}^{\vec v}} \leftarrow \vec Z_q^{\lfloor \frac{q}{\nu} \rfloor \cdot \vec v} \ket{\psi_{\vec y}}$.

        \item Run $\algo A_{\lambda,\vec A,\vec y}(\proj{\psi_{\vec y}^{\vec v}} \otimes \tau_\lambda)$ to generate a state $\proj{\vec w} \otimes \rho_{\vec{\textsc{aux}}}$.

        \item Run $\mathsf{Revoke}(\vec T_{\vec A},\vec v,\vec w)$. If it outputs $\top$, continue. Otherwise, output $\invalid$.

        \item Sample $\vec s \rand \Zq^n$, $\vec e \sim D_{\Z^m,\alpha q}$ and $e'\sim D_{\Z,\beta q}$. \rc{Let $\vec u = \vec A^\intercal \vec s+ \vec e$}.
        
       \item Run the distinguisher 
        \rc{$\mathcal{D}(\vec A,\vec y,\vec u,\vec u^\intercal\vec x_0 + e',\cdot)$} on system \textsc{Aux}. Output $b'$.
    \end{enumerate}

\item $\hybrid_2:$ This is the following distribution:
\begin{enumerate}
         \item Sample $(\vec A, \vec T_{\vec A}) \leftarrow \GenTrap(1^n,1^m,q)$.

        \item Sample a Gaussian vector $\vec x_0 \sim D_{\Z_q^m,\frac{\sigma}{\sqrt{2}}}$ and let $\vec y= \vec A \cdot \vec x_0 \Mod{q}$.

        \item Run $\ket{\psi_{\vec y}} \leftarrow \mathsf{QSampGauss}(\vec A,\vec T_{\vec A},\vec y,\sigma)$.

        \item Sample $\vec v \rand \bit^m$ and compute $\ket{\psi_{\vec y}^{\vec v}} \leftarrow \vec Z_q^{\lfloor \frac{q}{\nu} \rfloor \cdot \vec v} \ket{\psi_{\vec y}}$.

        \item Run $\algo A_{\lambda,\vec A,\vec y}(\proj{\psi_{\vec y}^{\vec v}} \otimes \tau_\lambda)$ to generate a state $\proj{\vec w} \otimes \rho_{\vec{\textsc{aux}}}$.

        \item Run $\mathsf{Revoke}(\vec T_{\vec A},\vec v,\vec w)$. If it outputs $\top$, continue. Otherwise, output $\invalid$.

        \item Sample  \rc{$\vec u \rand \Z_q^m$ and $e'\sim D_{\Z,\beta q}$}.
        
       \item Run the distinguisher 
    $\mathcal{D}(\vec A,\vec y,\vec u,\vec u^\intercal\vec x_0 + e',\cdot)$ on system \textsc{Aux}. Output $b'$.
    \end{enumerate}

\item $\hybrid_3$: This is the distribution $\simult.\unif.\dist^{\adversary,\algo D}\left( 1^{\secparam} \right)$ defined in \Cref{fig:unif.dist-classical}.

\end{description}
We now show the following:

\begin{claim} Hybrids $\hybrid_0$ and $\hybrid_1$ are statistically indistinguishable. In other words,
$$
\hybrid_0 \, \approx_s \,\hybrid_1.
$$
\end{claim}
\begin{proof} Here, we invoke the \emph{noise flooding} property in \Cref{lem:shifted-gaussian} to argue that $\vec e^\intercal \vec x_0 \ll e'$ holds with overwhelming probability for our choice of parameters.
Therefore, the distributions in $\hybrid_0$ and $\hybrid_1$ are computationally indistinguishable.
\end{proof}

\begin{claim} Assuming \Cref{conjecture-II} holds for our choice of parameters, the hybrids $\hybrid_1$ and $\hybrid_2$ are computationally indistinguishable,
$$
\hybrid_1 \, \approx_c \,\hybrid_2.
$$
\end{claim}
\begin{proof}
This follows from directly from \Cref{conjecture-II}, which essentially allows us to invoke a variant of the $\LWE$ assumption, even when the procedure $\mathsf{Revoke}(\vec T_{\vec A},\vec v,\vec w)$ is simultaneously taken into account. Here, we rely on the fact that during the reduction, we can simply sample $e'\sim D_{\Z,\beta q}$ to produce an identically distributed challenge distribution.
\end{proof}

\par Recall that $\hybrid_0$ and $\hybrid_3$ can be distinguished with probability $\eps=1/\poly(\secparam)$. We proved that the hybrids $\hybrid_0$ and $\hybrid_2$ are computationally indistinguishable. As a consequence, it holds that hybrids $\hybrid_2$ and $\hybrid_3$ can be distinguished with probability at least $\eps - \negl(\secparam)$.
\par We leverage this to obtain a Goldreich-Levin reduction. Consider the following distinguisher.

\begin{figure}[!htb]
   \begin{center} 
   \begin{tabular}{|p{14cm}|}
    \hline 
\begin{center}
\underline{$\tilde{\algo D}\big(\vec A,\vec y,\vec u,v,\rho\big)$}: 
\end{center}
Input: $\vec A \in \Zq^{n \times m}$, $\vec y \in \Z_q^n$, $\vec u \in \Z_q^n$, $v \in \Z_q$ and $\rho \in L(\algo H_{\vec{\textsc{Aux}}})$.\\
Output: A bit $b' \in \bit$.
\vspace{3mm}\\
\textbf{Procedure:}
\begin{enumerate}
   \item Sample $e' \sim D_{\Z,\beta q}$.
   
   \item Output $b' \leftarrow \algo D\big(\vec A,\vec y,\vec u, v + e',\rho\big)$.
\end{enumerate}
\ \\
\hline
\end{tabular}
    \caption{The distinguisher $\tilde{\algo D}\big(\vec A,\vec y,\vec u,v,\rho\big)$.}
    \label{fig:tilde-dist-classical}
    \end{center}
\end{figure}
Note that $r + e' \Mod{q}$ is uniform whenever $r \rand \Z_q$ and $e' \sim D_{\Z,\beta q}$. Therefore, our previous argument shows that there exists a negligible function $\eta$ such that:
\begin{align*}
\,&\vline\Pr \left[
\substack{
\mathsf{Revoke}(\vec T_{\vec A},\vec v,\vec w) = \top\\
\ \\
\bigwedge\\
\ \\
\tilde{\mathcal{D}}(\vec A,\vec y,\vec u, \vec u^\intercal \vec x_0, \rho_{\textsc{Aux}}) =1
}
\, : \, \substack{
(\vec A, \vec T_{\vec A}) \leftarrow \GenTrap(1^n,1^m,q)\\ \vec u \rand \Z_q^m\\
\vec x_0 \sim D_{\Z_q^m,\frac{\sigma}{\sqrt{2}}}, \, \vec y = \vec A \vec x_0 \text{ (mod $q$)}\\
\ket{\psi_{\vec y}} \leftarrow \mathsf{QSampGauss}(\vec A,\vec y,\vec T_{\vec A},\sigma)\\
\proj{\vec w}\otimes \rho_{\vec{\textsc{Aux}}} \leftarrow \algo A_{\lambda,\vec A,\vec y}(\proj{\psi_{\vec y}} \otimes \tau_\lambda)
}
\right] -\\
&\Pr \left[
\substack{
\mathsf{Revoke}(\vec T_{\vec A},\vec v,\vec w) = \top\\
\ \\
\bigwedge\\
\ \\
\tilde{\mathcal{D}}(\vec A,\vec y,\vec u,r, \rho_{\textsc{Aux}}) =1
}
\, : \, \substack{
(\vec A, \vec T_{\vec A}) \leftarrow \GenTrap(1^n,1^m,q)\\
\vec u \rand \Z_q^m, r \rand \Z_q\\
\vec x_0 \sim D_{\Z_q^m,\frac{\sigma}{\sqrt{2}}}, \, \vec y = \vec A \vec x_0 \text{ (mod $q$)}\\
\ket{\psi_{\vec y}} \leftarrow \mathsf{QSampGauss}(\vec A,\vec y,\vec T_{\vec A},\sigma)\\
\proj{\vec w}\otimes \rho_{\vec{\textsc{Aux}}} \leftarrow \algo A_{\lambda,\vec A,\vec y}(\proj{\psi_{\vec y}} \otimes \tau_\lambda)
}
\right]  \vline \geq \eps - \eta(\lambda).
\end{align*}

From \Cref{thm:QGL}, it follows that there exists a Goldreich-Levin extractor $\algo E$ running in time $T(\algo E) =\poly(\lambda,n,m,\sigma,q,1/\eps)$ that outputs a short vector in $\Lambda_q^{\vec y}(\vec A)$ with probability at least
\begin{align*}
&\Pr \left[\substack{
\mathsf{Revoke}(\vec T_{\vec A},\vec v,\vec w) = \top \\ 
\ \\
\bigwedge\\
\ \\
\mathcal{E}(\vec A,\vec y,\rho_{\vec{\textsc{Aux}}}) \, \in \, \Lambda_q^{\vec y}(\vec A) \,\cap\, \algo B^m(\vec 0,\sigma \sqrt{\frac{m}{2}})}
 \, : \, \substack{
 (\vec A, \vec T_{\vec A}) \leftarrow \GenTrap(1^n,1^m,q)\\
\vec x_0 \sim D_{\Z_q^m,\frac{\sigma}{\sqrt{2}}}, \, \vec y = \vec A \vec x_0 \text{ (mod $q$)}\\
\ket{\psi_{\vec y}} \leftarrow \mathsf{QSampGauss}(\vec A,\vec y,\vec T_{\vec A},\sigma)\\
\proj{\vec w}\otimes \rho_{\vec{\textsc{Aux}}} \leftarrow \algo A_{\lambda,\vec A,\vec y}(\proj{\psi_{\vec y}} \otimes \tau_\lambda)
    }\right]
\geq 1/\poly(\eps,1/q).
\end{align*}
This proves the claim.
\end{proof}

\paragraph{Towards a proof of the conjecture.}

We now give a proof of the \emph{simultaneous} search-to-decision reduction (\Cref{thm:search-to-decision-classical}) from standard assumptions in the special case when revocation succeeds with overwhelming probability.

\begin{theorem}\label{thm:search-to-decision-revoke-1-classical}
Let $n\in \N$. Let $q$ be a prime with $q=2^{o(n)}$, $m \geq 6n \log q$, $\sqrt{8m}< \sigma <q/\sqrt{8m}$, and let $\alpha,\beta \in (0,1)$ with $\beta/\alpha = 2^{o(n)}$ with $1/\alpha= 2^{o(n)} \cdot \sigma$. Let $\nu = 64m^2$. Let
$\algo A = \{(\algo A_{\lambda,\vec A,\vec y},\tau_\lambda)\}_{\lambda \in \N}$
be any non-uniform quantum algorithm consisting of a family of polynomial-sized quantum circuits
$$\Bigg\{\algo A_{\lambda,\vec A,\vec y}: \algo L(\algo H_q^m \otimes \algo H_{B_\lambda}) \allowbreak\rightarrow \algo L(\algo H_{R_\lambda} \otimes \algo H_{\textsc{aux}_\lambda})\Bigg\}_{\vec A \in \Zq^{n \times m}, \,\vec y \in \Z_q^n}$$
and polynomial-sized advice states $\tau_\lambda \in \algo D(\algo H_{B_\lambda})$ which are independent of $\vec A$ such that
$$
\Pr\left[
\mathsf{Revoke}(\vec T_{\vec A},\vec v,\vec w) = \top
 \, : \, \substack{
(\vec A, \vec T_{\vec A}) \leftarrow \GenTrap(1^n,1^m,q)\\
(\ket{\psi_{\vec y}}, \vec y) \leftarrow \mathsf{GenGauss}(\vec A,\sigma)\\
\vec v \rand \bit^m, \, \, \ket{\psi_{\vec y}^{\vec v}} \leftarrow \vec Z_q^{\lfloor \frac{q}{\nu} \rfloor \cdot \vec v} \ket{\psi_{\vec y}}\\
\proj{\vec w} \otimes \rho_{\textsc{Aux}} \leftarrow \algo A_{\lambda,\vec A,\vec y}(\proj{\psi_{\vec y}^{\vec v}} \otimes \tau_\lambda)
    }\right] = 1-\mu(\lambda),
$$
for some negligle function $\mu(\lambda)$.
Then,
assuming the quantum hardness of the $\LWE_{n,q,\alpha q}^m$ assumption, the following holds for every $\QPT$ distinguisher $\mathcal{D}$. Suppose that there exists a function $\eps(\lambda) = 1/\poly(\lambda)$ such that
\begin{align*}
&\vline \, \prob\left[ 1 \leftarrow \mathsf{SimultSearchToDecisionExpt}^{\adversary,\algo D}(1^{\secparam},0)\right]-\\
&\,\prob\left[ 1 \leftarrow \mathsf{SimultSearchToDecisionExpt}^{\adversary,\algo D}(1^{\secparam},1)\right]  \,\vline = \eps(\lambda).
\end{align*}

\noindent Then, there exists a quantum extractor $\algo E$ that takes as input $\vec A$, $\vec y$ and system $\vec{\textsc{Aux}}$ of the state $\rho_{\reg R,\vec{\textsc{Aux}}}$ and outputs a short vector in the coset $\Lambda_q^{\vec y}(\vec A)$ in time $\poly(\lambda,m,\sigma,q,1/\eps)$ such that
\begin{align*}
&\Pr \left[\substack{
\mathsf{Revoke}(\vec T_{\vec A},\vec v,\vec w)\vspace{0.2mm}\\ \bigwedge \vspace{0.2mm}\\ \mathcal{E}(\vec A,\vec y,\reg{Aux}) \,\,\in\,\, \Lambda_q^{\vec y}(\vec A) \,\cap\, \algo B^m(\vec 0,\sigma \sqrt{\frac{m}{2}})}
 \, : \, \substack{
 (\vec A, \vec T_{\vec A}) \leftarrow \GenTrap(1^n,1^m,q)\\
(\ket{\psi_{\vec y}}, \vec y) \leftarrow \mathsf{GenGauss}(\vec A,\sigma)\\
\vec v \rand \bit^m, \, \, \ket{\psi_{\vec y}^{\vec v}} \leftarrow \vec Z_q^{\lfloor \frac{q}{\nu} \rfloor \cdot \vec v} \ket{\psi_{\vec y}}\\
\proj{\vec w} \otimes \rho_{\textsc{Aux}} \leftarrow \algo A_{\lambda,\vec A,\vec y}(\proj{\psi_{\vec y}^{\vec v}} \otimes \tau_\lambda)
    }\right] \, \geq\, \poly(\eps,1/q).
\end{align*}
\end{theorem}

\begin{proof}
Let $\lambda \in \N$ and suppose that there exists an adversary consisting of a pair of $\QPT$ algorithms $(\algo A,\algo D)$ and a function $\eps(\lambda) = 1/\poly(\lambda)$ such that
\begin{align*}
&\vline\Pr \left[
\substack{
\mathsf{Revoke}(\vec T_{\vec A},\vec v,\vec w) = \top\\
\ \\
\bigwedge\\
\ \\
\mathcal{D}(\vec A,\vec y,\vec s^\intercal \vec A + \vec e^\intercal, \vec s^\intercal \vec y + e', \rho_{\vec{\textsc{Aux}}}) =1
}
\, : \, \substack{
(\vec A, \vec T_{\vec A}) \leftarrow \GenTrap(1^n,1^m,q)\\
\vec v \rand \bit^m, \,\vec s \rand \Z_q^n\\
\vec e \sim D_{\Z^{m},\,\alpha q}, \, e' \sim D_{\Z,\,\beta q}\\
(\ket{\psi_{\vec y}}, \vec y) \leftarrow \mathsf{GenGauss}(\vec A,\sigma)\\
\ket{\psi_{\vec y}^{\vec v}} \leftarrow \vec Z_q^{\lfloor \frac{q}{\nu} \rfloor \cdot \vec v} \ket{\psi_{\vec y}}\\
\proj{\vec w} \otimes \rho_{\textsc{Aux}} \leftarrow \algo A_{\lambda,\vec A,\vec y}(\proj{\psi_{\vec y}^{\vec v}} \otimes \tau_\lambda)
}
\right] -\\
&\Pr \left[
\substack{
\mathsf{Revoke}(\vec T_{\vec A},\vec v,\vec w) = \top\\
\ \\
\bigwedge\\
\ \\
\mathcal{D}(\vec A,\vec y,\vec u,r, \rho_{\vec{\textsc{Aux}}}) =1
}
\, : \, \substack{
(\vec A, \vec T_{\vec A}) \leftarrow \GenTrap(1^n,1^m,q)\\
\vec v \rand \bit^m, \vec u \rand \Z_q^m,\,
r \rand \Z_q\\
(\ket{\psi_{\vec y}}, \vec y) \leftarrow \mathsf{GenGauss}(\vec A,\sigma)\\
\ket{\psi_{\vec y}^{\vec v}} \leftarrow \vec Z_q^{\lfloor \frac{q}{\nu} \rfloor \cdot \vec v} \ket{\psi_{\vec y}}\\
\proj{\vec w} \otimes \rho_{\textsc{Aux}} \leftarrow\algo A_{\lambda,\vec A,\vec y}(\proj{\psi_{\vec y}^{\vec v}} \otimes \tau_\lambda)
}
\right] \vline = \eps(\lambda).
\end{align*}  
    From \Cref{lem:remove-revoke}, it follows that there exists a $\QPT$ distinguisher $\tilde{\algo D}$ that succeeds on the reduced state on \textsc{Aux} alone (where we drop the condition that $\mathsf{Revoke}(\vec T_{\vec A},\vec v,\vec w)$ outputs $\top$). In other words, there exists $\bar{\epsilon} = 1/\poly(\lambda)$ such that 
\begin{align*}
&\vline\Pr \left[
\tilde{\mathcal{D}}(\vec A,\vec y,\vec s^\intercal \vec A + \vec e^\intercal, \vec s^\intercal \vec y + e', \rho_{\vec{\textsc{Aux}}}) =1
\, : \, \substack{
(\vec A, \vec T_{\vec A}) \leftarrow \GenTrap(1^n,1^m,q)\\
\vec v \rand \bit^m, \,\vec s \rand \Z_q^n\\
\vec e \sim D_{\Z^{m},\,\alpha q}, \, e' \sim D_{\Z,\,\beta q}\\
(\ket{\psi_{\vec y}}, \vec y) \leftarrow \mathsf{GenGauss}(\vec A,\sigma)\\
\ket{\psi_{\vec y}^{\vec v}} \leftarrow \vec Z_q^{\lfloor \frac{q}{\nu} \rfloor \cdot \vec v} \ket{\psi_{\vec y}}\\
\proj{\vec w} \otimes \rho_{\textsc{Aux}} \leftarrow \algo A_{\lambda,\vec A,\vec y}(\proj{\psi_{\vec y}^{\vec v}} \otimes \tau_\lambda)
}
\right] -\\
&\Pr \left[
\tilde{\mathcal{D}}(\vec A,\vec y,\vec u,r, \rho_{\vec{\textsc{Aux}}}) =1
\, : \, \substack{
(\vec A, \vec T_{\vec A}) \leftarrow \GenTrap(1^n,1^m,q)\\
\vec v \rand \bit^m, \vec u \rand \Z_q^m,\,
r \rand \Z_q\\
(\ket{\psi_{\vec y}}, \vec y) \leftarrow \mathsf{GenGauss}(\vec A,\sigma)\\
\ket{\psi_{\vec y}^{\vec v}} \leftarrow \vec Z_q^{\lfloor \frac{q}{\nu} \rfloor \cdot \vec v} \ket{\psi_{\vec y}}\\
\proj{\vec w} \otimes \rho_{\textsc{Aux}} \leftarrow \algo A_{\lambda,\vec A,\vec y}(\proj{\psi_{\vec y}^{\vec v}} \otimes \tau_\lambda)
}
\right] \vline = \bar{\eps}(\lambda).
\end{align*}  
Because $\tilde{\algo D}$ succeeds with distinguishing advantage $\bar{\eps} = 1/\poly(\lambda)$, we can use \Cref{thm:search-to-decision-without-revoke} to argue that 
there exists a quantum extractor $\algo E$ that takes as input a set of good inputs $(\vec A,\vec y,\rho_{\reg{Aux}})$, and outputs a short vector in the coset $\Lambda_q^{\vec y}(\vec A)$ in time $\poly(\lambda,m,\sigma,q,1/\bar{\eps})$ such that
\begin{align*}
\Pr \left[\substack{
\mathcal{E}(\vec A,\vec y,\rho_{\vec{\textsc{Aux}}}) = \vec x_0\vspace{1mm}\\ \bigwedge \vspace{0.1mm}\\ {\vec x_0} \,\,\in\,\, \Lambda_q^{\vec y}(\vec A) \,\cap\, \algo B^m(\vec 0,\sigma \sqrt{\frac{m}{2}})} \, : \, \substack{
(\vec A, \vec T_{\vec A}) \leftarrow \GenTrap(1^n,1^m,q)\\
(\ket{\psi_{\vec y}}, \vec y) \leftarrow \mathsf{GenGauss}(\vec A,\sigma)\\
\vec v \rand \bit^m, \, \, \ket{\psi_{\vec y}^{\vec v}} \leftarrow \vec Z_q^{\lfloor \frac{q}{\nu} \rfloor \cdot \vec v} \ket{\psi_{\vec y}}\\
\proj{\vec w} \otimes \rho_{\textsc{Aux}} \leftarrow \algo A_{\lambda,\vec A,\vec y}(\proj{\psi_{\vec y}^{\vec v}} \otimes \tau_\lambda)
}\right] \geq \poly(\bar{\eps},1/q).
\end{align*}
Recall that revocation succeeds with overwhelming probability, i.e.,
$$
\Pr\left[
\mathsf{Revoke}(\vec T_{\vec A},\vec v,\vec w) = \top
 \, : \, \substack{
(\vec A, \vec T_{\vec A}) \leftarrow \GenTrap(1^n,1^m,q)\\
(\ket{\psi_{\vec y}}, \vec y) \leftarrow \mathsf{GenGauss}(\vec A,\sigma)\\
\vec v \rand \bit^m, \, \, \ket{\psi_{\vec y}^{\vec v}} \leftarrow \vec Z_q^{\lfloor \frac{q}{\nu} \rfloor \cdot \vec v} \ket{\psi_{\vec y}}\\
\proj{\vec w} \otimes \rho_{\textsc{Aux}} \leftarrow \algo A_{\lambda,\vec A,\vec y}(\proj{\psi_{\vec y}^{\vec v}} \otimes \tau_\lambda)
    }\right] = 1-\negl(\lambda).
$$
Therefore, we can use Bonferroni's inequality to argue that
\begin{align*}
&\Pr \left[\substack{
\mathsf{Revoke}(\vec T_{\vec A},\vec v,\vec w)  = \top\vspace{0.2mm}\\ \bigwedge \vspace{0.2mm}\\ \mathcal{E}(\vec A,\vec y,\reg{Aux}) \,\,\in\,\, \Lambda_q^{\vec y}(\vec A) \,\cap\, \algo B^m(\vec 0,\sigma \sqrt{\frac{m}{2}})}
 \, : \, \substack{
 (\vec A, \vec T_{\vec A}) \leftarrow \GenTrap(1^n,1^m,q)\\
(\ket{\psi_{\vec y}}, \vec y) \leftarrow \mathsf{GenGauss}(\vec A,\sigma)\\
\vec v \rand \bit^m, \, \, \ket{\psi_{\vec y}^{\vec v}} \leftarrow \vec Z_q^{\lfloor \frac{q}{\nu} \rfloor \cdot \vec v} \ket{\psi_{\vec y}}\\
\proj{\vec w} \otimes \rho_{\textsc{Aux}} \leftarrow \algo A_{\lambda,\vec A,\vec y}(\proj{\vec A} \otimes \proj{\vec y} \otimes \proj{\psi_{\vec y}^{\vec v}} \otimes \tau_\lambda )
    }\right]\\
&\geq \Pr \left[
\mathsf{Revoke}(\vec T_{\vec A},\vec v,\vec w)= \top  \, : \, \substack{
 (\vec A, \vec T_{\vec A}) \leftarrow \GenTrap(1^n,1^m,q)\\
(\ket{\psi_{\vec y}}, \vec y) \leftarrow \mathsf{GenGauss}(\vec A,\sigma)\\
\vec v \rand \bit^m, \, \, \ket{\psi_{\vec y}^{\vec v}} \leftarrow \vec Z_q^{\lfloor \frac{q}{\nu} \rfloor \cdot \vec v} \ket{\psi_{\vec y}}\\
\proj{\vec w} \otimes \rho_{\textsc{Aux}} \leftarrow \algo A_{\lambda,\vec A,\vec y}(\proj{\psi_{\vec y}^{\vec v}} \otimes \tau_\lambda )
    } \right]\\
    & \quad + \Pr \left[
\mathcal{E}(\vec A,\vec y,\rho_{\reg{Aux}}) \,\,\in\,\, \Lambda_q^{\vec y}(\vec A) \,\cap\, \algo B^m(\vec 0,\sigma \sqrt{m/2})  \, : \, \substack{
  (\vec A, \vec T_{\vec A}) \leftarrow \GenTrap(1^n,1^m,q)\\
(\ket{\psi_{\vec y}}, \vec y) \leftarrow \mathsf{GenGauss}(\vec A,\sigma)\\
\vec v \rand \bit^m, \, \, \ket{\psi_{\vec y}^{\vec v}} \leftarrow \vec Z_q^{\lfloor \frac{q}{\nu} \rfloor \cdot \vec v} \ket{\psi_{\vec y}}\\
\proj{\vec w} \otimes \rho_{\textsc{Aux}} \leftarrow \algo A_{\lambda,\vec A,\vec y}(\proj{\psi_{\vec y}^{\vec v}} \otimes \tau_\lambda )
    }\right] -  1\\
&\geq \Pr \left[
\mathcal{E}(\vec A,\vec y,\rho_{\reg{Aux}}) \,\,\in\,\, \Lambda_q^{\vec y}(\vec A) \,\cap\, \algo B^m(\vec 0,\sigma \sqrt{m/2})  \, : \, \substack{
  (\vec A, \vec T_{\vec A}) \leftarrow \GenTrap(1^n,1^m,q)\\
(\ket{\psi_{\vec y}}, \vec y) \leftarrow \mathsf{GenGauss}(\vec A,\sigma)\\
\vec v \rand \bit^m, \, \, \ket{\psi_{\vec y}^{\vec v}} \leftarrow \vec Z_q^{\lfloor \frac{q}{\nu} \rfloor \cdot \vec v} \ket{\psi_{\vec y}}\\
\proj{\vec w} \otimes \rho_{\textsc{Aux}} \leftarrow \algo A_{\lambda,\vec A,\vec y}( \proj{\psi_{\vec y}^{\vec v}} \otimes \tau_\lambda )
    }\right] -  \negl(\lambda)\\
&\geq \poly(\eps,1/q).
\end{align*}
This proves the claim.
\end{proof}

\paragraph{Proof of security.}

We can show that the Dual-Regev scheme with classical revocation achieves our notion key-revocable security from \Cref{def:krpke:security-classical}, assuming either \Cref{thm:search-to-decision-classical} is true, or using \Cref{thm:search-to-decision-revoke-1-classical} based on $\LWE$/$\SIS$
in case revocation succeeds with overwhelming probability.

To carry out the proof, we show that a successful key-revocation adversary allows us to break the uncertainty principle for Gaussian superpositions in \Cref{cor:uncertainty}. The idea is the following: Suppose there exists an algorithm $\algo D$ that can distinguish Dual-Regev from uniform given $\rho_{\textsc{Aux}}$ in the key-revocation experiment (conditioned on $\mathsf{Revoke}$ succeeding on certificate $\vec w$). Then we can use it to extract (using Goldreich-Levin) a short pre-image $\vec x_0$ from the internal state $\rho_{\textsc{Aux}}$ which we can send to the challenger as part of the experiment in \Cref{cor:uncertainty}. Moreover, because the revocation certificate $\vec w$ passes verification\footnote{Here, one can think of it as essentially being an LWE encryption $\vec w^\intercal = \hat{\vec s}^\intercal \vec A + \lfloor \frac{q}{\nu} \rfloor \cdot \vec v^\intercal  + \hat{\vec e}^\intercal$ of the string $\vec v$.}, it must yield $\vec v$ when $\vec w$ is decrypted using a short trapdoor basis (i.e., by running $\mathsf{Revoke}$).
Fortunately, once we present a pre-image $\vec x_0$ in the second half of the experiment, we are allowed to be computationally unbounded; meaning, we can simply decode $\vec w$ in exponential time (i.e. by breaking $\LWE$, or by finding a trapdoor that can be used to decode $\vec w$ using $\mathsf{Revoke}$). Therefore, we obtain an adversary that ends up finding both a pre-image $\vec x_0$ and $\vec v$ at the same time, which violates \Cref{cor:uncertainty}.

Our first result concerns $(\negl(\lambda),\negl(\lambda))$-security, i.e., where we assume that revocation succeeds with overwhelming probability.

\begin{theorem}\label{thm:dual-regev-classical}
Let $n\in \N$ and $q$ be a prime with $q=2^{o(n)}$ and $m = \lceil 6 n \log q\rceil$, each parameterized by $\lambda \in \N$. Let $q/\sqrt{8m} > \sigma > \sqrt{8m}$ and let $\alpha,\beta \in (0,1)$ be noise ratios chosen such that $\beta/\alpha =2^{o(n)}$ and $1/\alpha= 2^{o(n)} \cdot \sigma$. Let $\nu = 64m^2$.
Then, assuming the subexponential hardness of the $\mathsf{LWE}_{n,q,\alpha q}^m$
and $\mathsf{SIS}_{n,q,\sigma\sqrt{2m}}^m$ problems, the scheme
$\mathsf{CRevDual} = (\keygen,\enc,\dec,\mathsf{Delete},\mathsf{Revoke})$ in \Cref{cons:dual-regev-classical-revoc}
is a $(\negl(\lambda),\negl(\lambda))$-secure key-revocable public-key encryption scheme with classical revocation.
\end{theorem}

\begin{proof}
Let $\lambda \in \N$ and suppose that there exists an adversary consisting of a pair of $\QPT$ algorithms $(\algo A,\algo D)$ and a function $\eps(\lambda) = 1/\poly(\lambda)$ such that
\begin{align*}
&\vline\Pr \left[
\substack{
\mathsf{Revoke}(\vec T_{\vec A},\vec v,\vec w) = \top\\
\ \\
\bigwedge\\
\ \\
\mathcal{D}(\vec A,\vec y,\vec s^\intercal \vec A + \vec e^\intercal, \vec s^\intercal \vec y + e', \rho_{\vec{\textsc{Aux}}}) =1
}
\, : \, \substack{
(\vec A, \vec T_{\vec A}) \leftarrow \GenTrap(1^n,1^m,q)\\
\vec v \rand \bit^m, \,\vec s \rand \Z_q^n\\
\vec e \sim D_{\Z^{m},\,\alpha q}, \, e' \sim D_{\Z,\,\beta q}\\
(\ket{\psi_{\vec y}}, \vec y) \leftarrow \mathsf{GenGauss}(\vec A,\sigma)\\
\ket{\psi_{\vec y}^{\vec v}} \leftarrow \vec Z_q^{\lfloor \frac{q}{\nu} \rfloor \cdot \vec v} \ket{\psi_{\vec y}}\\
\proj{\vec w} \otimes \rho_{\textsc{Aux}} \leftarrow \algo A_{\lambda,\vec A,\vec y}(\proj{\psi_{\vec y}^{\vec v}} \otimes \tau_\lambda)
}
\right] -\\
&\Pr \left[
\substack{
\mathsf{Revoke}(\vec T_{\vec A},\vec v,\vec w) = \top\\
\ \\
\bigwedge\\
\ \\
\mathcal{D}(\vec A,\vec y,\vec u,r, \rho_{\vec{\textsc{Aux}}}) =1
}
\, : \, \substack{
(\vec A, \vec T_{\vec A}) \leftarrow \GenTrap(1^n,1^m,q)\\
\vec v \rand \bit^m, \vec u \rand \Z_q^m,\,
r \rand \Z_q\\
(\ket{\psi_{\vec y}}, \vec y) \leftarrow \mathsf{GenGauss}(\vec A,\sigma)\\
\ket{\psi_{\vec y}^{\vec v}} \leftarrow \vec Z_q^{\lfloor \frac{q}{\nu} \rfloor \cdot \vec v} \ket{\psi_{\vec y}}\\
\proj{\vec w} \otimes \rho_{\textsc{Aux}} \leftarrow \algo A_{\lambda,\vec A,\vec y}(\proj{\psi_{\vec y}^{\vec v}} \otimes \tau_\lambda)
}
\right] \vline = \eps(\lambda).
\end{align*}  
We now show how to use $(\algo A,\algo D)$ to violate the uncertainty principle for Gaussian states from \Cref{cor:uncertainty}.
Consider the adversary $(\algo B_0,\algo{B}_1)$ consisting of the following procedures:
\begin{itemize}
\item $\algo B_0(\vec A,\vec y,\ket{\psi_{\vec y}^{\vec v}})$: run the algorithm $\algo A$ to generate
$$
 \proj{\vec w} \otimes \rho_{\textsc{Aux}} \leftarrow \algo A_{\lambda,\vec A,\vec y}(\proj{\psi_{\vec y}^{\vec v}} \otimes \tau_\lambda).
$$
Then, run the Goldreich-Levin extractor $\algo E(\vec A,\vec y,\rho_{\vec{\textsc{aux}}})$ from \Cref{thm:search-to-decision-revoke-1-classical} (instantiated using $\algo A,\algo D$) and let $\vec x_0$ be the outcome. Send $\vec x_0$ to the challenger and initialize ${\algo B_1}$ with $(\vec A,\vec w)$.

\item $\algo B_1(\vec A,\vec w)$: compute (in exponential time) a $(20 n \log q)$-good trapdoor basis $\hat{\vec T}_{\vec A} \in \Z^{m \times m}$ for the matrix $\vec A \in \Z_q^{n \times m}$ as in \Cref{def:beta-good-lattice-td} with the properties that
$$
\vec A \cdot \hat{\vec T}_{\vec A} = \vec 0 \Mod{q} \quad\text{ and } \quad \|\hat{\vec T}_{\vec A}\| \leq 20 n \log q.
$$ 
Then, use $\hat{\vec T}_{\vec A}$ to do the following:
\begin{enumerate}
\item Compute $\vec c^\intercal = \vec w^\intercal \cdot \hat{\vec T}_{\vec A} \Mod{q}$.
\item Compute $\vec d^\intercal = \vec c^\intercal \cdot \hat{\vec T}_{\vec A}^{-1}$, where $\hat{\vec T}_{\vec A}^{-1}$ is the inverse matrix of $\hat{\vec T}_{\vec A}$ over $\mathbb{R}$.

    \item For $i \in [m]$, let $v'_i = 0$, if $d_i \in [-q/\sigma,q/\sigma]$, and let $v'_i = 1$, otherwise.

    \item Output $\vec v' = (v'_1,\dots,v'_m)$.
\end{enumerate}
\end{itemize}
Because $\algo D$ succeeds with distinguishing advantage $\eps = 1/\poly(\lambda)$, it follows from \Cref{thm:search-to-decision-revoke-1-classical} the extractor outputs a short vector in the coset $\Lambda_q^{\vec y}(\vec A)$ in time $\poly(\lambda,m,\sigma,q,1/\eps)$ such that
\begin{align*}
&\Pr \left[\substack{
\mathsf{Revoke}(\vec T_{\vec A},\vec v,\vec w)\vspace{0.2mm}\\ \bigwedge \vspace{0.2mm}\\ \mathcal{E}(\vec A,\vec y,\reg{Aux}) \,\,\in\,\, \Lambda_q^{\vec y}(\vec A) \,\cap\, \algo B^m(\vec 0,\sigma \sqrt{\frac{m}{2}})}
 \, : \, \substack{
 (\vec A, \vec T_{\vec A}) \leftarrow \GenTrap(1^n,1^m,q)\\
(\ket{\psi_{\vec y}}, \vec y) \leftarrow \mathsf{GenGauss}(\vec A,\sigma)\\
\vec v \rand \bit^m, \, \, \ket{\psi_{\vec y}^{\vec v}} \leftarrow \vec Z_q^{\lfloor \frac{q}{\nu} \rfloor \cdot \vec v} \ket{\psi_{\vec y}}\\
\proj{\vec w} \otimes \rho_{\textsc{Aux}} \leftarrow \algo A_{\lambda,\vec A,\vec y}(\proj{\psi_{\vec y}^{\vec v}} \otimes \tau_\lambda)
    }\right] \, \geq\, \poly(\eps,1/q).
\end{align*}
Finally, because $\hat{\vec T}_{\vec A} \in \Z^{m \times m}$ is a $(20 n \log q)$-good trapdoor basis for $\vec A \in \Z_q^{n \times m}$, this implies that the probability that $\algo B_1(\vec A,\vec w)$ correctly recovers $\vec v$ is also at least $1/\poly(\lambda)$.
Therefore,
$(\algo B_0,\algo{B}_1)$ has non-negligible probability of finding both the correct certificate $\vec v' = \vec v$ as well as a short pre-image $\vec x_0$ such that $\vec A \cdot \vec x_0 = \vec y \Mod{q}$ and $\| \vec x_0 \| \leq \sigma \sqrt{m/2}$. This proves the claim.
\end{proof}

Finally, we obtain the following stronger notion of $(\negl(\lambda),1-1/\poly(\lambda))$-security, assuming \Cref{thm:search-to-decision-classical} is true.

\begin{theorem}\label{thm:dual-regev-classical-conj}
Let $n\in \N$ and $q$ be a prime with $q=2^{o(n)}$ and $m = \lceil 6 n \log q\rceil$, each parameterized by $\lambda \in \N$. Let $q/\sqrt{8m} > \sigma > \sqrt{8m}$ and let $\alpha,\beta \in (0,1)$ be noise ratios chosen such that $\beta/\alpha =2^{o(n)}$ and $1/\alpha= 2^{o(n)} \cdot \sigma$. Let $\nu = 64m^2$.
Then, assuming \Cref{thm:search-to-decision-classical}, the scheme
$\mathsf{CRevDual} = (\keygen,\enc,\dec,\mathsf{Delete},\mathsf{Revoke})$ in \Cref{cons:dual-regev-classical-revoc}
is a $(\negl(\lambda),1-1/\poly(\lambda))$-secure key-revocable public-key encryption scheme with classical revocation.
\end{theorem}

\begin{proof}
The proof is analogous to \Cref{thm:dual-regev-classical}, except that we invoke \Cref{thm:search-to-decision-classical} instead of \Cref{thm:search-to-decision-revoke-1-classical}.
\end{proof}
\section{Key-Revocable Fully Homomorphic Encryption}\label{sec:dual-GSW}

In this section, we describe our key-revocable (leveled) fully homomorphic encryption scheme from $\LWE$ which is based on the so-called $\mathsf{DualGSW}$ scheme used by Mahadev~\cite{mahadev2018classical} which itself is a variant of the homomorphic encryption scheme by Gentry, Sahai and Waters~\cite{GSW2013}.

Let $\lambda \in \N$ be the security parameter. Suppose we would like to evaluate $L$-depth circuits consisting of $\mathsf{NAND}$ gates. We choose $n(\lambda,L) \gg L$ and a prime $q=2^{o(n)}$. Then, for integer parameters $m \geq 2 n \log q$ and $N = (m+1) \cdot \lceil \log q \rceil$, we let $\vec I$ be the $(m+1) \times (m+1)$ identity matrix and let
$\vec G = [\vec I \, \| \, 2 \vec I \, \| \, \dots \, \| \, 2^{\lceil \log q \rceil -1} \vec I] \in \Z_q^{(m+1) \times N}$ denote the so-called \emph{gadget matrix} which converts a binary representation of a vector back to its original vector representation over the field $\Z_q$. 
Note that the associated (non-linear) inverse operation $\vec G^{-1}$ converts vectors in $\Z_q^{m+1}$ to their binary representation in $\bit^N$. In other words, we have that $\vec G \circ \vec G^{-1}$ acts as the identity operator.

\subsection{Construction}\label{sec:DualGSW}

We now construct a key-revocable Dual-Regev (leveled) fully homomorphic encryption scheme based on our key-revocable Dual-Regev public-key scheme.
\begin{remark}
While the construction in this section can be readily adapted to feature \emph{classical revocation} via \Cref{cons:dual-regev-classical-revoc}, we choose to focus on \emph{quantum revocation} for simplicity by making use of our Dual-Regev scheme from \Cref{cons:dual-regev}.
\end{remark}

We are now ready to state our construction.

\begin{construction}[Key-Revocable \textsf{DualGSW} encryption]\label{cons:DualGSW}
Let $\lambda \in \N$ be the security parameter.
The scheme $\mathsf{RevDualGSW} = (\KeyGen,\Enc,\Dec,\Eval,\Revoke)$ consists of the following $\QPT$ algorithms:
\begin{description}
\item $\KeyGen(1^\lambda,1^L) \rightarrow (\pk,\sk):$ 
sample a pair $(\vec A \in \Z_q^{n \times m},\mathsf{td}_{\vec A}) \leftarrow \GenTrap(1^n,1^m,q)$ and generate a Gaussian superposition $(\ket{\psi_{\vec y}}, \vec y) \leftarrow \mathsf{GenGauss}(\vec A,\sigma)$ with
$$
\ket{\psi_{\vec y}} \,\,= \sum_{\substack{\vec x \in \Z_q^m\\ \vec A\vec x = \vec y}}\rho_{\sigma}(\vec x)\,\ket{\vec x},
$$
for some $\vec y \in \Z_q^n$.
Output $\pk = (\vec A,\vec y)$, $\sk = \ket{\psi_{\vec y}}$ and $\msk = \mathsf{td}_{\vec A}$.

\item $\Enc(\pk,\mu):$ to encrypt $\mu \in \bit$, parse $(\vec A,\vec y) \leftarrow \pk$, sample a random matrix $\vec S \rand \Z_q^{n \times N}$ and $\vec E \sim D_{\Z^{m\times N}, \,\alpha q}$ and row vector $\vec e \sim D_{\Z^{N}, \,\beta q}$, and
output the ciphertext 
$$
\ct= \left[\substack{
\vec A^\intercal \vec S + \vec E\vspace{1mm}\\
\hline\vspace{1mm}\\
\vec y^\intercal  \vec S + \vec e}
 \right] 
+ \mu \cdot \vec G \Mod{q} \in \Z_q^{(m+1)\times N}.
$$

\item $\Eval(\ct_0,\ct_1):$ to apply a $\mathsf{NAND}$ gate on a ciphertext pair $\ct_0$ and $\ct_1$, output the matrix
$$
\vec G - \ct_0 \cdot \vec G^{-1}(\ct_1) \Mod{q} \in \Z_q^{(m+1)\times N}.
$$

\item $\Dec(\sk,\ct) \rightarrow \bit:$ to decrypt $\ct$, 
apply the unitary $U: \ket{\vec x}\ket{0} \rightarrow \ket{\vec x}\ket{(-\vec x,1) \cdot \ct_N}$ on input $\ket{\psi_{\vec y}}\leftarrow\sk$, where $\ct_N \in \Z_q^{m+1}$ is the $N$-th column of $\ct$, and measure the second register in the computational basis. Output $0$, if the measurement outcome
is closer to $0$ than to $\lfloor \frac{q}{2} \rfloor$,
and output $1$, otherwise.

\item $\revoke(\msk,\pk,\rho) \rightarrow \{\top,\bot\}$: on input $\mathsf{td}_{\vec A} \leftarrow \msk$ and $(\vec A,\vec y) \leftarrow \pk$, apply the projective measurement $\{\ketbra{\psi_{\vec y}}{\psi_{\vec y}},I  - \ketbra{\psi_{\vec y}}{\psi_{\vec y}}\}$ onto $\rho$ using $\mathsf{SampGauss}(\vec A,\mathsf{td}_{\vec A},\vec y,\sigma)$ in Algorithm \ref{alg:SampGauss}. Output $\top$ if the measurement is successful, and output $\bot$ otherwise.
\end{description}
\end{construction}

\subsection{Proof of security}

\begin{figure}[!htb]
   \begin{center} 
   \begin{tabular}{|p{14cm}|}
    \hline 
\begin{center}
\underline{$\expt_{\adversary}(1^{\secparam},b)$}: 
\end{center}
\begin{enumerate}
    \item The challenger samples $(\vec A \in \Z_q^{n \times m},\mathsf{td}_{\vec A}) \leftarrow \GenTrap(1^n,1^m,q)$ and generates
$$
\ket{\psi_{\vec y}} \,\,= \sum_{\substack{\vec x \in \Z_q^m\\ \vec A \vec x = \vec y \Mod{q}}}\rho_{\sigma}(\vec x)\,\ket{\vec x},
$$
for some $\vec y \in \Z_q^n$, by running $(\ket{\psi_{\vec y}}, \vec y) \leftarrow \mathsf{GenGauss}(\vec A,\sigma)$. The challenger lets $\msk \leftarrow \mathsf{td}_{\vec A}$ and $\pk \leftarrow (\vec A,\vec y)$ and sends  $\sk \leftarrow \ket{\psi_{\vec y}}$ to the adversary $\mathcal A$.
\item $\mathcal A$ generates a (possibly entangled) bipartite state $\rho_{R,\vec{\textsc{aux}}}$ in systems $\algo H_R \otimes \algo H_{\textsc{Aux}}$ with $\algo H_R= \algo H_q^m$, returns system $R$ and holds onto the auxiliary system $\textsc{Aux}$.
\item The challenger runs $\revoke(\msk,\pk,\rho_R)$, where $\rho_R$ is the reduced state in system $R$. If the outcome is $\top$, the game continues. Otherwise, output \textsf{Invalid}.

\item $\algo A$ submits a plaintext bit $\mu \in \bit$.

\item The challenger does the following depending on $b \in \bit$:
\begin{itemize}
    \item if $b=0$: The challenger samples a random matrix $\vec S \rand \Z_q^{n \times N}$ and errors $\vec E \sim D_{\Z^{m\times N}, \,\alpha q}$ and row vector $\vec e \sim D_{\Z^{N}, \,\beta q}$, and
outputs the ciphertext 
$$
\ct
= \left[\substack{
\vec A^\intercal \vec S + \vec E\vspace{1mm}\\
\hline\vspace{1mm}\\
\vec y^\intercal  \vec S + \vec e}
 \right] + \mu \cdot \vec G
\,\, \in \Z_q^{(m+1)\times N}.
$$

\item if $b=1$: the challenger samples a matrix $\vec U \rand \Z_q^{m \times N}$ and row vector $r \rand \Z_q^N$ uniformly at random, and sends the following to $\algo A$:
$$
\left[\substack{
\, \vec U \, \vspace{1mm}\\
\hline\vspace{1mm}\\
\vec r}
 \right] 
\,\, \in \Z_q^{(m+1)\times N}.
$$
\end{itemize}

\item $\algo A$ returns a bit $b' \in \bit$.
\end{enumerate}\\
\hline
\end{tabular}
    \caption{The key-revocable security experiment according to \Cref{def:krpke:security}.}
    \label{fig:Dual-FHE-security}
    \end{center}
\end{figure}

Our first result on the security of \Cref{cons:DualGSW} concerns $(\negl(\lambda),\negl(\lambda))$-security, i.e., we assume that revocation succeeds with overwhelming probability.

\begin{theorem}\label{thm:security-Dual-GSW}
Let $L$ be an upper bound on the $\mathsf{NAND}$-depth of the circuit which is to be evaluated. Let $n\in \N$ and $q$ be a prime modulus with $n=n(\lambda,L) \gg L$, $q=2^{o(n)}$ and $m \geq 2n \log q$, each parameterized by the security parameter $\lambda \in \N$. Let $N = (m+1) \cdot \lceil \log q \rceil$ be an integer. Let $q/\sqrt{8m}> \sigma > \sqrt{8m}$ and let $\alpha,\beta \in (0,1)$ be parameters such that $\beta/\alpha =2^{o(n)}$ and $1/\alpha= 2^{o(n)} \cdot \sigma$.
Then, assuming the subexponential hardness of the $\mathsf{LWE}_{n,q,\alpha q}^m$
and $\mathsf{SIS}_{n,q,\sigma\sqrt{2m}}^m$ problems, the scheme
$\mathsf{RevDualGSW} = (\KeyGen,\Enc,\Dec,\Eval,\Revoke)$ in \Cref{cons:DualGSW}
is a $(\negl(\lambda),\negl(\lambda))$-secure key-revocable (leveled) fully homomorphic encryption scheme according to \Cref{def:krpke:security}.
\end{theorem}

\begin{proof} Let $\algo A$ be a $\QPT$ adversary and suppose that
$$ \vline \, \Pr\left[ 1 \leftarrow \expt_{\adversary}(1^{\secparam},0) \right] - \Pr\left[ 1 \leftarrow \expt_{\adversary}(1^{\secparam},1) \right] \,\vline =\epsilon(\secparam),$$
for some $\eps(\lambda)$ with respect to $\expt_{\adversary}(1^{\secparam},b)$ in \Cref{fig:Dual-FHE-security}. Note that the \textsf{RevDualGSW} ciphertext can (up to an additive shift) be thought of as a column-wise concatenation of $N$-many independent ciphertexts of our key-revocable Dual-Regev scheme in \Cref{cons:dual-regev}.
Therefore, we can invoke \Cref{thm:security-Dual-Regev-high-revoke} in order to argue that $\eps(\lambda)$ is at most negligible.

\end{proof}

Our second result concerns $(\negl(\lambda),1-1/\poly(\lambda))$-security, i.e., we do not make any requirements on the success probability of revocation. Here, we need to invoke \Cref{thm:search-to-decision}.

\begin{theorem}\label{thm:security-Dual-GSW}
Let $L$ be an upper bound on the $\mathsf{NAND}$-depth of the circuit which is to be evaluated. Let $n\in \N$ and $q$ be a prime modulus with $n=n(\lambda,L) \gg L$, $q=2^{o(n)}$ and $m \geq 2n \log q$, each parameterized by the security parameter $\lambda \in \N$. Let $N = (m+1) \cdot \lceil \log q \rceil$ be an integer. Let $q/\sqrt{8m}> \sigma > \sqrt{8m}$ and let $\alpha,\beta \in (0,1)$ be parameters such that $\beta/\alpha =2^{o(n)}$ and $1/\alpha= 2^{o(n)} \cdot \sigma$.
Then, assuming assuming \Cref{thm:search-to-decision}, the scheme
$\mathsf{RevDualGSW} = (\KeyGen,\Enc,\Dec,\Eval,\Revoke)$ in \Cref{cons:DualGSW}
is a $(\negl(\lambda),1-1/\poly(\lambda))$-secure key-revocable (leveled) fully homomorphic encryption scheme according to \Cref{def:krpke:security}.
\end{theorem}

\begin{proof} The proof is the same as in the theorem before, except that we invoke  \Cref{thm:security-Dual-Regev} instead of \Cref{thm:security-Dual-Regev-high-revoke} in order to argue security.
\end{proof}

\section{Revocable Pseudorandom Functions}

In this section, we introduce the notion of \emph{key-revocable} pseudorandom functions (or simply, called {\em revocable}) and present the first construction from (quantum hardness of) learning with errors.

\subsection{Definition}

Let us first recall the traditional notion of $\PRF$ security~\cite{GGM86}, defined as follows.

\begin{definition}[Pseudorandom Function]\label{def:pqPRF} Let $\lambda \in \N$ and $\kappa(\lambda),\ell(\lambda)$ and $\ell'(\lambda)$ be polynomials. A (post-quantum) pseudorandom function $(\pqPRF)$ is a pair $(\gen,\PRF)$ of $\PPT$ algorithms given by
\begin{itemize}
    \item $\gen(1^\lambda):$ On input $1^\lambda$, it outputs a key $k \in \bit^\kappa$.

    \item $\PRF(k,x):$ On input $k \in \bit^\kappa$ and $x \in \bit^{\ell}$, it outputs a value $y \in \bit^{\ell'}$.
\end{itemize}
with the property that, for any $\QPT$ distinguisher $\algo D$, we have
\begin{align*}
\vline \, \Pr\left[\algo D^{\PRF(k,\cdot)}(1^\lambda) = 1] \, : \, k \leftarrow \gen(1^\lambda)
\right]
- \Pr\left[\algo D^{F(\cdot)}(1^\lambda) = 1] \, : \, 
F \rand \algo F^{\ell,\ell'}
\right] \,\vline \,\leq \, \negl(\lambda),
\end{align*}
where $\algo F^{\ell,\ell'}$ is the set of all functions with domain $\bit^\ell$ and range $\bit^{\ell'}$.
    
\end{definition}

\noindent We now present a formal definition of revocable pseudorandom functions below.

\begin{definition}[Revocable Pseudorandom Function]\label{def:rprf}
Let $\lambda \in \N$ be the security parameter and let $\kappa(\lambda),\ell(\lambda)$ and $\ell'(\lambda)$ be polynomials.
A revocable pseudorandom function $(\mathsf{rPRF})$ is a scheme
$(\gen, \prf,\eval,\revoke)$ consisting of the following efficient algorithms:

\begin{itemize}
    \item $\gen(1^{\secparam})$: on input the security parameter $\secparam \in \N$, it outputs a $\prf$ key $k \in \bit^\kappa$, a quantum state $\rho_k$ and a master secret key $\msk$. 
    \item $\prf(k,x)$: on input a key $k \in \bit^\kappa$ and an input string $x \in \{0,1\}^{\ell}$, it outputs a value $y \in \{0,1\}^{\ell'}$. This is a deterministic algorithm. 
    \item $\eval(\rho_k,x)$: on input a state $\rho_k$ and an input $x \in \{0,1\}^{\ell}$, it outputs a value $y \in \{0,1\}^{\ell'}$. 
    \item $\revoke(\msk,\sigma)$: on input key $\msk$ and a state $\sigma$, it outputs $\valid$ or $\invalid$. 
\end{itemize}
\end{definition}
\noindent We additionally require that the following holds:
\paragraph{Correctness.} 
For each $(k,\rho_k,\msk)$ in the support of $\gen(1^{\secparam})$ and for every $x \in \{0,1\}^{\ell}$:
\begin{itemize}
    \item (Correctness of evaluation:)
$$ \prob\left[ \prf(k,x) = \eval(\rho_k,x) \right] \geq 1 - \negl(\secparam).$$
    \item (Correctness of revocation:)
$$ \prob\left[ \valid \leftarrow \revoke(\msk,\rho_k)\right] \geq 1 - \negl(\secparam).$$
\end{itemize}

\subsection{Security} 
We define revocable $\prf$ security below. 

\begin{figure}[!htb]
   \begin{center} 
   \begin{tabular}{|p{14cm}|}
    \hline 
\begin{center}
\underline{$\expt_{\adversary,\mu}(1^{\secparam},b)$}: 
\end{center}
\noindent {\bf Initialization Phase}:
\begin{itemize}
    \item The challenger computes $(k,\rho_k,\msk) \leftarrow \mathsf{Gen}(1^\lambda)$ and sends $\rho_k$ to $\adversary$. 
\end{itemize}
\noindent {\bf Revocation Phase}:
\begin{itemize}
    \item The challenger sends the message \texttt{REVOKE} to $\adversary$. 
    \item The adversary $\adversary$ sends a state $\sigma$ to the challenger.
    \item The challenger aborts if $\revoke\left(\msk,\sigma\right)$ outputs $\invalid$. 
\end{itemize}
\noindent {\bf Guessing Phase}:
\begin{itemize}
    \item The challenger samples bit $b \leftarrow \{0,1\}$.  
    \item  The challenger samples random inputs $x_1,\dots,x_\mu \rand \bit^\ell$ and then sends the values $(x_1,\dots,x_\mu)$ and $(y_1,\dots,y_\mu)$ to $\adversary$, where: 
    \begin{itemize}
        \item If $b=0$, set $y_1=\mathsf{PRF}(k,x_1),\, \dots, \, y_\mu=\mathsf{PRF}(k,x_\mu)$ and,
        \item If $b=1$, set $y_1,\dots,y_\mu \rand \{0,1\}^{\ell'}$. 
    \end{itemize}
    \item $\adversary$ outputs a bit $b'$ and wins if $b'=b$. 
\end{itemize}
\ \\ 
\hline
\end{tabular}
    \caption{Revocable $\PRF$ security}
    \label{fig:KR-PRF-Model2}
    \end{center}
\end{figure}

\begin{definition}[Revocable $\prf$ Security]\label{def:revocable-PRF}
A revocable pseudorandom function $(\rPRF)$ consisting of a tuple of $\QPT$ algorithms $(\gen, \prf,\eval,\revoke)$ has $(\eps,\delta,\mu)$ revocable $\PRF$ security if, for every $\QPT$ adversary $\adversary$ with 
$$\prob[\invalid \leftarrow \expt_{\algo A,\mu}(1^\lambda,b)] \leq \delta(\secparam)
$$
for $b \in \bit$, it holds that
$$ \left| \Pr\left[ 1 \leftarrow \expt_{\algo A,\mu}(1^{\secparam},0) \right] - \Pr\left[ 1 \leftarrow \expt_{\algo A,\mu}(1^{\secparam},1) \right]\right|\leq  \eps(\secparam),$$
where $\expt_{\adversary,\mu}$ is as defined in~\Cref{fig:KR-PRF-Model2}. If $\delta(\secparam) = 1-1/\poly(\lambda)$, $\eps(\secparam)=\negl(\lambda)$ we oftentimes drop $(\delta,\eps)$ and simply refer to it as $\rPRF$ satisfies $\mu$-revocable $\PRF$ security.
\end{definition}

\paragraph{From one-query to multi-query security.} We show that proving security with respect to $\mu=1$ is sufficient. That is, we show the following. 

\begin{claim}\label{claim:one-query-to-multi-query}
Supoose an $\rPRF$ scheme
$(\gen, \prf,\eval,\revoke)$ satisfies $1$-revocable $\PRF$ security. Then, $\rPRF$ also satisfies the stronger notion of (multi-query) revocable PRF security.  
\end{claim}
\begin{proof}
Let $\adversary$ be a $\QPT$ adversary that participating in the revocable $\PRF$ security experiment defined in~\Cref{fig:KR-PRF-Model2} and let $(x_1,y_1),\ldots,(x_{\mu},y_{\mu})$ denote the challenge input-output pairs, for some polynomial $\mu=\mu(\secparam)$. In the following, we denote by $k$ the $\PRF$ key sampled using $\gen$ by the challenger in~\Cref{fig:KR-PRF-Model2}. We consider a sequence of hybrids defined as follows.\\

 \noindent $\hybrid_{i}$, for $i \in [\mu+1]$: In this hybrid, $y_1,\ldots,y_{i-1}$ are sampled uniformly at random from $\{0,1\}^{\ell'}$ and $y_i,\ldots,y_{\mu}$ are generated as follows: $y_j=\prf(k,x_{j})$ for $j \geq i$.\\ 

 \noindent We claim that $\adversary$ cannot distinguish between hybrids $\hybrid_i$ and $\hybrid_{i+1}$, for all $i \in [\mu]$, with more than negligible advantage. Suppose for the sake of contradiction that the claim is not true, and that $\algo A$ can distinguish $\hybrid_i$ and $\hybrid_{i+1}$, for some index $i \in [\mu]$, with advantage at least $\eps(\lambda) = 1/\poly(\lambda)$.  We will now show that we can use the adversary $\adversary$ to break the 1-revocation security of $\rPRF$.
 
 Consider a reduction ${\cal B}$ that does the following: 
 \begin{enumerate}
     \item Receive the state $\rho_k$ from the challenger. 
     \item Sample $x_{i+1},\ldots,x_{\mu}$ uniformly at random from $\{0,1\}^{\ell}$. Denote $\rho_{k}^{(i+1)}=\rho_k$. Do the following for $j=i+1,\ldots,\mu$: $\eval(\rho_k^{(j)},x_j)$ to obtain $y_j$. Using the ``Almost As Good As New''~\cite{aaronson2016complexity}, recover $\rho_{k}^{(j+1)}$, where $\rho_{k}^{(j+1)}$ is negligibly\footnote{Technically, this depends on the correctness error and we start with a $\rPRF$ that is correct with probability negligibly close to 1.} close to $\rho_k$ in trace distance. 
     \item Forward the state $\rho_k^{(\mu+1)}$ to $\adversary$. 
     \item When the challenger submits the message  \texttt{REVOKE}, forward the same message to $\adversary$. 
     \item If $\adversary$ sends $\sigma$, then forward the same state $\sigma$ to the challenger. 
     \item If the revocation did not fail, the guessing phase begins. The challenger sends $(x^*,y^*)$. Then, sample $x_1,\ldots,x_{i-1}$ uniformly at random from $\{0,1\}^{\ell}$ and $y_1,\ldots,y_{i-1}$ uniformly at random from $\{0,1\}^{\ell'}$. Set $x_i=x^*$ and $y_i=y^*$. Send $(x_1,y_1),\ldots,(x_{\mu},y_{\mu})$ to $\adversary$. 
     \item Output $b$, where $b$ is the output of $\adversary$.  
 \end{enumerate}
 From the quantum union bound~(\Cref{lem:union}), the ``Almost As Good As New'' lemma (\Cref{lem:almost}) and the correctness of $\rPRF$, it follows that $\TD(\rho_k,\rho_{k}^{(\mu+1)}) \leq \negl(\secparam)$ and thus, the advantage of $\adversary$ when given $\rho_k^{(\mu+1)}$ instead of $\rho_k$ is now at least $\eps - \negl(\secparam)$. Moreover, by the design of ${\cal B}$, it follows that the success probability of ${\cal B}$ in breaking 1-revocation security of $\rPRF$ is exactly the same as the success probability of $\adversary$ in breaking revocation security of $\rPRF$. This contradicts the fact that $\rPRF$ satisfies 1-revocation security.
 
\end{proof}

\begin{remark}\label{remark:traditional-security}
Our notion of revocable $\PRF$ security from \Cref{def:revocable-PRF} does not directly imply traditional notion of $\pqPRF$ security\footnote{Although any revocable $\PRF$ is a {\em weak} PRF. Recall that a weak PRF is one where the adversary receives as input $(x_1,y_1),\ldots,(x_{\mu},y_{\mu})$, where $x_i$s are picked uniformly at random. The goal of the adversary is to distinguish the two cases: all $y_i$s are pseudorandom or all $y_i$s are picked uniformly at random.} from \Cref{def:pqPRF}. The reason is that the definition does not preclude the possibility of there being an input $x$ (say an all zeroes string) on which, $\PRF$ outputs $x$ itself (or the first bit of $x$ if the output of $\PRF$ is a single bit). 
\end{remark}

\noindent Motivated by \Cref{remark:traditional-security}, we now introduce the following notion of a \emph{strong} $\rPRF$.

\begin{definition}[Strong $\rPRF$]\label{def:strong-rprf}
We say that a scheme
$(\gen, \prf,\eval,\revoke)$ is a strong revocable pseudorandom function (or, strong $\rPRF$) if the following two properties hold:
\begin{enumerate}
    \item $(\gen, \prf,\eval,\revoke)$ satisfy revocable $\PRF$ security according to \Cref{def:revocable-PRF}, and
    \item $(\gen,\PRF)$ satisfy $\pqPRF$ security according to \Cref{def:pqPRF}.
\end{enumerate}
\end{definition}

\begin{remark}
When instantiating pseudorandom functions in the textbook construction of private-key encryption~\cite{oded06} from revocable pseudorandom functions, we immediately obtain a revocable private-key encryption scheme. 
\end{remark}

\noindent We show that the issue raised in \Cref{remark:traditional-security} is not inherent. In fact, we give a simple generic transformation that allows us to obtain strong $\rPRF$s by making use of traditional $\pqPRF$s.

\begin{claim}[Generic Transformation for Strong $\rPRF$s]
Let $(\gen, \prf,\eval,\revoke)$ be an $\rPRF$ scheme which satisfies revocable $\PRF$ security, and let  
$(\overline{\gen},\overline{\prf})$ be a $\pqPRF$.    
Then, the scheme $(\widetilde{\gen}, \widetilde{\prf},\widetilde{\eval},\widetilde{\revoke})$ is a strong $\rPRF$ which consists of the following algorithms:
\begin{itemize}
    \item $\widetilde{\gen}(1^{\secparam})$: on input the security parameter $1^{\secparam}$, first run $(k,\rho_k,\msk) \leftarrow \gen(1^\lambda)$ and then output $((K,k),(K,\rho_k),\msk)$, where $K \leftarrow \overline{\gen}(1^\lambda)$ is a $\pqPRF$ key.
    
    \item $\widetilde{\prf}((K,k),x)$: on input a key $(K,k)$ and string $x \in \{0,1\}^{\ell}$, output $\overline{\prf}(K,x) \oplus \PRF(k,x)$. 
    
    \item $\widetilde{\eval}((K,\rho_k),x)$: on input $(K,\rho_k)$ and $x \in \{0,1\}^{\ell}$, output $\overline{\prf}(K,x) \oplus \Eval(\rho_k,x)$. 
    \item $\widetilde{\revoke}(\msk,(K,\sigma))$: on input a master secret key $\msk$ and a pair $(K,\rho_k)$, first discard the key $K$ and then run $\revoke(\msk,\sigma)$.
\end{itemize}
\end{claim}

\begin{proof}
Let us first show that the scheme $(\widetilde{\gen}, \widetilde{\prf},\widetilde{\eval},\widetilde{\revoke})$ maintains revocable $\PRF$ security.
Suppose that there exists a $\QPT$ adversary $\algo A$ and a polynomial $\mu = \mu(\lambda) \in \N$ such that
$$  \left| \Pr\left[ 1 \leftarrow \expt_{\algo A,\mu}(1^{\secparam},0)\right] - \Pr\left[ 1 \leftarrow \expt_{\algo A,\mu}(1^{\secparam},1) \right]\right| = \epsilon(\secparam),$$
for some function $\epsilon(\lambda) = 1/\poly(\lambda)$, and where $\expt_{\adversary,\mu}$ is the experiment from~\Cref{fig:KR-PRF-Model2} with respect to the scheme $(\widetilde{\gen}, \widetilde{\prf},\widetilde{\eval},\widetilde{\revoke})$. We show that this implies the existence of a $\QPT$ distinguisher $\algo D$ that breaks the revocable $\PRF$ security of the scheme $(\gen, \prf,\eval,\revoke)$.

The distinguisher $\algo D$ proceeds as follows:
\begin{enumerate}
    \item $\algo D$ receives as input a quantum state $\rho_k$, where $(k,\rho_k,\msk) \leftarrow \gen(1^\lambda)$ is generated by the challenger. Then, $\algo D$ generates a $\pqPRF$ key $K \leftarrow \overline{\gen}(1^\lambda)$ and sends $(K,\rho_k)$ to $\algo A$.

    \item When $\algo A$ returns a state $\rho$, $\algo D$ forwards it to the challenger as part of the revocation phase.

    \item When $\algo D$ receives the challenge input $(x_1,\dots,x_\mu)$ and $(y_1,\dots,y_\mu)$ from the challenger, $\algo D$ sends $(x_1,\dots,x_\mu)$ and $(\overline{\prf}(K,x_1) \oplus y_1,\dots,\overline{\prf}(K,x_\mu) \oplus y_\mu)$ to $\algo A$.

    \item When $\algo A$ outputs $b'$, so does the distinguisher $\algo D$.
\end{enumerate}
Note that the simulated challenge distribution above precisely matches the challenge distribution from the experiment $\expt_{\adversary,\mu}$ from ~\Cref{fig:KR-PRF-Model2}. Therefore, if $\algo A$ succeeds with inverse polynomial advantage $\epsilon(\lambda) = 1/\poly(\lambda)$, so does $\algo D$ -- thereby breaking the revocable $\PRF$ security of the scheme $(\gen, \prf,\eval,\revoke)$. Consequently, $(\widetilde{\gen}, \widetilde{\prf},\widetilde{\eval},\widetilde{\revoke})$ satisfies revocable $\PRF$ security.

To see why $(\widetilde{\gen},\widetilde{\PRF})$ satisfy $\pqPRF$ security according to \Cref{def:pqPRF}, we can follow a similar argument as above to break the $\pqPRF$ security of $(\overline{\gen},\overline{\prf})$. Here, we rely on the fact that the keys $(k,\rho_k,\msk) \leftarrow \gen(1^\lambda)$ and $K \leftarrow \overline{\gen}(1^\lambda)$ are sampled independently from another.

\end{proof}

\subsection{Construction} 
\label{sec:intlemma}
\noindent We construct a $\PRF$ satisfying 1-revocation security (\Cref{def:revocable-PRF}). 


\newcommand{\bfz}{\textbf{z}}

\paragraph{Shift-Hiding Construction.}
We construct a \emph{shift-hiding} function which is loosely inspired by shift-hiding shiftable functions introduced by Peikert and Shiehian~\cite{PS18}.

Let $n,m \in \N$,  $q \in \N$ be a modulus and let $\ell = n m \lceil \log q \rceil$. In the following, we consider matrix-valued functions
$F :\{0,1\}^{\ell} \rightarrow \Zq^{n \times m}$, where $F$ is one of the following functions:
    \begin{itemize}
        \item $\zerof :\{0,1\}^{\ell} \rightarrow \Zq^{n \times m}$ which, on input $x \in \bit^\ell$, outputs an all zeroes matrix $\vec 0 \in \Z_q^{n \times m}$, or:
        \item $H_r:\{0,1\}^{\ell} \rightarrow \Zq^{n \times m}$ which, on input $x \in \{0,1\}^{\ell}$, outputs $\vec M \in \Z_q^{n \times m}$, where $r \in \{0,1\}^{\ell}$ and $x=r\oplus \bindecomp(\vec M)$, where $\vec M \in \Zq^{n \times m}$ and $\bindecomp(\cdot)$ takes as input a matrix and outputs a binary string that is obtained by concatenating the binary decompositions of all the elements in the matrix (in some order). 
    \end{itemize}
We show that there exist $\PPT$ algorithms $(\lemkg,\lemeval)$ (formally defined in \Cref{const:shift-hiding-const}) with the following properties:
\begin{itemize}
    \item $\lemkg(1^n,1^m,q,\bfA,F)$: on input $1^n,1^m$, a modulus $q \in \N$, a matrix $\bfA \in \Zq^{n \times m}$ and a function $F \in \{\algo Z\} \cup \{H_r:r \in \{0,1\}^{\ell}\}$, it outputs a pair of keys $(pk_F,sk_F)$.
    \item $\lemeval(pk_F,x)$: on input $pk_F$, $x \in \{0,1\}^{\ell}$, it outputs $\bfS_x \bfA + \bfE_x + F(x)$, where $\bfS_x \in \Zq^{n \times n}$ and $\bfE_x \in \Zq^{n \times m}$, where $||\bfE_x||\leq nm^2 \sigma \lceil \log(q) \rceil$. Moreover, there is an efficient algorithm that recovers $\bfS_x$ given $sk_F$ and $x$.  
\end{itemize}
We show that our construction of $(\lemkg,\lemeval)$ satisfies a 
\emph{shift-hiding property}; namely, for any $r \in \{0,1\}^{\ell}$,
$$\{pk_{\zerof} \} \approx_c \{ pk_{H_r}\},$$
for any $pk_F$ with $(pk_F,sk_F) \leftarrow \lemkg(1^n,1^m,q,\bfA,F)$, where $\bfA \xleftarrow{\$} \Zq^{n \times m}$, and $F \in \{\algo Z,H_r\}$.
\par In the construction below, we consider a bijective function $\phi:[n] \times [m] \times [\lceil \log(q) \rceil] \rightarrow [\ell]$.

\begin{construction}\label{const:shift-hiding-const}
Consider the $\PPT$ algorithms $(\lemkg,\lemeval)$  defined as follows:
\begin{itemize}
    \item $\lemkg(1^n,1^m,q,\bfA,F)$:  on input $1^n,1^m$, a modulus $q \in \N$, a matrix $\bfA \in \Zq^{n \times m}$ and function $F \in \{\algo Z\} \cup \{H_r:r \in \{0,1\}^{\ell}\}$, it outputs a pair of keys $\lemkey_{F}=(pk_F,sk_F)$ generated as follows: \begin{enumerate}
    \item For every $i,j \in [n],\tau \in [\lceil \log(q) \rceil]$, define $\{\bfM_{b}^{(i,j,\tau)}\}$ as follows:
    \begin{itemize}
        \item If $F=\zerof$, then for every $i \in [n],j \in [m],\tau \in [\lceil \log(q) \rceil]$, let $\bfM_{b}^{(i,j,\tau)}= \vec 0 \in \Z_q^{n \times n}$,
        \item If $F=H_r$, then for every $i \in [n],j \in [m],\tau \in [\lceil \log(q) \rceil]$, let $\bfM_{b}^{(i,j,\tau)}=(b \oplus r_{\phi(i,j,\tau)}) \cdot {\bf I}_{n \times n}$. 
    \end{itemize}

    \item For every $i \in [n],j \in [m],\tau \in [\lceil \log(q) \rceil]$, $b \in \{0,1\}$, compute:
    $$pk_b^{(i,j,\tau)} = \bfS_b^{(i,j,\tau)} \bfA + \bfE_{b}^{(i,j,\tau)} + \bfM_b^{(i,j,\tau)},$$
    $$sk_b^{(i,j,\tau)}=\left( \left\{\bfS_{b}^{(i,j,\tau)},\bfE_b^{(i,j,\tau)}\right\} \right),$$
    where for every $i \in [n],j \in [m],\tau \in [\lceil \log(q) \rceil]$, $b \in \{0,1\}$:
    \begin{itemize}
        \item $\bfS_{b}^{(i,j,\tau)} \leftarrow D_{\Zq,\sigma}^{n \times n}$, 
        \item $\bfE_{b}^{(i,j,\tau)} \leftarrow D_{\Zq,\sigma}^{n \times m}$
    \end{itemize}
    \item Output $pk_F=\left(\bfA,\left\{ pk^{(i,j,\tau)}_b \right\}_{\substack{i \in [n],j \in [m],\\ \tau \in [\lceil \log(q) \rceil],b \in \{0,1\}}} \right)$ and $sk_F=\left\{ sk^{(i,j,\tau)}_b \right\}_{\substack{i \in [n],j \in [m],\\ \tau \in [\lceil \log(q) \rceil],b \in \{0,1\}}}$. 
\end{enumerate}

\item $\lemeval(pk_F,x)$: on input $pk_F$ and $x \in \bit^\ell$, proceed as follows:
\begin{enumerate}
 \item Parse $pk_F=\left( \bfA \left\{ pk_b^{(i,j,\tau)} \right\}_{\substack{i \in [n],j \in [m],\\ \tau \in [\lceil \log(q) \rceil], b \in \{0,1\}}} \right)$
\item Output $\sum_{\substack{i \in [n],j \in [m],\\ \tau \in [\lceil \log(q) \rceil]}} pk_{x_{\phi(i,j,\tau)}}^{(i,j,\tau)}$.
\end{enumerate}
\end{itemize}
\end{construction}

\begin{claim}[Correctness]\label{claim:kg-e-correctness} Let $(\lemkg,\lemeval)$ be the pair of $\PPT$ algorithms in \Cref{const:shift-hiding-const}. Let $(pk_F,sk_F) \leftarrow \lemkg(1^n,1^m,q,\bfA,F)$ with $F \in \{\algo Z\} \cup \{H_r:r \in \{0,1\}^{\ell}\}$. Then, the output of $\lemeval(pk_F,x)$ is of the form:
$$\lemeval(pk_F,x) = \bfS_x \bfA + \bfE_x + F(x),$$ where $\bfS_x \in \Zq^{n \times n}$ and $\bfE_x \in \Zq^{n \times m}$ with $||\bfE_x|| \leq nm^2 \sigma \lceil \log(q) \rceil$. Moreover, there is an efficient algorithm that recovers $\bfS_x$ given $(pk_F,sk_F)$.  
\end{claim}
\begin{proof}
Let $(pk_F,sk_F) \leftarrow \lemkg(1^n,1^m,q,\bfA,F)$. Parse $pk_F=\left(\bfA,\left\{ pk^{(i,j,\tau)}_b \right\}_{\substack{i \in [n],j \in [m],\\ \tau \in [\lceil \log(q) \rceil],b \in \{0,1\}}} \right)$ and $sk_F=\left\{ sk^{(i,j,\tau)}_b \right\}_{\substack{i \in [n],j \in [m],\\ \tau \in [\lceil \log(q) \rceil],b \in \{0,1\}}}$, where: 
  $$pk_b^{(i,j,\tau)} = \bfS_b^{(i,j,\tau)} \bfA + \bfE_{b}^{(i,j,\tau)} + \bfM_b^{(i,j,\tau)},$$
    $$sk_b^{(i,j,\tau)}=\left( \{\bfS_{b}^{(i,j,\tau)},\bfE_b^{(i,j,\tau)}\} \right)$$
 \noindent There are two cases to consider here: \\

 \noindent \textbf{Case 1}. $F={\cal Z}$: in this case, $\bfM_{b}^{(i,j,\tau)} = \bf0$, for every $i \in [n],j \in [m], \tau \in [\lceil \log(q) \rceil], b \in \{0,1\}$. Thus, the following holds: 

 \begin{eqnarray*}
\sum_{\substack{i \in [n],j \in [m],\\ \tau \in [\lceil \log(q) \rceil]}} pk_{x_{\phi(i,j,\tau)}}^{(i,j,\tau)} & = & \underbrace{\left( \sum_{\substack{i \in [n],j \in [m],\\ \tau \in [\lceil \log(q) \rceil]}} \bfS_{x_{\phi(i,j,\tau)}}^{(i,j,\tau)} \right)}_{\bfS_x}  \bfA + \underbrace{\left( \sum_{\substack{i \in [n],j \in [m],\\ \tau \in [\lceil \log(q) \rceil]}} \bfE_{x_{\phi(i,j,\tau)}}^{(i,j,\tau)} \right)}_{\bfE_x} + \left( \sum_{\substack{i \in [n],j \in [m],\\ \tau \in [\lceil \log(q) \rceil]}} \bfM_{x_{\phi(i,j,\tau)}}^{(i,j,\tau)} \right) \\
& = & \bfS_x \bfA + \bfE_x + {\cal Z}(x) 
 \end{eqnarray*}

 \noindent Moreover, we have that $||\bfE_b^{(i,j,\tau)}|| \leq m \sigma$ and thus, $|| \bfE_x || \leq nm^2 \sigma \lceil \log(q) \rceil$. \\

\noindent \textbf{Case 2}. $F={H}_r$:
 \begin{eqnarray*}
\sum_{\substack{i \in [n],j \in [m],\\ \tau \in [\lceil \log(q) \rceil]}} pk_{x_{\phi(i,j,\tau)}}^{(i,j,\tau)} & = & \bfS_x \bfA + \bfE_x + \left( \sum_{\substack{i \in [n],j \in [m],\\ \tau \in [\lceil \log(q) \rceil]}} \bfM_{x_{\phi(i,j,\tau)}}^{(i,j,\tau)} \right) \\ 
& = & \bfS_x \bfA + \bfE_x + H_r(x),
 \end{eqnarray*}
where $\bfS_x$ and $\bfE_x$ are as defined above. The second equality holds because of the fact that $\bfM_{x_{\phi(i,j,\tau)}}^{(i,j,\tau)}$ has the value $(b \oplus r_{\phi(i,j,\tau)}) \cdot 2^{\tau}$ in the $(i,j)^{th}$ position and zero, everywhere else. Thus, summing up all the $\bfM_{x_{\phi(i,j,\tau)}}^{(i,j,\tau)}$ matrices results in the matrix $\bfM$, where $x \oplus r$ is the binary decomposition of $\bfM$. \\

\noindent Finally, it is clear that $\bfS_x$ can be efficiently recovered from $sk_F$ and $x$. 
\end{proof}

\begin{claim}[Shift-hiding property]\label{claim:kg-e-shift-hiding} Assuming the quantum hardness of learning with errors, the pair $(\lemkg,\lemeval)$ in \Cref{const:shift-hiding-const} has the property that
$$\{pk_{\zerof} \} \approx_c \{ pk_{H_r}\},$$
for any $pk_F$ with $(pk_F,sk_F) \leftarrow \lemkg(1^n,1^m,q,\bfA,F)$, where $\bfA \xleftarrow{\$} \Zq^{n \times m}$, $r \in \{0,1\}^{\ell}$ and $F \in \{\algo Z,H_r\}$.
\end{claim}

\begin{proof}
For every $i \in [n],j \in [m],\tau \in [\lceil \log(q) \rceil]$, $b \in \{0,1\}$, let $\bfM_{b}^{(i,j,\tau)}=(b \oplus r_{\phi(i,j,\tau)}) \cdot {\bf I}_{n \times n}$. Then from the quantum hardness of learning with errors, the following holds for every $(i,j,\tau)$ and $b \in \{0,1\}$: 
$$\{ \bfS_{b}^{(i,j,\tau)} \bfA + \bfE_{b}^{(i,j,\tau)}\} \approx_c \{ \bfS_{b}^{(i,j,\tau)} \bfA + \bfE_{b}^{(i,j,\tau)} + \bfM_b^{(i,j,\tau)} \}$$
Since $\{\bfS_{b}^{(i,j,\tau)}\}$ and $\{\bfE_{b}^{(i,j,\tau)}\}$ are sampled independently for every $(i,j,\tau)$ and $b\in \{0,1\}$, the proof of the claim follows. 
\end{proof}

\begin{remark}
When consider the all-zeroes function $\zerof$, we drop the notation from the parameters. For instance, we denote $pk_{\zerof}$ to be simply $pk$. 
\end{remark}

\paragraph{Construction.}

We consider the following parameters which are relevant to our $\prf$ construction. Let $n,m \in \N$ and let $q \in \N$ be a modulus with $q=2^{o(n)}$, and let $\ell = n m \lceil \log q \rceil$.  Let $\sigma$ be a parameter with $\sqrt{8m}< \sigma <q/\sqrt{8m}$ and let $p \ll q$ be a sufficiently large rounding parameter with
$$n\cdot m^{3}\sigma^2 \lceil \log q \rceil = (q/p) \cdot 2^{-o(n)}.$$

We now construct a key-revocable Dual-Regev (leveled) fully homomorphic encryption scheme based on our key-revocable Dual-Regev public-key scheme.
\begin{remark}
While the construction in this section can be readily adapted to feature \emph{classical revocation} via \Cref{cons:dual-regev-classical-revoc}, we choose to focus on \emph{quantum revocation} for simplicity by making use of our Dual-Regev scheme from \Cref{cons:dual-regev}.
\end{remark}

We describe our construction below.

\begin{construction}[Revocable $\prf$ scheme]\label{cons:prf}
Let $n \in \N$ be the security parameter and $m \in \N$. Let $q \geq 2$ be a prime and let $\sigma >0$ be a parameter. Let
$(\lemkg,\lemeval)$ be the procedure in \Cref{const:shift-hiding-const}. 
Our revocable $\prf$ scheme is defined as follows: 
\begin{itemize}

    \item $\gen(1^{\secparam})$: This is the following key generation procedure:
    \begin{enumerate}

    \item Sample $(\vec A,\mathsf{td}_{\vec A}) \leftarrow \GenTrap(1^n,1^m,q)$.


    \item Compute $\lemkey_{\zerof} \leftarrow \lemkg(1^n,1^m,q,\bfA,\zerof)$, where $\zerof:\{0,1\}^{\ell} \rightarrow \Zq^{n \times m}$ is the such that $\zerof(x)$ outputs an all zero matrix for every $x \in \{0,1\}^{\ell}$. Parse $\lemkey_{\zerof}$ as $(pk,sk)$.
    
    \item Generate a Gaussian superposition $(\ket{\psi_{\vec y}}, \vec y \in \Z_q^n) \leftarrow \mathsf{GenGauss}(\vec A,\sigma)$ with
$$
\ket{\psi_{\vec y}} \,\,= \sum_{\substack{\vec x \in \Z_q^m\\ \bfA \vec x = \vec y}}\rho_{\sigma}(\vec x)\,\ket{\vec x}.
$$
    \end{enumerate}
    Output $k=(pk,sk,\bfy)$, $\rho_{k} =  (pk,\ket{\psi_{\vec y}})$ and $\msk = \mathsf{td}_{\vec A}$. 

    \item$\prf(k,x)$: this is the following procedure:
    \begin{enumerate}
        \item Parse the key $k$ as a tuple $(pk,sk),\vec y)$.
        \item Output $\lfloor \bfS_x \bfy \rceil_p$. Here, $\bfS_x \in \Zq^{n \times n}$ is a matrix that can be efficiently recovered from $sk$ as stated in~\Cref{claim:kg-e-correctness}. 
    \end{enumerate}
    
    \item $\eval(\rho_k,x)$: this is the following evaluation algorithm:
    \begin{enumerate}
    \item Parse $\rho_k$ as $(pk,\rho)$.
    \item Compute ${\sf M}_x \leftarrow \lemeval(pk,x)$.
   
    \item Measure the register ${\sf Aux}$ of the state $U(\rho \otimes \proj{0}_{\sf Aux}) U^{\dagger}$. Denote the resulting outcome to be ${\bf z}$, where $U$ is defined as follows:
    $$U \ket{\bft}\ket{0}_{\sf Aux} \rightarrow \ket{\bft}\ket{\lfloor {\sf M}_{x} \cdot \bft \rceil_p}_{\sf Aux} $$

    \item Output ${\bf z}$. 
    \end{enumerate}

    \item $\revoke(\msk,\rho) $: given as input the trapdoor $\mathsf{td}_{\vec A} \leftarrow \msk$, apply the projective measurement $\{\ketbra{\psi_{\vec y}}{\psi_{\vec y}},I  - \ketbra{\psi_{\vec y}}{\psi_{\vec y}}\}$ onto the state $\rho$ using the procedure $\mathsf{QSampGauss}(\vec A,\mathsf{td}_{\vec A},\vec y,\sigma)$ in Algorithm \ref{alg:SampGauss}. Output $\valid$ if the measurement is successful, and $\invalid$ otherwise.
\end{itemize}

\end{construction}

\begin{lemma}
The above scheme satisfies correctness for our choice of parameters. 
\end{lemma} 

\begin{proof}
The correctness of revocation follows immediately from the correctness of $\mathsf{QSampGauss}$ in Algorithm \ref{alg:SampGauss}, which we showed in \Cref{lem:qdgs}.
Next, we show the correctness of evaluation. Let $\lemkey_{\zerof} \leftarrow \lemkg(1^n,1^m,q,\bfA,\zerof)$ with $\lemkey_{\zerof}=(\pk,\mathsf{SK})$.
From~\Cref{claim:kg-e-correctness}, we have for any $x \in \bit^\ell$:
$$\lemeval(\pk,x) = \bfS_x \bfA + \bfE_x \Mod{q},$$ where $\bfS_x \in \Zq^{n \times n}$ and $\bfE_x \in \Zq^{n \times m}$ with $||\bfE_x||_{\infty} \leq  nm^2 \sigma \lceil \log(q) \rceil$. Recall that $\mathsf{GenGauss}(\vec A,\sigma)$ outputs a state $\ket{\psi_{\vec y}}$ that is overwhelmingly supported on vectors $\vec t \in \Z_q^m$ such that $\|\vec t \| \leq \sigma \sqrt{\frac{m}{2}}$ with $\bfA \cdot \bft = \bfy \Mod{q}$. Therefore, we have for any input $x \in \bit^\ell$:
$$
\left\lfloor 
\lemeval(\pk,x)
\cdot \vec t
\right\rceil_p 
= \left\lfloor 
\bfS_x \bfA \cdot \vec t+ \bfE_x \cdot \vec t
\right\rceil_p  = \left\lfloor 
\bfS_x \cdot \vec y+ \bfE_x \cdot \vec t
\right\rceil_p = \left\lfloor 
\bfS_x \cdot \vec y
\right\rceil_p,
$$
where the last equality follows from the fact that 
$$\| \vec E_x \cdot \vec t \|_2 \leq
\| \vec E_x  \|_2 \cdot \|\vec t \|_2 \leq
\sqrt{m} \cdot \| \vec E_x  \| \cdot \|\vec t \|_2
\leq n \sqrt{m}m^2 \sigma \lceil \log(q) \rceil \cdot \sigma \sqrt{m/2}. 
$$
and $n\cdot m^{3}\sigma^2 \lceil \log q \rceil = (q/p) \cdot 2^{-o(n)}$ for our choice of parameters.
\end{proof}

\paragraph{Proof of security.}

Our first result on the security of \Cref{cons:prf} concerns $(\negl(\lambda),\negl(\lambda),1)$-security, i.e., we assume that revocation succeeds with overwhelming probability.

\begin{theorem}\label{thm:security-PRF}
Let $n\in \N$ and $q$ be a prime modulus with $q=2^{o(n)}$ and $m \geq 2n \log q$, each parameterized by $\lambda \in \N$. Let $\ell = n m \lceil \log q \rceil$. Let $\sqrt{8m}< \sigma <q/\sqrt{8m}$ and $\alpha \in (0,1)$ be any noise ratio with $1/\alpha= \sigma \cdot 2^{o(n)}$, and let $p \ll q$ be a sufficiently large rounding parameter with
$$
n\cdot m^{3}\sigma^2 \lceil \log q \rceil = (q/p) \cdot 2^{-o(n)}.
$$
Then, assuming the quantum subexponential hardness of  $\mathsf{LWE}_{n,q,\alpha q}^m$
and $\mathsf{SIS}_{n,q,\sigma\sqrt{2m}}^m$, our revocable $\PRF$ scheme
$(\gen,\prf,\eval,\revoke)$ defined in \Cref{cons:prf} satisfies $(\negl(\lambda),\negl(\lambda),1)$-revocation security according to \Cref{def:revocable-PRF}.
\end{theorem}

\begin{proof}
Let $\algo A$ be a $\QPT$ adversary and suppose that
$$ \vline \, \Pr\left[ 1 \leftarrow \expt^{\adversary}(1^{\secparam},0)\ \right] - \Pr\left[ 1 \leftarrow \expt^{\adversary}(1^{\secparam},1)\ \right] \, \vline = \epsilon(\secparam),$$
for some function $\eps(\lambda)$ with respect to security experiment $\expt^{\adversary}(1^{\secparam},b)$ from \Cref{fig:prf:expt}. To complete the proof, it suffices to show that $\eps(\lambda)$ is negligible.

\begin{figure}
   \begin{center} 
   \begin{tabular}{|p{13cm}|}
    \hline 
\begin{center}
\underline{$\expt^{\algo A}(1^{\secparam},b)$}: 
\end{center}
\noindent {\bf Initialization Phase}:
\begin{itemize}
    \item The challenger runs the procedure $\mathsf{Gen}(1^\lambda)$:
    \begin{enumerate}

    \item Sample $(\vec A,\mathsf{td}_{\vec A}) \leftarrow \GenTrap(1^n,1^m,q)$.

    \item Generate $\bfA_N \in \Zq^{(n+m) \times m}$ with $\overline{\bfA_{N}} \xleftarrow{\$} \Zq^{m \times m}$ and $\underline{\bfA_{N}} = \bfA$. 

    \item Compute $\lemkey_{\zerof} \leftarrow \lemkg(1^n,1^m,1^q,\bfA_N,\zerof)$, where $\lemkg$ is as defined in~\Cref{const:shift-hiding-const} and $\zerof:\{0,1\}^{\ell} \rightarrow \Zq^{n \times m}$ is such that $\zerof(x)$ outputs an all zero matrix for every $x \in \{0,1\}^{\ell}$. Parse $\lemkey_{\zerof}$ as $(pk,sk)$.
    
    \item Generate $(\ket{\psi_{\vec y}}, \vec y \in \Z_q^n) \leftarrow \mathsf{GenGauss}(\vec A,\sigma)$ with
$$
\ket{\psi_{\vec y}} \,\,= \sum_{\substack{\vec x \in \Z_q^m\\ \bfA \vec x = \vec y \Mod{q}}}\rho_{\sigma}(\vec x)\,\ket{\vec x}.
$$
\item Let $k=(pk,sk,\bfy)$, $\rho_{k} =  (pk,\ket{\psi_{\vec y}})$ and $\msk = \mathsf{td}_{\vec A}$.
    \end{enumerate}
\item The challenger sends $\rho_{k} =  (pk,\ket{\psi_{\vec y}})$ to $\algo A$.
\end{itemize}
\noindent {\bf Revocation Phase}:
\begin{itemize}
    \item The challenger sends the message \texttt{REVOKE} to $\adversary$. 
   \item $\mathcal A$ generates a (possibly entangled) bipartite quantum state $\rho_{R,\vec{\textsc{aux}}}$ in systems $\algo H_R \otimes \algo H_{\textsc{Aux}}$ with $\algo H_R= \algo H_q^m$, returns system $R$ and holds onto the auxiliary system $\textsc{Aux}$.
\item The challenger runs $\revoke(\msk,\rho_R)$, where $\rho_R$ is the reduced state in system $R$. If the outcome is $\invalid$, the challenger aborts.
\end{itemize}
\noindent {\bf Guessing Phase}:
\begin{itemize}
    \item  The challenger samples $x \leftarrow \bit^\ell$ and sends $(x,y)$ to $\adversary$, where
    \begin{itemize}
        \item If $b=0$: compute $\bfS_x$ from $sk$ as in~\Cref{claim:kg-e-correctness}. 
        Set $y=\lfloor \bfS_x \bfy \rceil_p$.
        \item If $b=1$: sample $y \leftarrow \{0,1\}^{n}$.
    \end{itemize}
    \item $\adversary$ outputs a string $b'$ and wins if $b'=b$. 
\end{itemize}
\ \\
\hline
\end{tabular}
    \caption{The revocable $\prf$ experiment $\expt^{\algo A}(1^\lambda,b)$ for \Cref{cons:prf}.}
    \label{fig:prf:expt}
    \end{center}
\end{figure}
Suppose for the sake of contradition that $\epsilon(\lambda)= 1/\poly(\lambda)$. 
Let us now introduce a sequence of hybrid experiments which will be relevant for the remainder of the proof.
\par Let $\mathsf{RevDual}=(\keygen,\enc,\dec,\revoke)$ be the $n$-bit key-revocable Dual-Regev scheme from \Cref{cons:dual-regev}. Fix $\mu=0^n$, where $\mu$ is the challenge message in the dual-Regev encryption security.\\ 

\noindent $\hybrid_0$: This is $\expt^{\adversary}(1^{\secparam},0)$ in \Cref{fig:prf:expt}.\\

\noindent $\hybrid_1$: This is the same experiment as $\expt^{\algo A}(1^{\secparam},0)$, except for the following changes:
\begin{itemize}
\item Sample a random string $r \leftarrow  \{0,1\}^{\ell}$. 
\item Run the procedure $\mathsf{RevDual}.\keygen(1^{\secparam})$ instead of $\GenTrap(1^n,1^m,q)$ and $\mathsf{GenGauss}(\vec A,\sigma)$ to obtain $(\vec A\in \Zq^{n \times m},\vec y \in \Z_q^n,\msk,\sk)$.
\item Compute $(\ct_1,\ct_2) \leftarrow \mathsf{RevDual}.\enc(\vec A,\vec y,\mu)$, where $\ct_1 \in \Zq^{n \times m}$ and $\ct_2 \in \Z_q^n$. 
\item Set $x = r \oplus \bindecomp(\ct_1)$.
\end{itemize}
\noindent The rest of the hybrid is the same as before. 
\par Note that Hybrids $\hybrid_0$ and $\hybrid_1$ are identically distributed. \\

\noindent $\hybrid_2$: This is the same experiment as before, except that the challenger now uses an alternative key-generation algorithm:
\begin{itemize}
    \item As before, run the procedure $\mathsf{RevDual}.\keygen(1^{\secparam})$ instead of $\GenTrap(1^n,1^m,q)$ and $\mathsf{GenGauss}(\vec A,\sigma)$ to obtain $(\vec A\in \Zq^{n \times m},\vec y \in \Z_q^n,\msk,\sk)$. Sample $r \leftarrow  \{0,1\}^{\ell}$.
    \item Let $H_r:\{0,1\}^{\ell} \rightarrow \Zq^{n \times m}$ be as defined in the beginning of~\Cref{sec:intlemma}. 
    \item Run the alternate algorithm $\lemkey_{H} \leftarrow  \lemkg(1^n,1^m,1^q,\bfA,H_r)$ instead of $\lemkey_{\zerof} \leftarrow  \lemkg(1^n,1^m,1^q,\bfA,\zerof)$. 
    \item Compute the ciphertext $(\ct^*_1,\ct^*_2) \leftarrow \mathsf{RevDual}.\enc(\vec A,\vec y,\mu)$, where $\ct^*_1 \in \Zq^{n \times m}$. Then, set $x^* =  r \oplus \bindecomp(\ct_1^*)$. Send $x^*$ to the adversary in the guessing phase. 
\end{itemize}
\ \\
\noindent $\hybrid_3$: This is the same hybrid as before, except that we choose $\ct_1^* \rand \Zq^{n \times m}$ and $\ct_2^* \rand \Zq^n$.\\
\ \\
$\hybrid_4$: This is the $\expt^{\adversary}(1^{\secparam},1)$ in \Cref{fig:prf:expt}.
\par Note that hybrids $\hybrid_3$ and $\hybrid_4$ are identically distributed.
\ \\
\ \\
We now show the following.

\begin{claim}\label{claim:shift-hiding}
By the shift-hiding property of $(\lemkg,\lemeval)$ in \Cref{claim:kg-e-shift-hiding}, we have that the two hybrids $\hybrid_1$ and $\hybrid_2$ are computationally indistinguishable, i.e.
$$
\hybrid_1 \approx_c \hybrid_2.
$$
\end{claim}
\begin{proof}
Suppose for the sake of contradiction that there exists a non-negligble difference in the advantage of the adversary $\algo A$ in the two hybrids $\hybrid_1$ and $\hybrid_2$.

We now design a reduction ${\cal B}$ that violates the shift-hiding property as follows. 
\begin{enumerate}
    \item Sample $r \xleftarrow{\$} \{0,1\}^{\ell}$. Send $(\zerof,H_r)$ to the challenger. 
    \item The challenger responds with $pk=\left(\bfA,\left\{ \ct_{b}^{(i,j,\tau)} \right\}_{\substack{i \in [n],j \in [m]\\ \tau \in [\lceil \log(q) \rceil],b \in \{0,1\}}} \right)$
    \item Compute $(\ket{\psi_{\vec y}}, \vec y \in \Z_q^n) \leftarrow \mathsf{GenGauss}(\vec A,\sigma)$ from the challenger.
    \item Set $\rho_k=(pk,\rho)$. 
    \item Compute $(\ct_1,\ct_2) \leftarrow \mathsf{RevDual}.\enc(\vec A,\vec y,\mu)$, where $\ct_1 \in \Zq^{n \times m}$ and $\ct_2 \in \Z_q^n$. Then, set $x^* = r \oplus \bindecomp(\ct_1)$.

    \item Compute $\eval(\rho_k,x^*)$ to obtain $y^*$ while recovering $\rho^*_k$ (using the "Almost as Good As New Lemma", \Cref{lem:almost}) such that ${\sf TD}(\rho^*_k,\rho_k) \leq \negl(\secparam)$.  
    
    \item Send $\rho_k^*$ to $\adversary$. 
    \item $\adversary$ computes a state on two registers $R$ and $\textsc{aux}$. It returns the state on the register $R$. 
    \item $\adversary$, on input the register $\textsc{aux}$ and $(x^*,y^*)$, outputs a bit $b'$. 
    \item Output $b'$. 
\end{enumerate}
\par If $pk$ is generated using $ \lemkg(1^n,1^m,q,\bfA,{\cal Z})$ then we are precisely in $\hybrid_1$ (except that $\revoke$ is not performed). Moreover, if $pk$ is generated using $ \lemkg(1^n,1^m,q,\bfA,H_r)$ then we are in the hybrid $\hybrid_2$ (except that $\revoke$ is not performed). Therefore $\algo B$ has a non-negligible advantage at distinguishing $\hybrid_1$ and $\hybrid_2$ whenever $\revoke$ outputs $\top$ on system \textsc{R}.
Using \Cref{lem:remove-revoke}, this implies that we can break the shift-hiding property with non-negligible advantage. This proves the claim.
\end{proof}

Next, we invoke the security of the $n$-bit variant of our key-revocable Dual-Regev scheme (which is implied by~\Cref{thm:security-Dual-Regev-high-revoke}) to show the following.

\begin{claim}\label{claim:DR} 
By the security of our $n$-bit key-revocable Dual-Regev encryption scheme, we have that the two hybrids $\hybrid_2$ and $\hybrid_3$ are computationally indistinguishable, i.e.
$$
\hybrid_2 \approx_c \hybrid_3.
$$
\end{claim}

\begin{proof} 
Suppose for the sake of contradiction that there exists a non-negligble difference in the advantage of $\algo A$ in the two hybrids $\hybrid_2$ and $\hybrid_3$. Using $\adversary$, we can now design a reduction ${\cal B}$ that violates the revocation security of our $n$-bit revocable Dual-Regev scheme which is implicit in~\Cref{thm:security-Dual-Regev-high-revoke}. 
\par The reduction ${\cal B}$ proceeds as follows.
\begin{enumerate}
    \item First, it receives as input $\vec A,\vec y$ and a quantum state 
$$
    \ket{\psi_{\vec y}} = \sum_{\substack{\vec x \in \Z_q^{m}\\ \bfA \vec x = \vec y}} \rho_\sigma(\vec x) \ket{\vec x}.
$$   
\item The reduction generates a quantum state $\rho_k$ as follows: 
\begin{itemize}
    \item Sample a random string $ r \rand  \bit^\ell$.
    \item Let $H_r:\{0,1\}^{\ell} \rightarrow \Zq^{n \times m}$ be as defined in the beginning of~\Cref{sec:intlemma}.

  \item Run the algorithm $\lemkey_{H} \leftarrow  \lemkg(1^n,1^m,1^q,\bfA,H_r)$ and parse $\lemkey_H$ as $(pk,sk)$.
  \item Set $\rho_k=(pk,\ket{\psi_{\bfy}})$. 
\end{itemize}
Send $\rho_k$ to $\adversary$. 
\item $\adversary$ outputs a state on two registers $R$ and $\textsc{aux}$. The register $R$ is returned. The reduction forwards the register $R$ to the challenger. 
\item The reduction then gets the challenge ciphertext $\mathsf{CT} = [\ct_1, \ct_2]^\intercal \in \Z_q^{n\times m} \times \Z_q^n$. The reduction then sets
$$
x^* = r \oplus \bindecomp(\ct_1)
$$
and sends $x^*$ to $\algo A$ in the guessing phase, together with $y=\lfloor  \bfS_{x^*} \bfy + \ct_2 \rceil_p$
which is computed using the secret key $\mathsf{SK}$ (c.f.~\Cref{claim:kg-e-correctness}).
\item $\algo A$ outputs a bit $b'$. ${\cal B}$ outputs $b'$.
\end{enumerate}
There are two cases to consider here. In the first case, we have $\mathsf{CT} = [\ct_1, \ct_2]^\intercal \in \Z_q^{n\times m} \times \Z_q^n$ is a Dual-Regev ciphertext. Here, $y=\lfloor  \bfS_{x^*} \bfy + \ct_2 \rceil_p$ precisely corresponds to the output of the pseudorandom function on $\rho_k$ and $x$. In the second case, we have $\mathsf{CT} = [\ct_1, \ct_2]^\intercal$, where $\ct_1 \xleftarrow{\$} \Zq^{n \times m}$ and $\ct_2 \xleftarrow{\$} \Zq^m$. Therefore, the resulting string $y=\lfloor  \bfS_{x^*} \bfy + \ct_2 \rceil_p$ is negligibly close (in total variation distance) to a uniform distribution on $\Z_p^m$. 
\par Putting everything together, we find that the first case corresponds precisely to $\hybrid_2$, whereas the second case corresponds to $\hybrid_3$. As a result, ${\cal B}$ violates the revocation security of our $n$-bit revocable Dual-Regev scheme which is implicit in~\Cref{thm:security-Dual-Regev-high-revoke}.This completes the proof. 
\end{proof}

\noindent Putting everything together, we have shown that
$$  \vline \, \Pr\left[ 1 \leftarrow \expt^{\adversary}(1^{\secparam},0)\ \right] - \Pr\left[ 1 \leftarrow \expt^{\adversary}(1^{\secparam},1)\ \right] \, \vline \leq \negl(\secparam).$$
\end{proof} 

Finally, we prove that \Cref{cons:prf} achieves the stronger notion of $(\negl(\lambda),1-1/\poly(\lambda),1)$-security assuming \Cref{thm:search-to-decision}.

\begin{theorem}\label{thm:security-PRF-conj}
Let $n\in \N$ and $q$ be a prime modulus with $q=2^{o(n)}$ and $m \geq 2n \log q$, each parameterized by $\lambda \in \N$. Let $\ell = n m \lceil \log q \rceil$. Let $\sqrt{8m}< \sigma <q/\sqrt{8m}$ and $\alpha \in (0,1)$ be any noise ratio with $1/\alpha= \sigma \cdot 2^{o(n)}$, and let $p \ll q$ be a sufficiently large rounding parameter with
$$
n\cdot m^{3}\sigma^2 \lceil \log q \rceil = (q/p) \cdot 2^{-o(n)}.
$$
Then, assuming \Cref{thm:search-to-decision}, our revocable $\PRF$ scheme
$(\gen,\prf,\eval,\revoke)$ defined in \Cref{cons:prf} has $(\negl(\lambda),1-1/\poly(\lambda),1)$-revocation security according to \Cref{def:revocable-PRF}.
\end{theorem}
\begin{proof}
    The proof is the same as in \Cref{thm:security-PRF}, except that we invoke \Cref{thm:security-Dual-Regev} instead of \Cref{thm:security-Dual-Regev-high-revoke} for Dual-Regev security.
\end{proof}

\printbibliography

\newpage 
\appendix

\end{document}